\documentclass[11pt]{article}

\title{\textbf{A characterization of graph properties testable for \\ \emph{general planar graphs} with one-sided error \\ (It's all about forbidden subgraphs)
}}
\author{\textbf{Artur Czumaj}
        \thanks{Department of Computer Science and Centre for Discrete Mathematics and its Applications (DIMAP), University of Warwick. Email: A.Czumaj@warwick.ac.uk.  Research partially supported by the Centre for Discrete Mathematics and its Applications (DIMAP), by IBM Faculty Award, and by EPSRC award EP/N011163/1.}
    \and
\textbf{Christian Sohler}
        \thanks{Department of Computer Science, TU Dortmund. Email: christian.sohler@tu-dortmund.de. Research supported by ERC grant No. 307696.}
}

\date{}

\usepackage{amsthm,amsmath,amssymb,amsfonts}
\usepackage{graphicx}
\usepackage{subfigure}
\usepackage{latexsym}
    \usepackage{times} %replaced by three commands below
        \usepackage[scaled=.90]{helvet}
        \usepackage{courier}
\usepackage{fancybox}
\usepackage{fullpage}
\usepackage{color}
\usepackage{paralist}
\usepackage{environ}
\usepackage{xstring}
\usepackage{thmtools}

\usepackage[T1]{fontenc}\usepackage{yfonts}

\usepackage{ifpdf}
\newcommand{\mydriver}{hypertex}
\ifpdf
 \renewcommand{\mydriver}{pdftex}
\fi
\usepackage[breaklinks,\mydriver]{hyperref}

\newcount\shortyear\newcount\shorthour\newcount\shortminute
\shorthour=\time\divide\shorthour by 60\shortyear=\shorthour
\multiply\shortyear by 60\shortminute=\time\advance\shortminute by
-\shortyear \shortyear=\year\advance\shortyear by -1900

\def\zeit{\number\shorthour:\ifnum\shortminute<10 0\number\shortminute
\else\number\shortminute\fi}

\let\oldproofname=\proofname
\renewcommand{\proofname}{\rm\bf{\oldproofname}}

\newtheorem{theorem}{Theorem}

\newtheorem{lemma}[theorem]{Lemma}

\newtheorem{defn}[theorem]{Definition}
\newenvironment{definition}[1][]{\IfStrEq{#1}{}{\begin{defn}}{\begin{defn}[#1]}\rm}{\end{defn}}
    \newtheorem{remk}[theorem]{Remark}
    \newenvironment{remark}[1][]{\IfStrEq{#1}{}{\begin{remk}}{\begin{remk}[#1]}\small\rm}{\qed\end{remk}}
\newtheorem{claim}[theorem]{Claim}

\newtheorem{fact}[theorem]{Fact}

\renewcommand{\qed}{\hspace*{\fill}\ensuremath{\blacksquare}}
\newcommand{\COMMENTED}[1]{{}}
\newcommand{\junk}[1]{\COMMENTED{#1}}

\newcommand{\NAT}{\ensuremath{\mathbb{N}}}
\newcommand{\NATURAL}{\NAT}
\renewcommand{\Pr}[1]{\ensuremath{\mathbf{Pr}[#1]}}

\newcommand{\Ex}[1]{\ensuremath{\mathbf{E}[#1]}}

\def\epsilon{\ensuremath{\varepsilon}}
\newcommand{\eps}{\ensuremath{\epsilon}}

\newcommand{\dg}{\ensuremath{\mathfrak{d}}}
\newcommand{\ld}{\ensuremath{\mathfrak{t}}}
\newcommand{\RBE}{{\bf\sffamily Random-Bounded-BFS-Exploration}}
\renewcommand{\RBE}{{\bf\sffamily Random-Exploration}}
\newcommand{\RLBD}{{\bf\sffamily Random-Bounded-BFS-Tester}}
\renewcommand{\RLBD}{{\bf\sffamily Tester}}
\newcommand{\RLBFS}{{\bf\sffamily Random-bounded-BFS}}
\renewcommand{\RLBFS}{{\bf\sffamily Random-Traverse}}
\newcommand{\Traverse}{{\bf\sffamily Bounded-BFS-Traverse}}
\newcommand{\HRLBD}{{\bf\sffamily HTester}}
\newcommand{\HRLBFS}{{\bf\sffamily Random-HTraverse}}

\newcommand{\AL}{\textbf{\sffamily Assigning-Levels}}

\newcommand{\DDCH}[1]{{\ensuremath{\mathcal{C}_{#1}^{\langle H \rangle}}}} % disjoint copies of H
\renewcommand{\DDCH}[1]{{\ensuremath{\mathbb{Q}_{#1}}}} % disjoint copies of H
\newcommand{\DCH}{\ensuremath{\mathcal{H}}} % disjoint copies of H
\renewcommand{\DCH}{{\ensuremath{\mathcal{C}^{\langle H \rangle}}}} % disjoint copies of H
\renewcommand{\DCH}{{\DDCH{}}} % disjoint copies of H
\newcommand{\G}{\ensuremath{\mathcal{G}}} % graph induced by copies
\renewcommand{\G}{\ensuremath{G}} % graph induced by copies
\newcommand{\GG}{\ensuremath{\mathfrak{H}}} % shadow graph
\renewcommand{\GG}{{\ensuremath{\mathfrak{G}}}} % shadow graph
\newcommand{\U}{\ensuremath{\mathcal{G}}}
\renewcommand{\U}{\ensuremath{\mathbb{G}}}
\newcommand{\M}{\ensuremath{\mathcal{M}}}
\renewcommand{\M}{\ensuremath{\mathcal{M}}}

\newcommand{\HQ}{\ensuremath{\mathcal{H}}}
\renewcommand{\HQ}{\ensuremath{\mathcal{H}}}
\renewcommand{\P}{\ensuremath{\mathcal{P}}}
\renewcommand{\P}{\ensuremath{\mathfrak{Rep}}}

\newcommand{\h}{\ensuremath{\mathfrak{h}}}
\newcommand{\e}{\ensuremath{\mathfrak{e}}}
\newcommand{\cc}{\ensuremath{\mathfrak{c}}}
\newcommand{\N}{\ensuremath{\mathcal{N}}}
\newcommand{\lab}{\ensuremath{\sigma}}
\newcommand{\clab}{\ensuremath{\sigma^*}}
\newcommand{\Hsize}{\ensuremath{|V(H)|}}

\newcommand{\Gi}{\ensuremath{G^*_i}}
\renewcommand{\Gi}{\ensuremath{\GG(\HQ_i(\DDCH{i}))}}
\newcommand{\Gii}{\ensuremath{G^*_i}}
\newcommand{\Vii}{\ensuremath{V^*_i}}
\newcommand{\Eii}{\ensuremath{E^*_i}}
\newcommand{\NEIG}[2]{\ensuremath{\mathcal{N}_{#1}(#2)}}

\newcommand{\undefine}[1]{\let#1\ThisCommandWasUndefined}

\newcommand{\Artur}[1]{\footnote{{\sc\small \textcolor[rgb]{1.00,0.00,0.00}{\textbf{Artur:}}} \textcolor[rgb]{1.00,0.00,0.00}{#1}}}
\newcommand{\SArtur}[1]{}
\newlength{\savedparindent}
\newcommand{\SaveIndent}{\setlength{\savedparindent}{\parindent}}
\newcommand{\RestoreIndent}{\setlength{\parindent}{\savedparindent}}

\newcommand{\InGray}[1]{%
\SaveIndent{} %
\noindent{} \fcolorbox[rgb]{0,0,0}{0.95,0.95,0.95}{
\begin{minipage}{0.965\linewidth} %
\RestoreIndent{}%
#1
\end{minipage}
} }

\newcommand{\InGrayMiddle}[1]{%
\SaveIndent{} %
\centerline{ \fcolorbox[rgb]{0,0,0}{0.95,0.95,0.95}{
\begin{minipage}{0.8\linewidth} %
\RestoreIndent{}%\
#1
\end{minipage}
} } }

\NewEnviron{algo}{\medskip\InGrayMiddle{\mbox{}\\\small\BODY\mbox{}\\[-0.35in]}\medskip}
\NewEnviron{walgo}{\medskip\InGray{\mbox{}\\\small\BODY\mbox{}\\[-0.35in]}\medskip}

\date{\small \zeit{}, \today}
\date{}

\begin{document}

%\maketitle

\begin{titlepage}
\maketitle\thispagestyle{empty}

\begin{abstract}
The problem of characterizing testable graph properties (properties that can be tested with a number of queries independent of the input size) is a fundamental problem in the area of property testing. While there has been some extensive prior research characterizing testable graph properties in the dense graphs model and we have good understanding of the bounded degree graphs model, no similar characterization has been known for \emph{general graphs}, with \emph{no degree bounds}. In this paper we take on this major challenge and consider the problem of characterizing all testable graph properties in \emph{general planar graphs}.

We consider the model in which a general planar graph can be accessed by the random neighbor oracle that allows access to any given vertex and access to a random neighbor of a given vertex. We show that, informally, a graph property $\mathcal{P}$ is testable with one-sided error for general planar graphs \emph{if and only if} testing $\mathcal{P}$ can be reduced to testing for a finite family of finite forbidden subgraphs.
\junk{
We consider the model in which a general planar graph can be accessed by the random neighbor oracle enabling two queries: to access a random vertex and to access a random neighbor of a given vertex\Artur{I'm not sure, but maybe we just want to have ``to access a vertex and to access a random neighbor of a given vertex.''}. A property tester is called oblivious if its decisions are independent of the size of the input graph\Artur{I'm not sure yet if we need the assumption that the tester must be oblivious. (Term ``oblivious'' comes from \cite{AS08})}. We show that a graph property $\mathcal{P}$ has an oblivious one-sided error tester for planar graphs \emph{if and only if} testing $\mathcal{P}$ can be reduced to testing for a finite family of forbidden subgraphs.

It is not difficult to see that any testable graph property with one-sided error can be reduced to testing for a finite family of forbidden subgraphs.\Artur{You may disagree with the claim ``it's not difficult to see'' but I believe now that it's rather obvious that the claim is true. In fact, we possibly don't even need our paper \cite{CFPS19} here, though we definitely should keep it here.} Our main technical contribution is in showing that the fundamental property of being $H$-free can be tested with a constant number of queries and with one-sided error for every fixed finite graph $H$. The proof of this result is by analysing some basic random graph exploration algorithm: starting at a random vertex, we perform a random constant-width, constant-depth breadth-first-search, to determine whether a forbidden graph $H$ is as a subgraph. The analysis of this algorithm is the main technical challenge and it involves a series of nontrivial extensions of earlier works on testing bipartiteness and of the analysis of random explorations of planar graphs.
}%
While our presentation focuses on planar graphs, our approach extends easily to general minor-free graphs.

%Our characterization is obtained by proving both sufficient and necessary conditions for testable graph properties with one-sided error in planar graphs.
Our analysis of the necessary condition
%that a graph property is testable with one-sided error for general planar graphs only if its testing can be reduced to testing for a finite family of finite forbidden subgraphs
relies on a recent construction of canonical testers in the random neighbor oracle model that is applied here to the one-sided error model for testing in planar graphs.
The sufficient condition in the characterization reduces the problem to the task of testing $H$-freeness in planar graphs, and is \emph{the main and most challenging technical contribution of the paper: we show that for planar graphs (with arbitrary degrees), the property of being $H$-free is testable with one-sided error for every finite graph $H$, in the random neighbor oracle model}.
\end{abstract}
\end{titlepage}

%---------------------------------------------------------------------------------------------------------------------------------------------------------

\section{Introduction}

\SArtur{Remember: testable in constant time or with constant query complexity = \emph{testable}. (Disclaimer: we don't want to use \emph{constant time}!)}
\SArtur{Emphasize: \emph{one-sided error}!}
The fundamental problem in the area of graph property testing is for a given undirected graph $G$ to distinguish if $G$ satisfies some graph property $\mathcal{P}$ or if $G$ is $\eps$-far from satisfying $\mathcal{P}$, where $G$ is said to be $\eps$-far from satisfying $\mathcal{P}$ if an $\eps$-fraction of its representation should be modified in order to make $G$ satisfy~$\mathcal{P}$. %Here, a graph property $\mathcal{P}$ is a family of graphs closed under isomorphism.
%\Artur{Christian, I know you like the story about network motifs, but I feel it's better to sell the paper in a very different way, as a characterization of one-sided error testable properties for planar graphs. And only somewhere, later, point out that people do care about testing subgraph freeness.}%
%
The notion of testability of combinatorial structures and of graphs, has been introduced by Goldreich et al.\ \cite{GGR98}, who have shown that many natural graph properties such as $k$-colorability or having a large clique are \emph{testable}, that is, have a tester, whose query complexity, that is, the number of oracle queries to the input representation (in \cite{GGR98}, to the graph adjacency matrix) can be upper bounded by a function that depends only on the property $\mathcal{P}$ and on $\eps$, the proximity parameter of the test, and is independent of the size of the input graph $G$. This has been later extended to show that testability in the \emph{dense graph model} (of \cite{GGR98}) is closely related to the graph regularity lemma as one can show that a property is testable (with two-sided error) if and only if it can be reduced to testing for a finite number of regular partitions \cite{AFNS06}; for one-sided error testing, it has been shown that a property is testable if and only if it is hereditary or close to hereditary \cite{AS08}. In particular, we know that subgraph freeness is testable with one-sided error in this model (see, e.g., \cite{AS05}). We also know of similar logical characterization of families of testable graph properties (for example, every first-order graph property of type $\exists\forall$ is testable, while there are first-order graph properties of type $\forall\exists$ that are not testable~\cite{AFKS00}).

While for many years the main efforts in property testing have been concentrated on the dense graph model, there has been also an increasing amount of research focusing on the \emph{bounded degree graph model} introduced by Goldreich and Ron \cite{GR97}, the model more suitable for sparse graphs. For example, while it is trivial to test the subgraph freeness with one-sided error in this model, testing $H$-minor freeness is more complex, and is possible with constant query complexity only if $H$ is cycle-free \cite{CGRSSS14}; if $H$ has a cycle, then $\widetilde{\Omega_{\eps}}(\sqrt{n})$ queries are required and effectively sufficient \cite{CGRSSS14,ERS17}. Among further highlights, it is known that every hyperfinite property is testable with two-sided error \cite{NS13} (see also earlier works in \cite{BSS10,CSS09,GR99}).

Rather surprisingly, much less is known for \emph{general graphs}, that is, graphs with no bound for the maximum degree (see, e.g., \cite[Chapter~10]{Gol17}). The model has been initially studied by Kaufman et al.\ \cite{KKR04}, Parnas and Ron \cite{PR02}, and Alon et al.\ \cite{AKKR08}, where the main goal was to study the trade-off between the complexity for sparse graphs with that for dense graphs (it should be noted though that these papers were using a slightly different access oracle to the input graph). These results show that most of even very basic properties are \emph{not testable}. Czumaj et al.\ \cite{CMOS11} addressed a related question in this model, and show that in fact if one restricts the input graphs to be planar (but without any constraints on the maximum degree), then the benchmark problem of testing bipartiteness is testable in the random neighbor query model. In a similar vein, Ito \cite{Ito16} extended the framework from \cite{NS13} and show that all graph properties are testable for a certain special class of multigraphs called hierarchical-scale-free multigraphs. Still, despite these few results and despite its natural importance, \emph{our understanding of graph property testing for degree-unconstrained graphs is very limited}.
%
%\Artur{Should we mention D\'{\i}az et al \cite{DJKO18}?}
%
In this paper we take on this major challenge and consider the problem of characterizing all testable graph properties in \emph{general planar graphs}. We consider the model in which a general planar graph can be accessed by the \emph{random neighbor oracle} that allows access to any given vertex and access to a random neighbor of a given vertex. We show that, informally, a graph property $\mathcal{P}$ is testable with one-sided error for general planar graphs \emph{if and only if} testing $\mathcal{P}$ can be reduced to testing for a finite family of finite forbidden subgraphs. While our presentation focuses on planar graphs, our approach extends easily to general minor-free graphs.

%---------------------------------------------------------------------------------------------------------------------------------------------------------

\paragraph{Testing for subgraphs-freeness.}

The central combinatorial problem considered in this paper is that of subgraph detection. The question of identifying frequent subgraphs in big graphs is one of the most fundamental problems in network analysis, extensively studied in the literature. It has been empirically shown that different classes of networks have the same frequent subgraphs and they differ for different network classes \cite{MSIKCA02}. In this context, frequently occurring subgraphs are also known as \emph{network motifs} \cite{MSIKCA02}. This raises the question how quickly we can identify the motifs of a given network. Recent work approaches this question by approximating the number of occurrences of certain subgraphs using random sampling \cite{ELRS17,ERS17,GRS11}. In this paper, we will study the corresponding property testing question: \emph{Can we distinguish a graph that has no copies of a predetermined subgraph $H$ from a graph in which we need to remove more than an $\eps$-fraction of its edges in order to obtain a graph that contains no copy of $H$}. This question has received a lot of attention in the property testing setting and it is known that subgraph freeness can be tested with a constant number of queries both in the dense graph model (see, e.g., \cite{AS05}) and in the bounded degree graph model, where testing subgraph freeness is simple. The problem of testing subgraph freeness has also been studied in the setting of general graphs \cite{AKKR08}, where the authors give a lower bound of $\Omega(n^{1/3})$ queries for testing triangle freeness\SArtur{I'm puzzled by this claim, since for general graphs, even for sparse ones, testing triangle freeness trivially requires $\Omega(\sqrt{n})$ complexity \dots take our standard example of a clique on $\sqrt{n}$ vertices and the rest consists of $n - \sqrt{n}$ isolated vertices.}, which can be extended to other non-bipartite subgraphs. They also give an upper bound of $O(n^{6/7})$ queries. We continue this line of research, but will put our focus on sparse graphs, i.e., graphs with bounded average degree. Since it seems that for many properties we cannot hope for extremely efficient, that is, testing algorithms with a constant number of queries in general graphs (often a hard example is a clique on $\sqrt{n}$ vertices), we focus our attention on \emph{planar graphs}. It has been only recently shown that bipartiteness in planar graphs can be tested with a constant number of queries  \cite{CMOS11}. %(as well as any property that is constant time testable in general graphs).
Our result can be viewed as a major extension of that result: we prove that for every fixed graph $H$, the property of $H$-freeness can be tested with a constant number of queries. Our approach extends to general \emph{minor-free graphs}.

%---------------------------------------------------------------------------------------------------------------------------------------------------------

\subsection{Basic notation}

Before we proceed with detailed description of our results, let us begin with some basic definitions.

%---------------------------------------------------------------------------------------------------------------------------------------------------------

\paragraph{Notation.}
\label{Notation}

Throughout the paper we use several constants depending on $H$ (forbidden subgraph) and $\eps$. We use lower case Greek letters to denote constants that are typically smaller than $1$ (e.g., $\delta_i(\eps,H)$) and lower case Latin letters to denote constants that are usually larger than $1$ (e.g., $f_i(\eps,H)$). All these constants are always positive.
Furthermore, throughout the paper we use the asymptotic symbols $O_{\eps,H}(\cdot)$, $\Omega_{\eps,H}(\cdot)$, and $\Theta_{\eps,H}(\cdot)$, which ignore multiplicative factors that depend only on $H$ and $\eps$ and that are positive for $\eps>0$.

Throughout the paper, for any set of edge-disjoint subgraphs $\mathcal{S}$ of $G = (V,E)$, we write $\G[\mathcal{S}]$ to denote the graph with vertex set $V$ and edge set being the set of edges from the sets in $\mathcal{S}$.

%---------------------------------------------------------------------------------------------------------------------------------------------------------

\subsubsection{Property testing and $H$-freeness}

A \emph{graph property $\mathcal{P}$} is any family of graphs closed under isomorphism. (For example, bipartiteness is a graph property $\mathcal{P}$ defined by a family of all bipartite graphs.) We are interested in finding an algorithm (called \emph{tester}) for testing a given graph property $\mathcal{P}$, i.e., an algorithm that inspects only a very small part of the input graph $G$, and accepts if $G$ satisfies $\mathcal{P}$ with probability at least $\frac23$, and rejects $G$ if it is $\eps$-far away from $\mathcal{P}$ with probability at least $\frac23$, where $\eps$ is a proximity parameter, $0 \le \eps \le 1$. We say a simple graph $G$ is \emph{$\eps$-far from $\mathcal{P}$} if one has to delete or insert more than $\eps |V|$ edges from $G$ to obtain a graph satisfying $\mathcal{P}$\footnote{Similarly as in \cite{CMOS11}, we notice that the standard definition of being $\eps$-far (see, e.g., the definition in \cite{Gol17} or \cite{KKR04}) expresses the distance as the fraction of edges that must be modified in $G$ to obtain a graph satisfying $\mathcal{P}$; comparing this to our definition, instead of modifying $\eps |V|$ edges, one modifies $\eps |E|$ edges. In this paper we prefer to use the definition with $\eps |V|$ edge modifications because our focus is on the study of sparse graphs, graphs with $|E| = O(|V|)$. Indeed, for any class of planar graphs or graphs with an excluded minor, which are the main classes of graphs studied in this paper, the number of edges in the graph is upper bounded by $O(|V|)$. Moreover, unless the graph is very sparse (i.e., most of its vertices are isolated, in which case even finding a single edge in the graph may take a large amount of time), the number of edges in the graph is $\Omega(|V|)$. Thus, under the standard assumption that $|E| = \Omega(|V|)$, the $\eps$ in our definition and the $\eps$ in the previous definitions remain within a constant factor. We use our definition of being $\eps$-far for simplicity; our analysis can be extended to the standard definition in a straightforward way.}.
%
%While a large part of our analysis focuses on general planar graphs and on general minor-free graphs, and so the input graph $G$ will be either a general planar graph or a general minor-free graph, some part of our study can be applied to general graphs, not necessarily even sparse graphs.

The main focus of this paper is on the study of testers with \emph{one-sided error}, that is, testers that always accept all graphs satisfying $\mathcal{P}$ and can err only for graphs $\eps$-far from $\mathcal{P}$. (In contrast, \emph{two-sided error} testers can err (with probability at most $\frac13$) both for graphs $\eps$-far from $\mathcal{P}$ \emph{and} for graphs satisfying $\mathcal{P}$.)

The main graph properties considered in this paper are related to \emph{forbidden subgraphs}. Throughout the entire paper we will fix $H = (V(E),E(H))$ to be an arbitrary, simple, finite undirected graph. The notion of a graph $H$ being \emph{finite} means that its size is constant, though we will allow the constant to be a function of $\eps$, the proximity parameter for property testing, which will be clear from the context. (That is, for a given graph property $\mathcal{P}$ and a proximity parameter $\eps$, $0 < \eps < 1$, we will say that a graph $H$ is \emph{finite} (for $\mathcal{P}$ and $\eps$) if there is $s  = s(\eps) = O_{\eps}(1)$, such that $|V(H)| \le s$ for every $n \in \NATURAL$.)
\junk{\Artur{You may want to have a more elaborate definition, but I don't think it's needed:
\it
For a given graph property $\mathcal{P} = (\mathcal{P}_n)_{n \in \NATURAL}$, a proximity parameter $\eps$, $0 < \eps < 1$, and $n \in \NATURAL$, we will say that a graph $H$ is \emph{finite} (for $\mathcal{P}$, $\eps$, and $n$) if there is $s  = s(\eps) = O_{\eps}(1)$, such that $|V(H)| \le s$.
For given $\mathcal{P}$, $\eps$, and $n$, a \emph{family $\mathcal{H}$ of finite graphs} is any family such that each $H \in \mathcal{H}$ is finite (for $\mathcal{P}$, $\eps$, and $n$).
Furthermore, notice that since each $H \in \mathcal{H}$ is of size $O_{\eps}(1)$, so is the size of $\mathcal{H}$, that is, $|\mathcal{H}| = O_{\eps}(1)$. Hence, $\mathcal{H}$ is also a \emph{finite family} of \emph{finite graphs}.
}}

We say that a given graph $G$ is \emph{$H$-free} if $G$ does not contain a copy of $H$.\SArtur{Is this clear? Should we emphasize that we're not talking about induced copies?} Following the definition above, we say that a simple graph $G$ is \emph{$\eps$-far from $H$-free} if one has to delete more than $\eps |V|$ edges from $G$ to obtain an $H$-free graph.

Our definitions extends to families of forbidden graphs in a natural way. If $\mathcal{H}$ is a finite family of finite graphs, then a given graph $G$ is \emph{$\mathcal{H}$-free} if for every $H \in \mathcal{H}$, $G$ is $H$-free. Similarly, $G$ is \emph{$\eps$-far from $\mathcal{H}$-free} if for every $H \in \mathcal{H}$, $G$ is $\eps$-far from $H$-free. Further, notice that if $\mathcal{H}$ is a finite family of finite graphs then since each $H \in \mathcal{H}$ is of size $O_{\eps}(1)$, so is the size of $\mathcal{H}$; hence, $\mathcal{H}$ is also a \emph{finite family} of \emph{finite graphs}.

In our paper we will also consider the following generalization of $\mathcal{H}$-freeness. In what follows, for a given graph property $\mathcal{P}$ and $n \in \NATURAL$, let $\mathcal{P}_n$ be the graph property $\mathcal{P}$ for $n$-vertex graphs.

\begin{definition}\textbf{(Semi-subgraph-freeness)}\it
\label{def:semi-H-freeness}
\SArtur{This definition should be independent of any complexity assumptions, independent on the query access model, etc.}%
A graph property $\mathcal{P} = (\mathcal{P}_n)_{n \in \NATURAL}$ is \textbf{\emph{semi-subgraph-free}} if for every $\eps$, $0 < \eps < 1$, and every $n \in \NATURAL$, there is a finite family $\mathcal{H}$ of finite graphs such that the following hold:
\begin{enumerate}[(i)]
\item any graph $G$ satisfying $\mathcal{P}_n$ is $\mathcal{H}$-free, and
\item any graph $G$ which is $\eps$-far from satisfying $\mathcal{P}_n$, is not $\mathcal{H}$-free (contains a copy of some $H \in \mathcal{H}$).
\end{enumerate}
\end{definition}

%(In short, $\mathcal{P}$ is semi-subgraph-free if $\mathcal{P}$ is $\mathcal{H}$-free or close to $\mathcal{H}$-free, for some finite family $\mathcal{H}$ of finite graphs.\Artur{If somebody is very precise, then this might be tricky, since in fact $\mathcal{H}$ may be different for different $\eps$ and $n$.})

Let us emphasize that in Definition \ref{def:semi-H-freeness} by a \emph{finite} family $\mathcal{H}$ of \emph{finite graphs} we mean that even though $\mathcal{H}$ may depend on $n$, the sizes of $\mathcal{H}$  and of any $H \in \mathcal{H}$ are always upper bounded by a function independent on $n$, $|\mathcal{H}| = O_{\eps}(1)$ and $|V(H)| = O_{\eps}(1)$.

%---------------------------------------------------------------------------------------------------------------------------------------------------------

\subsubsection{Oracle access model: random neighbor queries}
\label{subsubsec:oracle}

The access to the input graph is given by an \emph{oracle}. We consider the \emph{random neighbor oracle}, in which an algorithm is given $n \in \NATURAL$ and access to an input graph $G = (V,E)$ by a query oracle, where $V = \{1, \dots, n\}$. The \emph{random neighbor query} specifies a vertex $v \in V$ and the oracle returns a vertex that is chosen i.u.r. (independently and uniformly at random) from the set of all neighbors of $v$. (Notice that in the random-neighbor model, since $V = \{1, \dots, n\}$, the algorithm can also trivially select a vertex from $V$ i.u.r.)

We believe that the random-neighbor model is the most natural model of computations in the property testing framework in the context of very fast algorithms, and therefore our main focus is on that model.

\begin{remark}
\label{remark:modified-random-neighbor-oracle}
We notice that all our results could be also presented in a variant of the model above in which we would allow only two types of queries: \emph{random vertex query}, which returns a random vertex, and \emph{random neighbor query}, which returns a random neighbor of a given vertex $v$.

Each time we call the random neighbor oracle, the returned random vertex or its random neighbor is chosen \emph{independently and uniformly at random (i.u.r.)}. All vertices of the input graph are accessible and distinguishable by their IDs, and there is no requirement about the IDs other than that they are all distinct. Notice that in this model, the tester does not know $n$, the size of the input graph $G$.
\end{remark}

%---------------------------------------------------------------------------------------------------------------------------------------------------------

\paragraph{Query complexity.}

The \emph{query complexity} of a tester is the number of oracle queries it makes.

We say a graph property $\mathcal{P}$ is \emph{testable} if it has a \emph{tester with constant query complexity}, that is, for every $\eps$, $0 < \eps < 1$, there is $q = q(\eps)$ such that for every $n \in \NATURAL$ the tester has query complexity upper bounded by~$q$ (the complexity may depend on $\mathcal{P}$ and $\eps$, but not on the input graph nor its size).\SArtur{We may want to make a definition out of it, but since this would change the numbering of theorems/claims/etc, I decided to not do it \emph{yet} so, leaving it in the main text.}

%---------------------------------------------------------------------------------------------------------------------------------------------------------

\paragraph{Other oracle access models.}

There are some natural variations of the random neighbor oracle model that have been considered in the literature and that can be relevant here.\SArtur{Of course, we could also mention here the original model of incidence and adjacency queries proposed/used by Goldreich and Ron (see \cite{Gol17}, and also \cite{AKKR08,KKR04,PR02}), which assumes that one knows $n$ and one can ask the following queries:
\begin{itemize}
\item an incidence function $g_1 : [n] \times [n-1] \rightarrow \{0, 1, \dots, n\}$ such that $g_1(u,i) = 0$ if $u$ has less than $i$ neighbors and $g_1(u,i) = v$ if $v$ is the $i$th neighbor of $u$;
\item an adjacency predicate $g_2 : [n] \times [n] \rightarrow \{0, 1\}$ such that $g_2(u,v) = 1$ if and only if $\{u, v\} \in E$.
\end{itemize}
While Goldreich \cite{Gol17} considers to allow both incidence queries and adjacency queries to be the ``most natural choice'' in the context of property testing of general graphs, we believe that especially in the context of analysing constant-query property testers (which is the main focus of this paper), the \emph{random neighbor oracle model} and other models considered in the paper to be \emph{more natural}.
}

\begin{enumerate}[I.]
%\item One could modify the random neighbor oracle model to an oracle model, where instead of using random vertex queries, one assumes that the vertex set is $V = \{1, \dots, n\}$ for $n$ known to the algorithm. (Knowing $n$, one can sample random vertices from $V$ to simulate the random vertex query.)

\item One could extend the random neighbor oracle model to the \emph{random distinct neighbor oracle model}, where one allows for every vertex to query for \emph{distinct} random neighbors (that is, each time we call the random distinct neighbor query for a given vertex $v$, the oracle returns a neighbor of $v$ chosen i.u.r. %independently and uniformly at random
    among all neighbors not returned earlier); if all neighbors have been already returned then the oracle would return a special symbol.

\item One could consider a model allowing two other types of queries:
\begin{inparaenum}%[\sl (a)]
\item[\emph{degree queries:}] for every vertex $v \in V$, one can query the degree of $v$, and
\item[\emph{neighbor queries:}] for every vertex $v \in V$, one can query its $i^{\text{th}}$ neighbor.
\end{inparaenum}
Observe that by first querying the degree of a vertex, we can always ensure that the $i^{\text{th}}$ neighbor of the vertex exists in the second type of query.
\end{enumerate}

It should be noted that while our main focus is on the random neighbor oracle model, our testers (and their analysis) for $H$-freeness can be trivially modified to work in the other three oracle access models (in particular, Theorems \ref{thm:main-H-freeness}, \ref{thm:main-extension}, and \ref{thm:main-H-minor-free-extension} hold in all these models). However, our main result, the characterization of testable properties in planar graphs cannot be extended to the other models (see Section \ref{subsubsec:sensitivity-of-the-models}), other than the variant of the random neighbor oracle discussed in Remark \ref{remark:modified-random-neighbor-oracle}.

%---------------------------------------------------------------------------------------------------------------------------------------------------------

For the sake of completeness, in Appendix \ref{subsec:basic-planar} we recall some basic properties of planar graphs.%, as frequently used in our paper.

%---------------------------------------------------------------------------------------------------------------------------------------------------------

\subsection{Our results}

In this paper we present a characterization of all testable graph properties for general planar graphs in the random neighbor oracle model, showing that, informally, a graph property $\mathcal{P}$ is testable with one-sided error for general planar graphs if and only if testing $\mathcal{P}$ can be reduced to testing for a finite family of finite forbidden subgraphs (see Theorem \ref{thm:characterization}). Further, the results extend to general families of \emph{minor-free graphs}~$G$.

The result is proven in two steps: First we apply a recent result from \cite{CFPS19} (see Theorem \ref{thm:canonical-tester}) to argue in \textbf{Theorem \ref{thm:main-1-sided-reduces-to-H-freeness}} the (easier) necessary condition, that 
\begin{itemize}
\item in the random neighbor oracle model, any graph property $\mathcal{P}$ testable with one-sided error can be reduced to testing for a finite family of forbidden subgraphs.
\end{itemize}
Then we prove \emph{our main technical contribution}, \textbf{Theorem \ref{thm:main-H-freeness}}, that
\begin{itemize}
\item for a given connected finite graph $H$, \emph{subgraph freeness} is \emph{testable} (can be tested with a constant number of queries) on any input \emph{planar graph} $G$, assuming the access to $G$ is via the random neighbor oracle.
\end{itemize}
This latter result extends to arbitrary (not necessarily connected) finite graphs $H$ and to testing for $\mathcal{H}$-freeness for any finite family $\mathcal{H}$ of finite graphs, see Theorem \ref{thm:main-extension} in Section \ref{sec:dics-many}.
By combining these results, we obtain in \textbf{Theorem \ref{thm:characterization}}
\begin{itemize}
\item a characterization of graph properties testable with one-sided error for general planar graphs; this result extends to general minor-free graphs.
\end{itemize}
%

%Further, the results extend to general families of \emph{minor-free graphs} $G$, see Theorem %\ref{thm:main-H-minor-free} --
%\ref{thm:main-H-minor-free-extension} in Section \ref{sec:minor-free}.\Artur{This should be extended in Section \ref{sec:minor-free} to the characterization for these classes of graphs (I mean, Theorem \ref{thm:main-H-minor-free-extension} and Section \ref{sec:minor-free} is for now only about testing $\mathcal{H}$-freeness).}

While we believe that our general characterization of all testable graph properties of planar and minor-free graphs is a central problem in property testing and is the main contribution of this paper, we also hope that our constant query time tester for subgraph freeness will further advance our understanding of efficient algorithms for that fundamental problem.\SArtur{Possibly more text here?}

Our work is a continuation of our efforts to understand the complexity of testing basic graph properties in graphs with no bounds for the degrees. Indeed, while major efforts in the property testing community have been put to study dense graphs and bounded degree graphs (cf. \cite[Chapter~8-9]{Gol17}), we have seen only limited advances in the study of arbitrary graphs, in particular, sparse graphs but without any bounds for the maximum degrees. We believe that this model is one of the most natural models, and it is also most relevant to computer science applications (see also the motivation in \cite[Chapter~10.5.3]{Gol17}). While the understanding of testing in general graphs is still elusive, our work makes a major step forward towards understanding of testing properties for most extensively studied classes of graphs, in our case of planar and minor-free graphs.

%---------------------------------------------------------------------------------------------------------------------------------------------------------

\subsubsection{Overview: Any testable property can be reduced to testing for forbidden subgraphs}
\label{subsubsec:necessary-cond-testing}

We begin with an \emph{easier part} of our characterization (see Section \ref{sec:testable-implies-forbidden-subgraphs} for details). Our approach follows the method of canonical testers for graph properties testable for general graphs developed recently in \cite{CFPS19}. The \emph{intuition} here is rather simple: if a graph property $\mathcal{P}$ is testable then all what the tester can do is for a given input graph $G$ to randomly sample a constant number of vertices and then to explore their neighborhoods of constant size, and on the basis of the visited subgraph $U$ of $G$ to decide whether to accept the input graph or to reject it. Further, the assumption that we consider a one-sided error tester implies that the tester must always accept any graph $G$ satisfying $\mathcal{P}$. Therefore, in particular, if we define $\mathcal{H}$ as the family of all $U$ for which the tester rejects any input graph $G$ that contains $U$, then we can argue that any graph $G$ satisfying $\mathcal{P}$ must be $\mathcal{H}$-free. The analysis can be easily extended to hold for an arbitrary class of the input graphs, e.g., for planar graphs.

(Notice that these arguments show only that any testable graph property $\mathcal{P}$ has a finite family $\mathcal{H}$ of finite graphs such that $\mathcal{P}$ is $\mathcal{H}$-free. However, we do not say anything about any other properties of $\mathcal{P}$; indeed, $\mathcal{P}$ may be not only $\mathcal{H}$-free but also may have some other properties. A good example showing the sensitivity of this notion is testing bipartiteness. It has been shown \cite{CMOS11} that for general planar graphs bipartiteness is testable with one-sided error, but clearly, bipartiteness cannot be defined as a property of $\mathcal{H}$-freeness for a \emph{finite} family $\mathcal{H}$ of forbidden graphs. However, one can easily show (cf. [10, Section 2.1]) that if an input graph $G$ is $\eps$-far from bipartitiness, then there must be an odd $k = O(1/\eps^2)$, so that $G$ is $O(\eps)$-far from $C_k$-free, and this fact suffices to argue that bipartitiness for planar graphs is testable.)

To turn this intuition into a formal proof, we need to do some additional work. We rely heavily on the canonical tester developed recently in \cite{CFPS19} to argue that to test any testable graph property we can assume that the tester at hand is ``oblivious'' and works non-adaptively. This allows us to obtain a clean characterization of forbidden subgraphs for any given testable property $\mathcal{P}$. Further, we lift this characterization to extend the analysis to \emph{semi-subgraph-free} graph properties, %see Definition \ref{def:semi-H-freeness},
which are graph properties defined as $\mathcal{H}$-free or close to $\mathcal{H}$-free, for some finite family $\mathcal{H}$ of finite graphs. The analysis is presented in Section~\ref{sec:testable-implies-forbidden-subgraphs} (see Theorem~\ref{thm:main-1-sided-reduces-to-H-freeness}).

%---------------------------------------------------------------------------------------------------------------------------------------------------------

\subsubsection{Overview: Testing for forbidden subgraphs in planar graphs and minor-free graphs}
\label{subsubsec:sufficient-cond-testing}

The \emph{main technical contribution} of this paper is a proof that \emph{for planar graphs}, the property of being $H$-free is testable with one-sided error for every connected finite subgraph $H$, in the random neighbor oracle model, see Theorem \ref{thm:main-H-freeness}. %The proof of this result is complex and long.
This result extends to arbitrary (not necessarily connected) finite graphs $H$ and to testing for $\mathcal{H}$-freeness for any finite family $\mathcal{H}$ of finite graphs, see Theorem \ref{thm:main-extension}. Further, the results extend to general families of \emph{minor-free graphs} $G$, see Theorem %\ref{thm:main-H-minor-free} --
\ref{thm:main-H-minor-free-extension}.

Let us first discuss the challenges of the task of testing $H$-freeness. It has been known for a long time that for \emph{bounded degree graphs} one can test $H$-freeness with a constant number of queries using the following simple tester: randomly sample a constant number of vertices and check whether any of them belongs to a copy of $H$. This result relies on two properties of bounded degree graphs:
\begin{inparaenum}[\it (i)]
\item that it is easy to test whether a given vertex belongs to a copy of $H$ (just run a BFS of depth $|V(H)|$), and
\item that if a given graph is $\eps$-far from $H$-free then it has many edge-disjoint copies of $H$ that cover a total of a linear number of vertices.
\end{inparaenum}
But both these properties fail to work for general graphs. For example, for {\it (ii)}, consider an $n$-vertex graph $G$ with $n - \sqrt{n}$ isolated vertices and $\sqrt{n}$ vertices forming a clique. It is easy to see that $G$ is $\eps$-far from $H$-free (for a sufficiently small $\eps$ with respect to the size of $H$), but all copies of $H$ in $G$ are covered only by $\sqrt{n}$ vertices and as the result, testing $H$-freeness trivially requires $\Omega(\sqrt{n})$ queries: one has to perform so many queries (in expectation) to hit a first non-isolated vertex.

In our analysis, by focusing on planar (or minor-free) graphs, we are able to circumvent the latter obstacle {\it (ii)} (argued implicitly in Lemma \ref{lemma:edge-disjoint-copies-H-free}), but the former obstacle {\it (i)} still persists. Our approach to cope with {\it (i)} is by devising a simple modification of BFS search, random bounded-breadth bounded-depth search. By bounding the breadth and depth of the graph exploration we are able to ensure that the complexity of the tester is bounded. However, then the main challenge in our analysis is to analyze this process, to show that indeed, it distinguishes between $H$-free graphs and graphs that are $\eps$-far from $H$-free.

Our approach relies on a proof that for any planar graph $G$ that is $\eps$-far from $H$-free there exists a set $\DCH$ of edge-disjoint copies of $H$ such that,
\begin{enumerate}[\it (i)]
%\begin{inparaenum}[\it (i)]
\item if we can find a copy of $H$ in $\G[\DCH]$ with a constant number of queries, then also in $G$ we can find a copy of $H$ with a constant number of queries, and
\item if the input graph was $\G[\DCH]$, then we could find a copy of $H$ with a constant number of queries.\SArtur{To make the picture clearer, maybe we could have some arguments why on its own {\it (ii)} is non-trivial. I can see some potential readers claiming that {\it (ii)} is trivial since ???, e.g., \textcolor[rgb]{0.50,0.00,1.00}{\emph{it's obvious that if $G$ consists only of edge-disjoint copies of $H$, and in fact of a linear number of them, then we can find one copy easily}}.}
%\end{inparaenum}
\end{enumerate}
%
%(Here, $\G[\DCH]$ denotes the subgraph of $G$ induced by the edges of the copies of $H$ in $\DCH$.)

The construction of the set $\DCH$ is existential, and is performed by a process of gradually deleting edges of $G$ so that after each round of edge deletions, {\it (i)} is maintained, and so that at the end, the structure of $\G[\DCH]$ is simple enough so that {\it (ii)} is easy. The process is controlled by a sequence of \emph{contractions}: we reduce the problem of finding a copy of $H$ in $G$ to the problem of finding a copy of $H$ with one vertex contracted, which in turn, we reduce to the problem of finding a copy of $H$ with two vertices contracted, and so on so forth. The idea is that if at the end of this process, we have to find a copy of $H$ contracted to single vertex, this task is easy to analyze. The main challenge of our analysis here is to carefully manage the contractions to have the analysis going through. In a similar context, the authors in \cite{CMOS11} have been arguing that this task is already very complex for cycles in the analysis of constant-length random walks in planar graphs, that is, graphs with good separators and bad expansion. However, by using a sequence of self-reductions relying on contractions (and hence reducing testing $C_k$-freeness to testing $C_{k-1}$-freeness, where $C_k$ is a cycle of length $k$), the authors in \cite{CMOS11} were able to show there that for planar graphs, \emph{testing bipartiteness} (implicitly, testing $C_k$-freeness for constant $k$) can be done with constant query complexity and with one-sided error.

The approach presented in our paper can be seen as a major extension of the approach used for testing bipartiteness in \cite{CMOS11} to test $H$-freeness, %even though we argue that
though
the problem of testing $H$-freeness is significantly more complex. Indeed, the central tool used for bipartiteness, contractions of a path or a cycle, becomes problematic when the forbidden graph $H$ has vertices of degree higher than~2. The challenge here is that to contract vertices of higher degrees, the information about their neighbors is difficult to be maintained. Still, we follow a similar approach, but since we cannot perform the contraction in term of graphs, we do it via introducing \emph{hyperedges}, to ensure that after contracting high degree vertices the information about their neighbors is memorized in a form of a \emph{hypergraph}. This extension of the framework from graphs to hypergraphs makes the entire analysis significantly more complicated and one of our main technical contributions is to make the analysis work for this case. For example, one central challenge is to ensure that the input graph, originally planar, maintain some planarity properties even after applying a sequence of contractions. This task is not very difficult if the contractions were performed in graphs, but when we have to process hypergraphs, maintaining planarity seems to be entirely \emph{hopeless}. Still, we will show how to efficiently model the connectivity information of the hypergraph using the concept of shadow graphs that are unions of planar graphs.

The analysis is long, with many subtle fine points, and is presented in details in Sections \ref{sec:outline-proof}--\ref{sec:final-proof}.

%---------------------------------------------------------------------------------------------------------------------------------------------------------

\begin{remark}
While in our analysis we did not try to optimize the complexity of the $H$-freeness tester, focusing on the task of obtaining the query complexity of $O_{\eps,H}(1)$, let us mention that in fact, with the analysis as it is now, without any optimization efforts, the complexity of our tester is doubly exponential in $\Hsize/\eps$.
%$2^{2^{\text{poly}(\Hsize/\eps)}}$. %\Artur{I put a random number; maybe one will want to check it later.}
\end{remark}

%---------------------------------------------------------------------------------------------------------------------------------------------------------

\begin{remark}
While our main focus is on the random neighbor oracle model, it is straightforward to extend our testers and their analysis for $H$-freeness and for $\mathcal{H}$-freeness to the other two oracle access models presented in Section \ref{subsubsec:oracle}.
(However, our main result, the characterization of testable properties in planar graphs (and Theorem \ref{thm:main-1-sided-reduces-to-H-freeness}), cannot be extended to the other models (cf. Section \ref{subsubsec:sensitivity-of-the-models}), except the variant of the random neighbor oracle from Remark~\ref{remark:modified-random-neighbor-oracle}.)\SArtur{Think about it (\& remember to stay consistent, since this issue was mentioned in some other places (cf. Section~\ref{sec:minor-free})).}
\end{remark}

%---------------------------------------------------------------------------------------------------------------------------------------------------------

\subsubsection{Characterization of graph properties testable with one-sided error for planar/minor-free graphs}
\label{Intro-characterization}

By combining the results sketched in Sections \ref{subsubsec:necessary-cond-testing} and \ref{subsubsec:sufficient-cond-testing}, the following characterization of graph properties testable with one-sided error (in the random neighbor oracle model) for general planar graphs and for minor-free graphs follows:\SArtur{Later one may want to add some text about \textbf{non-uniformity} etc.}

\begin{theorem}
\label{thm:characterization}
A graph property $\mathcal{P}$ is testable with one-sided error in the random neighbor oracle model for planar graphs (and for minor-free graphs) if and only if $\mathcal{P}$ is semi-subgraph-free.
\end{theorem}

The proof of Theorem \ref{thm:characterization} follows immediately from our Theorem \ref{thm:main-1-sided-reduces-to-H-freeness} (necessary condition) and Theorems~\ref{thm:main-extension} and \ref{thm:main-H-minor-free-extension} (sufficient condition).

One can read this characterization informally as follows:

\noindent\emph{A graph property $\mathcal{P}$ is testable with one-sided error in the random neighbor oracle model for planar graphs (or for minor-free graphs) if and only if $\mathcal{P}$ can be described as a property of testing forbidden subgraphs of constant size (the maximum size of any forbidden subgraph can be a function of $\mathcal{P}$ and $\eps$).}

%---------------------------------------------------------------------------------------------------------------------------------------------------------

\subsubsection{Remarks on the sensitivity and robustness of the oracle access models}
\label{subsubsec:sensitivity-of-the-models}

While our tester for $H$-freeness (Section \ref{subsubsec:sufficient-cond-testing}) is robust, the characterization presented in Theorem \ref{thm:characterization} is very sensitive to the oracle model. For example, it might be natural to consider a variant of our random neighbor oracle model to allow for every vertex to query for \emph{distinct} random neighbors. That is, each time we call the random distinct neighbor query for a given vertex $v$, the oracle will return a neighbor of $v$ chosen i.u.r. among all neighbors not returned earlier. One important feature of this model is that after $\deg(v)+1$ queries for a random distinct neighbor of vertex $v$, we are able to detect the degree $\deg(v)$ of vertex $v$ in the input graph. This makes this model more powerful than our random neighbor oracle model, and in particular, it allows to test some properties that cannot be reduced to testing for forbidden subgraphs. For example, in that model one can test connectivity with $O(1/\eps^3)$ queries and one-sided error (see, e.g., \cite{GR97}). Indeed, if the input graph $G$ is $\eps$-far from being connected, then it is easy to see that $G$ must have $\frac12 \eps n$ connected components of size at most $\frac2{\eps}$. Therefore, after randomly sampling $\frac{3}{\eps}$ vertices, with probability at least $\frac23$ one of the randomly sampled vertices will be in one of these small connected components. Since all vertices in this component must have degree at most $\frac2{\eps}$, we can run BFS algorithm to explore the entire connected component with $O(1/\eps^2)$ random distinct neighbor queries and verify that this connected component is indeed small, proving that the input graph is $\eps$-far from being connected. This can be easily formalized to obtain a one-sided error tester for connectivity with query complexity $O(1/\eps^3)$ in the random distinct neighbor oracle model. However, this task cannot be efficiently performed in our random neighbor oracle model (since we can never confirm with a finite number of queries a degree of a given vertex, even if its degree is constant, even if it is 1), and indeed, connectivity testing cannot be reduced to testing for a finite family of forbidden subgraphs and is not is a semi-subgraph-free graph property, even in planar graphs. (This is in contrast to other characterizations presented earlier in the literature, e.g., in \cite{AS08}, where the tester for the dense graphs model reduces to testing for forbidden \emph{induced} subgraphs, giving a complete characterization of properties testable with one-sided error in terms of hereditary properties.) And so, even for planar graphs, \emph{testing connectivity in the random neighbor oracle model is impossible with one-sided error}!\footnote{To see this, consider two \emph{planar} graphs: a cycle $C_n$ on $n$ vertices, which is connected, and a perfect matching $M_n$ on $n$ vertices, which is $\eps$-far from connected (for $\eps < \frac12$). Any tester should reject $M_n$ with probability at least $\frac23$. But at the same time, if we consider the tester on $C_n$ (which must be accepted) then after performing $q$ queries, with probability at least $2^{-q}$, and so with positive probability, it will see only a subgraph of $M_n$. Therefore, since we consider one-sided error testers which must accept $C_n$, we conclude that no \emph{one-sided error} tester \emph{can reject} $M_n$.}

%---------------------------------------------------------------------------------------------------------------------------------------------------------

\subsection{Organization of the paper}

We begin in Section \ref{sec:testable-implies-forbidden-subgraphs} with a formal analysis showing the necessary part of our characterization of testable properties, that any testable property %can be reduced to test forbidden subgraphs, that is,
is semi-subgraph-free (cf. Theorem \ref{thm:main-1-sided-reduces-to-H-freeness} in Section \ref{subsec:canonican-tester-and-reduction}).

Then, in Sections \ref{sec:outline-proof}--\ref{sec:final-proof}, we present the \emph{main technical contribution} of this paper, a complete analysis showing the sufficient part of our characterization of testable properties in planar graphs, that for any finite graph $H$, testing $H$-freeness is testable in planar graphs. The analysis here is split into several sections, with some auxiliary and technical results deferred to the appendix (Appendix \ref{sec:condition-a-prime}--\ref{sec:proof-lemma:small-vertices-new}). We begin in Section \ref{sec:outline-proof} with an outline of the proof of testing $H$-freeness, focusing on connected $H$. Then, in Section \ref{sec:exploration-testing-H-free}, we present our tester and define our framework. Section \ref{subsec:constructing-H1} gives the first (and easiest) step in our transformation and show that any graph that is $\eps$-far from $H$-free has a linear number of edge-disjoint copies of $H$. Then, in Section \ref{subsec:Ui->Ui+1}, we show how the contractions (cf. Section \ref{subsubsec:sufficient-cond-testing}) can be performed in hypergraphs, to ensure existence of a sought set $\DCH$ of edge-disjoint copies of $H$ in which we can detect a copy of $H$. The analysis is then completed in Section \ref{sec:final-proof}. Finally, in Section \ref{sec:dics-many} we discuss the extension to families of arbitrary finite graphs and in Section \ref{sec:minor-free} we discuss the extension to minor-free graphs.

Some final conclusions are in Section \ref{sec:conclusions}.
\SArtur{Do we want to have any additional section about the characterization? Or does Section \ref{Intro-characterization} suffice?}

%---------------------------------------------------------------------------------------------------------------------------------------------------------

\section{Any testable property can be reduced to testing for forbidden subgraphs}
\label{sec:testable-implies-forbidden-subgraphs}

\SArtur{Of course, one could change the title to the more precise: \textbf{``Any testable property is semi-subgraph-free''} but to me, such title would sound less informative. But maybe I'm wrong.}%
In this section we provide a formal proof of the necessary (and easier) condition in our characterization, Theorem \ref{thm:main-1-sided-reduces-to-H-freeness}, that any one-sided-error testable property for arbitrary graphs can be reduced to testing for forbidden subgraphs of constant size (this claims holds for any finite family of graphs, not only for planar graphs). It should be noted that each graph in the family of forbidden graphs may have size depending on $\eps$, the proximity parameter of the tester.\SArtur{Note: in principle, our tester have forbidden graphs depending on $n$, even if their sizes are upper bounded by $O_{\eps}(1)$.}

Our analysis critically relies on a recently developed in \cite{CFPS19} canonical tester that shows that to test any testable graph property we can assume that the tester at hand is ``oblivious'' and works non-adaptively. This will allow us later to obtain a clean characterization of forbidden subgraphs for any given testable property~$\mathcal{P}$.

%---------------------------------------------------------------------------------------------------------------------------------------------------------

\subsection{Bounded-breadth bounded-depth graph exploration and bounded-discs}
\label{subsec:BFS-like-search-bounded-discs}

Our analysis relies on a random (BFS-like) bounded-breadth bounded-depth search, \Traverse{} below, an exploration algorithm similar to BFS of depth $\ld$. The algorithm runs from a given vertex a random BFS-like exploration of breadth $\dg$ and of depth $\ld$ using the random neighbor oracle (i.e., every vertex selects $\dg$ of its neighbors i.u.r. and recursively continues the process from them, until depth $\ld$ is reached). The main difference is that instead of visiting all neighbors of every vertex, like in the standard BFS algorithm, we \emph{visit only $\dg$ neighbors chosen i.u.r.}, to limit the complexity of the search algorithm.

%---------------------------------------------------------------------------------------------------------------------------------------------------------
\begin{algo}\label{alg:traverse}
\Traverse\,${(G,v,\dg,\ld)}$:
\begin{itemize}
\item %Pick a random vertex $v \in V$; let $L_0 = \{v\}$.
    Let $L_0 = \{v\}$.
\item For $\ell = 1$ to $\ld$ do:
    \begin{itemize}[$\diamond$]
    \item Let $L_{\ell} = \emptyset$ and $\mathcal{E}_{\ell} = \emptyset$.
    \item For every $u \in L_{\ell-1}$ do:
        \begin{itemize}[$\circ$]
        \item Choose $\dg$ neighbors of $u$ using $\dg$ \emph{random neighbor queries}; call them $\Gamma_u$.
        \item Let $\mathcal{E}_u = \{ (u,x): x \in \Gamma_u\}$.
        \item Set $L_{\ell} = L_{\ell} \cup \Gamma_u$ and $\mathcal{E}_{\ell} = \mathcal{E}_{\ell} \cup \mathcal{E}_u$.
        \end{itemize}
    \item $L_{\ell} = L_{\ell} \setminus \bigcup_{i=0}^{\ell-1} L_i$.
    \end{itemize}
\item \textbf{Return} the subgraph of $G$ induced by the edges $\bigcup_{\ell=1}^{\ld} \mathcal{E}_{\ell}$.
\end{itemize}
\end{algo}
%---------------------------------------------------------------------------------------------------------------------------------------------------------

We use the notion of bounded-breadth/depth search \Traverse{} to define bounded discs.
	
\begin{definition}\textbf{($(\dg,\ld)$-bounded disc)}\it
\label{def:bounded-disc}
For given $\dg, \ld \in \NATURAL$, graph $G = (V,E)$, and vertex $v \in V$, a \textbf{$(\dg,\ld)$-bounded disc} of $v$ in $G$ is any subgraph $U$ of $G$ that can be returned by \Traverse\,${(G,v,\dg,\ld)}$.
\junk{
satisfying the following properties:
\begin{itemize}
\item the maximum distance from $v$ to any other vertex in $U$ is at most $\ld$;
\item there exists an ordering $u_0 = v, u_1, \cdots, u_{|V(U)|}$ of vertices in $U$, and for each vertex $u_i$, there exists a subset $F_i$ of the set of all edges incident to $u_i$ with $|F_i| \le \max\{q, \deg(u_i)\}$, such that
    \begin{itemize}[$\diamond$]
    \item one can run \Traverse\,${(G,\dg,\ld)}$ starting at vertex $v$ with respect to the vertex ordering and for each $u_i$ uses edges in $F_i$ to exactly recover $U$.
    \end{itemize}
\end{itemize}
}

Vertex $v$ is called a \textbf{root} of the $(\dg,\ld)$-bounded disc $U$.
\end{definition}

Let us observe that, assuming that $\dg \ge 2$, \Traverse\,${(G,v,\dg,\ld)}$ performs $\sum_{i=1}^{\ld} \dg^i \le 2 \dg^{\ld}$ queries to the input graphs. Accordingly, for $\dg \ge 2$, any $(\dg,\ld)$-bounded disc has at most $\sum_{i=0}^{\ld} \dg^i
%= \frac{\dg^{\ld+1} - 1}{\dg - 1}
\le 2 \dg^{\ld}$ vertices and at most $\sum_{i=1}^{\ld} \dg^i \le 2 \dg^{\ld}$ edges.

%---------------------------------------------------------------------------------------------------------------------------------------------------------

\subsection{Rooted graphs, their basic properties, and semi-rooted-subgraph-freeness}

In our analysis it will be sometimes useful to consider also \emph{rooted} graphs, that is, graphs with some number of vertices distinguished as special vertices called \emph{roots}. (For example, bounded discs from Definition \ref{def:bounded-disc} are rooted graphs.) To analyze similarities between rooted graphs, we will use the following definition.
	
\begin{definition}\textbf{(Root-preserving isomorphism)}
\it
\label{def:root-preserving-isomorphism}
Let $Q = (V(Q), E(Q))$ and $Q' = (V(Q'), E(Q'))$ be two rooted graphs. A \textbf{root-preserving isomorphism} between $Q$ and $Q'$, denoted $Q \cong_r Q'$, is a bijection $f: V(Q) \rightarrow V(Q')$ such that $u$ is the root of $V(Q)$ if and only if $f(u)$ is the root of $V(Q')$, and $(u,v) \in E(Q)$ if and only if $(f(u),f(v)) \in E(Q')$.

If $Q \cong_r Q'$, then we say that $Q$ is \textbf{root-preserving isomorphic} to $Q'$.
\end{definition}

We will extend this definition to compare a rooted graph with its occurrences (in a sense of root-preserving isomorphisms) in a large graph (which does not necessarily have to be rooted).
	
\begin{definition}\it
\label{def:root-preserving-copies}
Let $G$ be an undirected graph and let $Q$ be a rooted graph.
A \textbf{rooted copy of $Q$ in $G$} is a subgraph $U$ of $G$ such that one can assign the roots to $U$ so that there is a root-preserving isomorphism between $Q$ and the rooted version of $U$.
%
%We say that \textbf{$G$ is $Q$-rooted-free} if there is no rooted copy of $Q$ in~$G$.
%
For an arbitrary set $\mathcal{Q}$ of rooted graphs, we say that \textbf{$G$ is $\mathcal{Q}$-rooted-free} if for every $Q \in \mathcal{Q}$, there is no rooted copy of $Q$ in $G$.
\end{definition}

With these definitions, we are ready to present our auxiliary graph property notion.

\begin{definition}\textbf{(Semi-rooted-subgraph-freeness)}\it
\label{def:semi-H-rfreeness}
A graph property $\mathcal{P}$ is \textbf{\emph{semi-rooted-subgraph-free}} if for every $\eps$, $0 < \eps < 1$, and every $n \in \NATURAL$, there is a finite family $\mathcal{H}$ of finite graphs such that the following hold:
\begin{enumerate}[(i)]
\item any graph $G$ satisfying $\mathcal{P}_n$ is $\mathcal{H}$-rooted-free, and
\item any graph $G$ which is $\eps$-far from satisfying $\mathcal{P}_n$, is not $\mathcal{H}$-rooted-free.
\end{enumerate}
\end{definition}

Similarly as in Definition \ref{def:semi-H-freeness}, the notion of a family $\mathcal{H}$ of \emph{finite graphs} means that every graph $H \in \mathcal{H}$ is finite, i.e., $|V(H)| = O_{\eps}(1)$.

%---------------------------------------------------------------------------------------------------------------------------------------------------------

\subsection{Modeling forbidden subgraphs in rooted graphs}

While our analysis uses rooted graphs, their use is purely auxiliary because of the following simple fact.

\begin{lemma}
\label{lemma:semi-rooted-subgraph-free-yields-semi-subgraph-free}
If a graph property $\mathcal{P}$ is semi-rooted-subgraph-free then $\mathcal{P}$ is also semi-subgraph-free.
\end{lemma}

\begin{proof}
This follows easily from the definitions of semi-rooted-subgraph-free and semi-subgraph-free properties. For any rooted graph $H$, let $\overline{H}$ denote the same graph with removed roots (that is, we remove the labels defining the roots); similarly, for any family $\mathcal{H}$ of rooted graphs, let $\overline{\mathcal{H}} = \{\overline{H}: H \in \mathcal{H}\}$. Then we claim that for any graph $G$ be an arbitrary graph and any family $\mathcal{H}$ of rooted graphs,
\begin{enumerate}[\it \ \ \ \ \ \ \ \ (a)]
%\begin{inparaenum}[\it (a)]
\item if $G$ is $\mathcal{H}$-rooted-free then $G$ is also $\overline{\mathcal{H}}$-free, and
\item if $G$ is not $\mathcal{H}$-rooted-free, then $G$ is also not $\overline{\mathcal{H}}$-free.
%\end{inparaenum}
\end{enumerate}

Indeed, to see part \emph{(a)}, suppose, by contradiction, that $G$ is not $\overline{\mathcal{H}}$-free, that is, there is $\overline{H}$ with $H \in \mathcal{H}$ such that $\overline{H}$ is a subgraph of $G$. But then $G$ has a rooted copy of $H$, since we can take the roots of $H$ and assign them to $\overline{H}$, so that there is a root-preserving isomorphism between $H$ and the rooted version of $\overline{H}$. Since $G$ has a rooted copy of $H$, we conclude that $G$ is not $\mathcal{H}$-rooted-free, which is contradiction.

To see part \emph{(b)}, suppose, by contradiction, that $G$ is $\overline{\mathcal{H}}$-free, that is, there is no $\overline{H} \in \overline{\mathcal{H}}$ such that $G$ has a copy of $\overline{H}$. But then, clearly, $G$ is $\mathcal{H}$-rooted-free, since otherwise, there would be $H \in \mathcal{H}$ such that $G$ had a rooted copy of $H$, which would imply that $\overline{H}$ was a subgraph $G$; contradiction.

Now, we are ready to complete the proof of Lemma \ref{lemma:semi-rooted-subgraph-free-yields-semi-subgraph-free}. By Definition \ref{def:semi-H-rfreeness}, if $\mathcal{P}$ is semi-rooted-subgraph-free then there exists a finite family $\mathcal{H}$ of finite rooted graphs such that
\begin{inparaenum}[(i)]
\item any graph $G$ satisfying $\mathcal{P}$ is $\mathcal{H}$-rooted-free, and
\item any graph $G$ which is $\eps$-far from satisfying $\mathcal{P}$, is not $\mathcal{H}$-rooted-free.
\end{inparaenum}
If we combine these properties with our claim above, then we obtain that for the finite family of finite graphs $\overline{\mathcal{H}} = \{\overline{H}: H \in \mathcal{H}\}$,
\begin{enumerate}[\it (i')]
\item any graph $G$ satisfying $\mathcal{P}$ is $\mathcal{H}$-rooted-free, and thus (by \emph{(a)}) also $\overline{\mathcal{H}}$-free, and
\item any graph $G$ which is $\eps$-far from satisfying $\mathcal{P}$, is not $\mathcal{H}$-rooted-free, and thus (by \emph{(b)}) also not $\overline{\mathcal{H}}$-free.
\end{enumerate}

Therefore $\mathcal{P}$ is semi-subgraph-free (cf. Definition \ref{def:semi-H-freeness}).
\end{proof}

%---------------------------------------------------------------------------------------------------------------------------------------------------------

\subsection{Canonical testers and reduction to testing for forbidden subgraphs}
\label{subsec:canonican-tester-and-reduction}

Next, our analysis follows the framework described in Section \ref{subsubsec:necessary-cond-testing}. We rely on the following Theorem~3.6 from \cite{CFPS19} describing a \emph{canonical way of designing any tester in the random neighbor oracle model}.\SArtur{A generic comment about the analysis below: I have taken \emph{some} version of the Canonical Tester from \cite{CFPS19}, and it's possible that one will have to revise the formulation of Theorem \ref{thm:canonical-tester} (I hope though, not the statement).}

\begin{theorem}[\textbf{Canonical tester \cite{CFPS19}}]
\label{thm:canonical-tester}
Let $\mathcal{P} = (\mathcal{P}_n)_{n \in \NATURAL}$ be a graph property that can be tested in the random neighbor oracle model with query complexity $q = q(\eps)$ and error probability at most~$\frac13$. Then for every $\eps$, there exists $q' = \Theta(q)$, and an infinite sequence $\mathcal{Q} = (\mathcal{Q}_{n})_{n \in \NATURAL}$ such that for every $n \in \NATURAL$,
\begin{itemize}
\item $\mathcal{Q}_{n}$ is a set of rooted graphs such that each $Q \in \mathcal{Q}_{n}$ is the union of $q'$ many $(q',q')$-bounded discs;
\item the property $\mathcal{P}_n$ on $n$-vertex graphs can be tested with error probability at most $\frac13$ by the following canonical tester (with query complexity $q^{O(q)}$):
    \begin{itemize}[$\diamond$]
	\item sample a set (possibly, a multiset) $S$ of $q'$ vertices chosen i.u.r.;
	\item for each sampled vertex $v$, run \Traverse\,${(G,v,q',q')}$ to get a $(q',q')$-bounded disc~$U_v$;
	\item reject if and only if there exists a root-preserving isomorphism between the union of the explored $(q',q')$-bounded discs and some element $Q \in \mathcal{Q}_{n}$, that is, there is $Q \in \mathcal{Q}_{n}$ with $\bigcup_{v \in S}U_v \cong_r Q$.
	\end{itemize}
\end{itemize}
Furthermore, if $\mathcal{P} = (\mathcal{P}_n)_{n \in \NATURAL}$ can be tested in the random neighbor oracle model with query complexity $q(\eps)$ with one-sided error, then the resulting canonical tester for $\mathcal{P}$ has one-sided error too.
\end{theorem}

Theorem \ref{thm:canonical-tester} from \cite{CFPS19} shows that without loss of generality, we can assume that any testable graph property can be tested by a canonical tester with constant query complexity. With Theorem \ref{thm:canonical-tester}, Lemma \ref{lemma:semi-rooted-subgraph-free-yields-semi-subgraph-free}, and Definitions \ref{def:semi-H-freeness} and \ref{def:semi-H-rfreeness} at hand, we are now ready to present the main result of this section.

\begin{theorem}
\label{thm:main-1-sided-reduces-to-H-freeness}
If a graph property $\mathcal{P}$ is testable with one-sided error in the random neighbor oracle model then $\mathcal{P}$ is semi-subgraph-free.
\end{theorem}

\begin{proof}
First, notice that thanks to Lemma \ref{lemma:semi-rooted-subgraph-free-yields-semi-subgraph-free}, it is enough to show that if a graph property $\mathcal{P}$ is testable with one-sided error in the random neighbor oracle model then $\mathcal{P}$ is semi-\emph{rooted}-subgraph-free (cf. Definition~\ref{def:semi-H-rfreeness}).
\junk{
(Our analysis follows similar arguments as those used for oblivious one-sided-error testing of semi-hereditary properties for the dense graph model, due to Alon and Shapira \cite{AS08}.)
}

Let us fix $n \in \NATURAL$ and $\eps$, and suppose that $\mathcal{P}_n$ is a graph property on $n$-vertex graphs that can be tested in the random neighbor oracle model with query complexity $q(\eps)$ and one-sided error. By Theorem \ref{thm:canonical-tester} from \cite{CFPS19}, we can assume that $\mathcal{P}_n$ is tested by a canonical tester $\mathfrak{T}$ that satisfies the conditions of Theorem \ref{thm:canonical-tester}. In particular, let $\mathcal{Q}_{n}$ be the family of forbidden rooted graphs for $\mathcal{P}_n$ (union of $q'$ many $(q', q')$-bounded discs) whose existence follows from Theorem \ref{thm:canonical-tester}. We will show that so defined family of rooted graphs satisfies the conditions in Definition~\ref{def:semi-H-rfreeness}, proving that $\mathcal{P}$ is semi-subgraph-free.

Let us first notice that each rooted graph $\mathcal{Q}_{n}$ has at most $2 (q')^{q'}$ vertices and at most $2 (q')^{q'}$ edges, and so $\mathcal{Q}_{n}$ is a finite family of finite rooted graphs.

Let us next show item {\it (i)} of Definition \ref{def:semi-H-rfreeness}, that any $n$-vertex graph $G$ satisfying $\mathcal{P}_n$ is $\mathcal{Q}_n$-rooted-free (cf. Definition~\ref{def:root-preserving-copies}). The proof is by contradiction and so suppose that there is a graph $G$ satisfying $\mathcal{P}_n$ which contains a rooted copy of $Q \in \mathcal{Q}_n$. Then, with a positive probability the canonical tester $\mathfrak{T}$ on $G$ will take that copy of $Q \in \mathcal{Q}_n$, and by the definition, it will reject $G$. This means that the tester has a nonzero probability of rejecting $G$, contradicting our assumption that the tester $\mathfrak{T}$ is one-sided.

Now, we want to prove item {\it (ii)} of Definition \ref{def:semi-H-rfreeness}. Let $G$ be an $n$-vertex graph that is $\eps$-far from satisfying $\mathcal{P}_n$. Any tester for $\mathcal{P}_n$ should reject $G$ with nonzero probability. By definition of our canonical tester $\mathfrak{T}$, $G$ is rejected by $\mathfrak{T}$ only if $G$ contains a rooted subgraph $U$ such that if the tester $\mathfrak{T}$ gets $U$ from the oracle, then $U \cong_r Q$. By definition of $\mathfrak{T}$ and $\mathcal{Q}_n$ this means that $Q \in \mathcal{Q}_n$, which proves item {\it (ii)} of Definition \ref{def:semi-H-rfreeness}.

We have shown that if a graph property $\mathcal{P}$ is testable with one-sided error in the random neighbor oracle model then $\mathcal{P}$ is semi-rooted-subgraph-free. By Lemma \ref{lemma:semi-rooted-subgraph-free-yields-semi-subgraph-free}, this yields that $\mathcal{P}$ is semi-subgraph-free, completing the proof.
\end{proof}

\begin{remark}
While Theorem \ref{thm:main-1-sided-reduces-to-H-freeness} is presented for any general graphs, it is straightforward to extend it to hold also for infinite classes of graphs, for example, for planar graphs, or for the family of minor-closed graphs.\SArtur{Do we need more comments here?}
\end{remark}

%---------------------------------------------------------------------------------------------------------------------------------------------------------

\subsection{Uniform characterization using oblivious testers and forbidden subgraphs}
\label{subsec:uniform-characterization}

While Theorem \ref{thm:canonical-tester} from \cite{CFPS19} allows to simplify the analysis of testable properties, the analysis as in Theorem \ref{thm:main-1-sided-reduces-to-H-freeness} obtains non-uniform testers, in the sense of the dependency on $n$. We could make our result uniform by considering a special class of uniform testers, which we call \emph{oblivious testers}, that capture the essence of testers of testable properties in the flavor of Theorem \ref{thm:canonical-tester} (see \cite{AS08} for a similar notion in the context of testing dense graphs). We will discuss this characterization in Appendix \ref{sec:uniform-characterization}.

%---------------------------------------------------------------------------------------------------------------------------------------------------------

\section{Testing $H$-freeness: high-level view}
\label{sec:outline-proof}

We begin our analysis with fixing an arbitrary finite, \emph{connected}, undirected, simple graph $H = (V(E),E(H))$.\footnote{While our analysis here assumes that $H$ is connected, this is clearly not required for the main result. If $H$ is disconnected then with the coloring trick (cf. Section \ref{subsubsec:coloring-H}), one could have identical analysis and consider all connected components one by one, extending the results to arbitrary, not necessarily connected $H$. We will discuss this in details in Section \ref{sec:dics-many}.}

Our tester of $H$-freeness relies on a simple graph exploration. We first describe our algorithm for testing $H$-freeness of planar graphs with arbitrary degrees and provide the high level structure of its analysis. We defer most of technical details to Sections \ref{sec:exploration-testing-H-free}-- %\ref{subsec:constructing-H1}--%\ref{subsec:Ui->Ui+1},
\ref{sec:final-proof} and Appendix.

Our algorithm relies on a random bounded-breadth bounded-depth search, \RLBFS{} below, which uses \Traverse\,${(G,v,\dg,\ld)}$ from Section \ref{subsec:BFS-like-search-bounded-discs}. (Let us remind, cf. page~\pageref{alg:traverse}, that \Traverse\,${(G,v,\dg,\ld)}$ is similar to BFS of depth $\ld$ starting at vertex $v$, though instead of visiting all neighbors of every vertex, %like in the standard BFS,
one visits only $\dg$ neighbors chosen i.u.r., to limit the complexity of the algorithm.)

%---------------------------------------------------------------------------------------------------------------------------------------------------------
\begin{algo}\label{alg:RLBFS}
\RLBFS\,${(G,\dg,\ld)}$:
\begin{itemize}
\item Pick a random vertex $v \in V$ i.u.r. and run \Traverse\,${(G,v,\dg,\ld)}$.
\end{itemize}
\junk{
\begin{itemize}
\item Pick a random vertex $v \in V$; let $L_0 = \{v\}$.
\item For $\ell = 1$ to $\ld$ do:
    \begin{itemize}[$\diamond$]
    \item Let $L_{\ell} = \emptyset$ and $\mathcal{E}_{\ell} = \emptyset$.
    \item For every $u \in L_{\ell-1}$ do:
        \begin{itemize}[$\circ$]
        \item Choose $\dg$ edges incident to $u$ in $G$ i.u.r.; call them $\mathcal{E}_u$.
        \item Let $\Gamma_u$ be the set of vertices in $\mathcal{E}_u$.
        \item Set $L_{\ell} = L_{\ell} \cup \Gamma_u$ and $\mathcal{E}_{\ell} = \mathcal{E}_{\ell} \cup \mathcal{E}_u$.
        \end{itemize}
    \item $L_{\ell} = L_{\ell} \setminus \bigcup_{i=0}^{\ell-1} L_i$.
    \end{itemize}
\item \textbf{Return} the edges $\bigcup_{\ell=1}^{\ld} \mathcal{E}_{\ell}$.
\end{itemize}
}
\end{algo}
%---------------------------------------------------------------------------------------------------------------------------------------------------------

Our tester \RBE\, runs $f(\eps,H)$ times our search algorithm \RLBFS{} with parameters $\dg = h(\eps,H)$, $\ld = g(\eps,H)$, each time checking if the graph induced by the visited edges contains a copy of $H$, or does not. The algorithm accepts $G$ as $H$-free if and only if all calls found no copy of $H$ in~$G$.

%---------------------------------------------------------------------------------------------------------------------------------------------------------
\medskip
\begin{walgo}
\emph{\textbf{Tester:}} \RBE\,$(G,H,\eps)$: \hfill{\small\it (with three implicit parameters, integer functions $f,g,h$)}
\begin{itemize}
\item Repeat $f(\eps,H)$ times:
    \begin{itemize}[$\diamond$]
    \item Run \RLBFS\,${(G,h(\eps,H),g(\eps,H))}$ and let $\mathcal{E}$ be the resulted set of edges.
    \item If the subgraph of $G$ induced by the edges $\mathcal{E}$ contains a copy of $H$, then \textbf{reject}.
    \end{itemize}
\item If every subgraph explored is $H$-free, then \textbf{accept}.
\end{itemize}
\end{walgo}
%---------------------------------------------------------------------------------------------------------------------------------------------------------

The following main theorem describes key properties of our tester.

\begin{theorem}
\label{thm:main-H-freeness}
Let $H$ be connected. There are positive functions $f$, $g$, $h$, such that for any planar graph $G$:
\begin{itemize}
\item if $G$ is $H$-free, then \RBE$(G,H,\eps)$ accepts $G$, and
\item if $G$ is $\eps$-far from $H$-free, then \RBE$(G,H,\eps)$ rejects $G$ with probability at least $0.99$.
\end{itemize}
\end{theorem}

It is obvious that the first claim holds: if $G$ is $H$-free, then so is every subgraph of $G$, and therefore \RBE{} always accepts. Therefore, to prove our main result,  Theorem \ref{thm:main-H-freeness}, it suffices to show that if $G$ is $\eps$-far from $H$-free, then \RBE{} rejects $G$ with probability at least $0.99$. In view of that, from now on, we assume that the input graph $G$ is $\eps$-far from $H$-free for some constant $\eps>0$.

We note  that it is enough to show that a \emph{single} instance of the random bounded-breadth bounded-depth search (\RLBFS) of breadth $O_{\eps,H}(1)$ and depth $O_{\eps,H}(1)$ finds a copy of $H$ with probability $\Omega_{\eps,H}(1)$. Indeed, for any functions $f$, $g$, and $h$, if
\RLBFS\,${(G,\dg,\ld)}$ with $h(\eps,H) = O_{\eps,H}(1)$ and $g(\eps,H) = O_{\eps,H}(1)$
%a random bounded-breadth bounded-depth search of breadth $h(\eps,H) = O_{\eps,H}(1)$ and depth $g(\eps,H) = O_{\eps,H}(1)$
finds a copy of $H$ with probability at least $5/f(\eps,H) = \Omega_{\eps,H}(1)$, then this implies that $f(\eps,H) = O_{\eps,H}(1)$ independent
calls to \RLBFS\,${(G,\dg,\ld)}$
%random bounded-breadth bounded-depth searches
detect at least one copy of $H$ with probability at least $1-(1-5/f(\eps,H))^{f(\eps,H)} \ge 1-e^{-5} \ge 0.99$. Therefore, in the remainder of the paper, we analyze the following algorithm \RLBD$(G,H,\dg,\ld)$.

%---------------------------------------------------------------------------------------------------------------------------------------------------------
\begin{walgo}
\RLBD\,${(G,H,\dg,\ld)}$:
\begin{itemize}
\item Run \RLBFS\,${(G,\dg,\ld)}$ and let $\mathcal{E}$ be the resulted set of edges.
\item If the subgraph of $G$ induced by the edges $\mathcal{E}$ contains a copy of $H$, then \textbf{reject}.
\item  If not, then \textbf{accept}.
\end{itemize}
\end{walgo}
%---------------------------------------------------------------------------------------------------------------------------------------------------------

We will show the following central technical theorem.

\begin{theorem}
\label{thm:main-H-freeness-single-call}
Let $H$ be a connected undirected graph. There are positive functions $\dg = \dg(\eps,H) = O_{\eps,H}(1)$ and $\ld = \ld(\eps,H) = O_{\eps,H}(1)$ such that for any planar graph $G$ that is $\eps$-far from $H$-free, \RLBD$(G,H,\dg,\ld)$ finds a copy of $H$ with probability $\Omega_{\eps,H}(1)$. The query complexity of \RLBD$(G,H,\dg,\ld)$ is $O(\dg^{\ld}) = O_{\eps,H}(1)$.
\end{theorem}

Since by our discussion above Theorem \ref{thm:main-H-freeness-single-call} yields Theorem \ref{thm:main-H-freeness}, we will focus on proving Theorem \ref{thm:main-H-freeness-single-call}. We also notice that the query complexity of the tester follows directly from its definition, and so we will concentrate on showing that for $\dg = O_{\eps,H}(1)$ and $\ld = O_{\eps,H}(1)$, \RLBD$(G,H,\dg,\ld)$ finds a copy of $H$ with probability $\Omega_{\eps,H}(1)$.

\SArtur{
\begin{remark}
While in our analysis we did not try to optimize the complexity of the tester, focusing on the task of obtaining the query complexity of $O_{\eps,H}(1)$, let us mention that in fact, with the analysis as it is now, without any optimization efforts, the complexity of our tester is $2^{2^{\text{poly}(\Hsize/\eps)}}$. %\Artur{I put a random number; maybe one will want to check it later.}
\end{remark}
}

%---------------------------------------------------------------------------------------------------------------------------------------------------------

\subsection{Outline of the proof of testing $H$-freeness}
\label{subsec:outline-proof}

In this subsection we outline the key ideas behind our proof of testing $H$-freeness. Since the proof is long and complex, we will give here mostly some underlying intuitions, leaving the details to Sections \ref{sec:exploration-testing-H-free}--\ref{sec:final-proof}.

\SArtur{I feel there are too many repetitions, especially since the overview presented in this section is often repeated in the section when one discusses that specific claim/framework. I'm not sure (yet) if it's good or not.}By our discussion above, %in order to prove our main result, Theorem \ref{thm:main-H-freeness-single-call}, that \RLBD\ is a property tester,
it suffices to focus on the case when the input graph $G$ is $\eps$-far from $H$-free. Our analysis relies on the following result (shown in Lemma \ref{lemma:ExistenceOfH}) that every simple planar graph $G$ that is $\eps$-far from $H$-free has a subgraph $\U$ satisfying the following:
\begin{enumerate}[(a)]
\item if \RLBD{($\U,H,\dg,\ld$)} finds a copy of $H$ in $\U$ with probability $\Omega_{\eps,H}(1)$, then \RLBD{($G,H,\dg,\ld$)} finds a copy of $H$ in $G$ with probability $\Omega_{\eps,H}(1)$, and
    %\label{part-a-lemma:ExistenceOfH}
\item \RLBD{($\U,H,\dg,\ld$)} finds a copy of $H$ in $\U$ with probability $\Omega_{\eps,H}(1)$.
    %\label{part-b-lemma:ExistenceOfH}
\end{enumerate}

Our first (and easy) step towards proving this property is to show that $G$ contains a linear number of edge-disjoint copies of $H$ (see Lemma \ref{lemma:edge-disjoint-copies-H-free}). This follows by iteratively removing copies of $H$ and observing that by the definition of being $\eps$-far from $H$-free, we have to remove $\eps n$ edges to make $G$ free of copies of $H$. In the following we will use $\DCH$ to denote a set (of linear size) of edge-disjoint copies of $H$ in $G$. We continue by showing that given $\DCH$, we can compute a subset $\DCH'$ of linear size such that the graph $G[\DCH']$ (subgraph of $G$ on vertex set $V$ and with edge set being the union of the edges of the subgraphs in $\DCH'$) satisfies the first property above. The proof essentially shows that one can remove copies from $\DCH'$ until every vertex in $G[\DCH']$ has degree either $0$ or a small positive constant times its degree in $G$. %On a high level, our proof follows similar ideas to those in \cite{CMOS11}.

Next, we would like to define a sequence of sets $\DCH = \DDCH{0} \supseteq \DDCH{1} \supseteq \dots \supseteq \DDCH{\Hsize}$ with associated hypergraphs with the following interpretation. The hyperedges will be labelled in such a way that we are able to recover the set $\DDCH{i}$ from it. We will use hyperedges to replace certain subgraphs of $H$ and their corresponding part in $G$.

%---------------------------------------------------------------------------------------------------------------------------------------------------------

\paragraph{Hyperedges.}
We will now describe the use of hyperedges as replacements for copies of subgraphs of $H$ in $G$. Let $G^*$ be a subgraph of $G$ that has a copy of $H$. Consider a subgraph $H_1$ of $H$ and let $u_1, \dots, u_{\ell}$ be the vertices in the copy of $H_1$ in $G^*$ that separate $G \setminus H_1$ from $H_1$, so that (cf. Figure \ref{fig:ex-gadgets-1}):
%
%\begin{enumerate}[\sl (a)]
%\begin{inparaenum}[\sl (a)]
\begin{compactenum}[\sl \qquad (a)]
\item every vertex from $\{u_1, \dots, u_{\ell}\}$ is adjacent in $G^*$ to some vertex $H_1 \setminus \{u_1, \dots, u_{\ell}\}$,
\item every vertex in $H_1 \setminus \{u_1, \dots, u_{\ell}\}$ is adjacent in $G^*$ only to vertices from $H_1$, and
\item $\{u_1, \dots, u_{\ell}\}$ forms an independent set in $H_1$.
\end{compactenum}
%\end{inparaenum}
%\end{enumerate}
%
Then, we can construct a gadget to represent that copy of $H_1$ by removing from $H_1$ all vertices and edges from $H_1 \setminus \{u_1, \dots, u_{\ell}\}$ and replacing them by a single \emph{hyperedge} $\{u_1, \dots, u_{\ell}\}$.

\begin{figure}[t]
\centerline
{(a) \includegraphics[width=.3\textwidth]{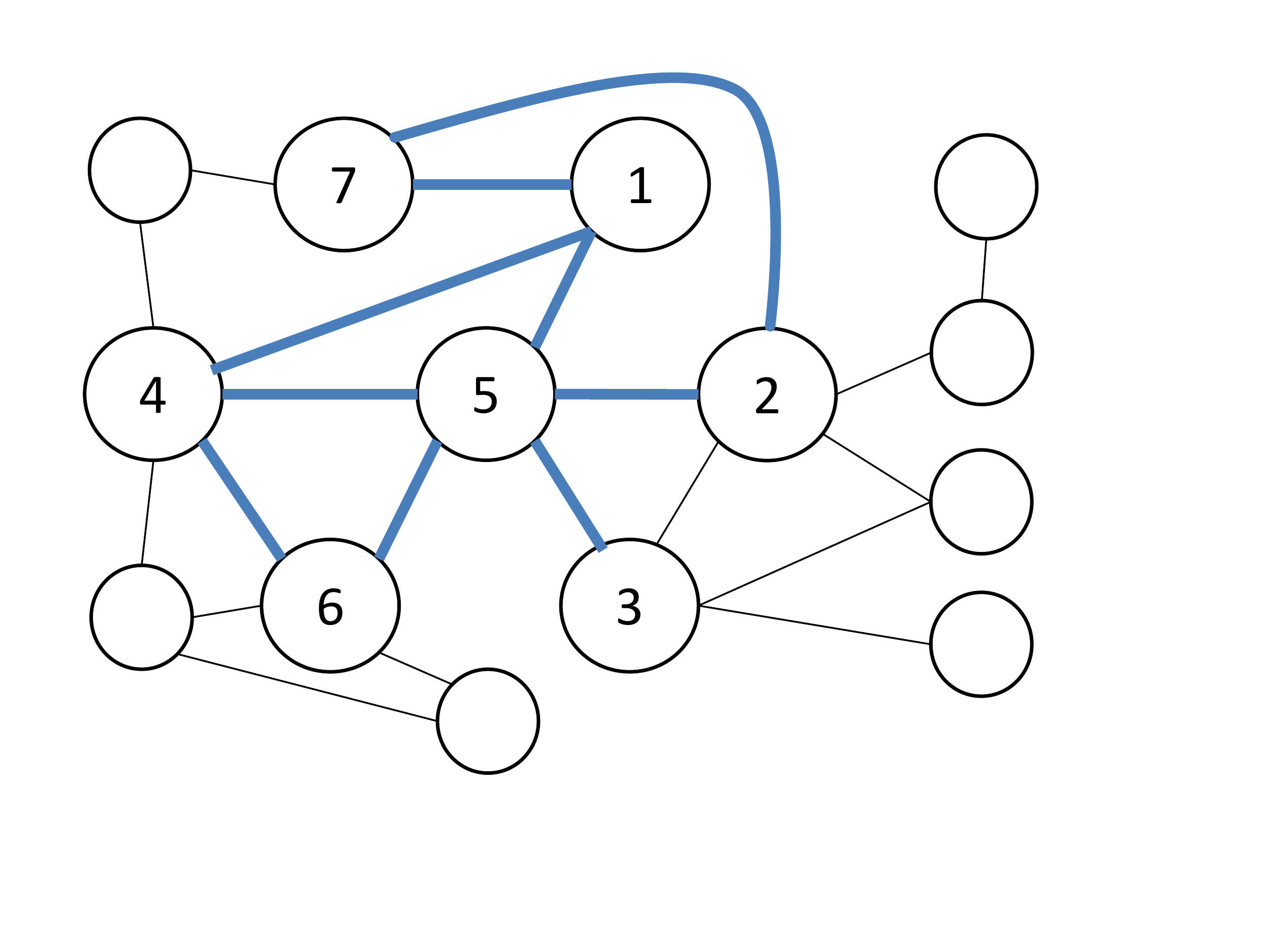}
(b) \includegraphics[width=.3\textwidth]{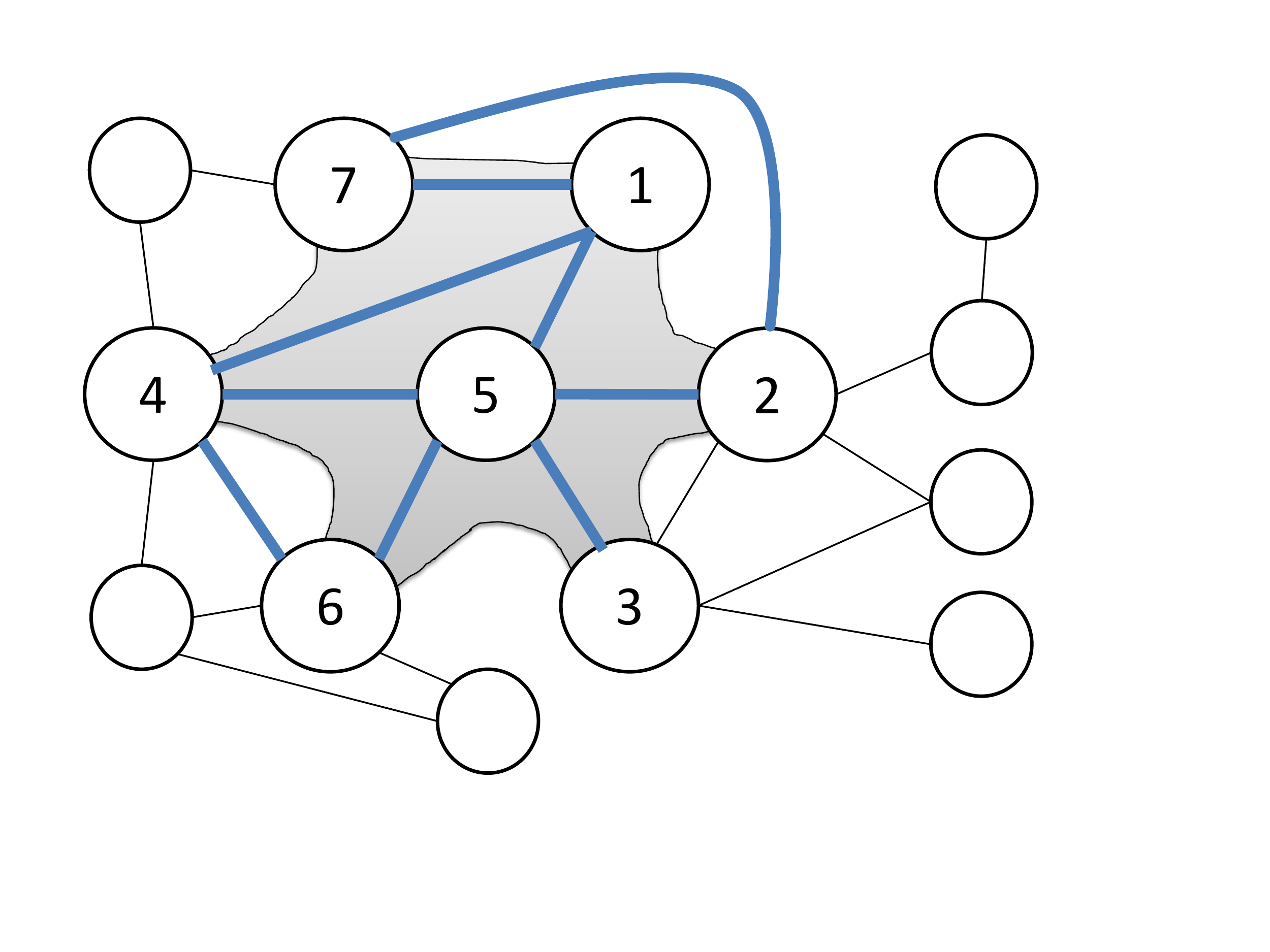}
(c) \includegraphics[width=.3\textwidth]{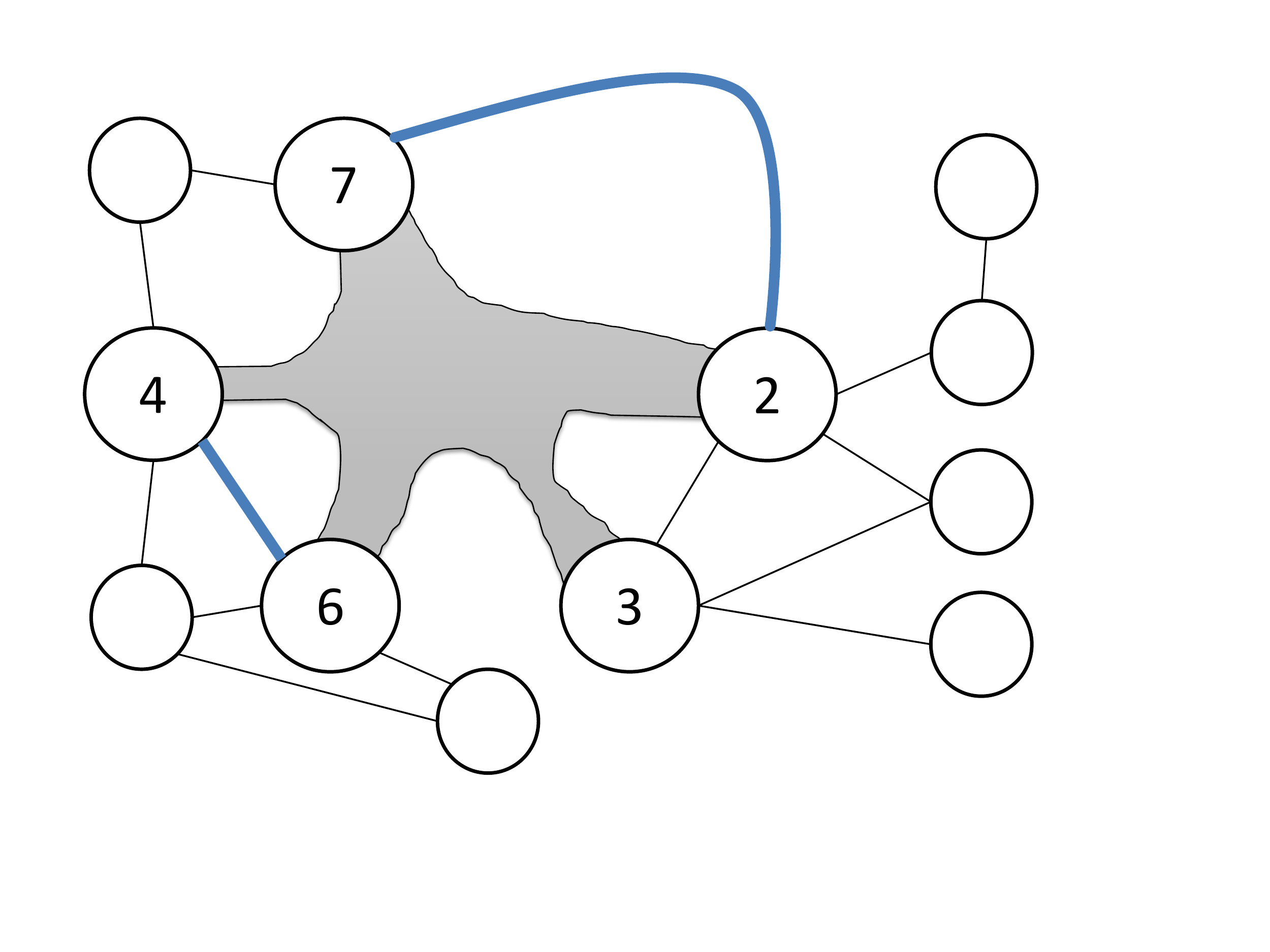}
}
\caption{\small (a) Consider a part of the input graph $G$ with numbered vertices corresponding to the colored vertices in a copy of $H$ in $G$ and thick edges corresponding to the edges in that copy of $H$.
(b) We have a subgraph $H_1$ of $H$ consisting of the vertices and edges marked by the grey area, with vertices $\{2,3,4,6,7\}$ separating $H_1$ from the rest of $G$.
(c) The gadget obtained by removing internal vertices $\{1,5\}$ and replacing $H_1$ by a hyperedge connecting vertices $\{2,3,4,6,7\}$.}
\label{fig:ex-gadgets-1}
\end{figure}

We will encode the structural information of the subgraph replaced by the hyperedge in a label, so that it may happen that we have parallel hyperedges with different labels. In addition to the above structural role we recall from the previous section that the idea of hyperedges was to encode that whenever (a hypergraph version of) \RLBFS\ enters the hyperedge then it will reach all its vertices. Our final goal will be to construct a hypergraph that only consists of selfloops, so that we can argue easily that our tester finds a copy of $H$ by finding a corresponding set of labelled selfloops.

%---------------------------------------------------------------------------------------------------------------------------------------------------------

\paragraph{Vertex coloring.}
A major difficultly in applying our approach is to find subgraphs that can be replaced. One way to simplify this question is to color both the vertices of $H$ and the vertices of $G$ with $\Hsize$ colors, such that every vertex of $H$ receives a distinct color and every copy of $H$ in $\DCH$ has the same coloring as $H$. We show in Lemma \ref{lemma:edge-disjoint-copies-H-free} that there is a coloring $\chi$ of $G$ and $H$ such that $G$ contains a set $\DCH$ containing a linear number of such edge-disjoint colored copies of $H$. An important feature of this coloring, which will be very useful in finding vertices that can be replaced by hyperedges, is that \emph{every vertex has the same role in all subgraphs from $\DCH$} it is contained in.

%---------------------------------------------------------------------------------------------------------------------------------------------------------

\paragraph{Getting from $\DDCH{i}$ to $\DDCH{i+1}$.}
Next we describe how we move from the set $\DDCH{i}$ to $\DDCH{i+1}$. This is the main step in our reduction and it will be partitioned in a number of substeps. We start with an overview. In each round we perform the following high level process:

%---------------------------------------------------------------------------------------------------------------------------------------------------------
\begin{walgo}\vspace*{-0.3in}
\begin{itemize}[$\bullet$]
\item Select a vertex $v_i \in V(H)$.
\item Simultaneously, contract every vertex $u \in V(\HQ_i(\DDCH{i}))$ with $\chi(u) = \chi(v_i)$ as follows:
\begin{itemize}[$\diamond$]
\item for every colored copy $\h$ of $H$ in $\DDCH{i+1}$ that contains vertex $u$:
    \begin{itemize}[$\circ$]
    \item add a new hyperedge consisting of vertices in $\N_i^{\h}\langle u \rangle$, where $\N_i^{\h}\langle u \rangle$ is the set of neighbors of $u$ in $\h$ (in the corresponding hypergraph) other than $u$ (that is, $u \notin \N_i^{\h}\langle u \rangle$);
    \end{itemize}
\item remove vertex $u$ (with all incident edges from $\HQ_i(\DDCH{i})$).
\end{itemize}
\end{itemize}
\end{walgo}
%---------------------------------------------------------------------------------------------------------------------------------------------------------

We remark that our algorithm above ensures that no neighboring vertices are contracted since the coloring $\chi$ has no monochromatic edges. This follows from the fact that every edge in $G[\DCH]$ belongs to some copy of $H$ and the coloring of $H$ has no monochromatic edge. Thus, we can perform the contractions independently.

In our construction we will require that the contracted vertices additionally satisfy some stronger properties. This is to maintain (approximately) some
basic properties of planar graphs.
\begin{itemize}
\item We want to ensure that all contractions in $\HQ_i(\DDCH{i})$ corresponding to the contraction of $v_i$ are \emph{consistent}, that is, the contraction of $u$ is the same in every colored copy of $H$ that contains~$u$ (that is, for every vertex $u$ in with $\chi(u) = \chi(v_i)$, for any two colored copies $\h_1,\h_2$ of $H$ in $\DDCH{i+1}$ containing vertex $u$, we have $\N_i^{\h_1}\langle u \rangle = \N_i^{\h_2}\langle u \rangle$).
\end{itemize}
The required property is captured in the following definition (see also Definition \ref{def:safe-vertices}).

\begin{definition}\textbf{(Safe vertices])}
Let $\DDCH{i}$ be a set of edge-disjoint colored copies of $H$ in $G$ and let $\DCH \subseteq \DDCH{i}$. We call a vertex $u$ \textbf{\emph{safe}} if for all colored copies $\h \in \DCH$ of $H$ that contain $u$, the sets $\N_i^{\h}\langle u\rangle$ are the same.
\end{definition}

%---------------------------------------------------------------------------------------------------------------------------------------------------------

\paragraph{Finding safe vertices.}
Our next challenge is to show that we can find many (a linear number) safe vertices of the same color. In order to do so, we will delete elements from the current set $\DDCH{i}$ in a controlled way until we can guarantee that many safe vertices of the same color exist. An important concept that we define here is that of a \emph{shadow graph}. A shadow graph is a union of $\Hsize$ planar graphs and it models the neighborhood relation of our hypergraph, such that two vertices are adjacent in the shadow graph if and only if they belong to the same edge in the hypergraph. The main use of shadow graphs is to show in the upcoming construction that our hypergraph still satisfies some near-planar properties that will be useful in the analysis. The concept of shadow graphs and the proof of their existence is one of the main new ideas in this paper.

Using the existence of shadow graphs, we can properly implement the process of contractions via hyperedges, proceed similarly as in an earlier paper about testing of bipartiteness in planar graphs \cite{CMOS11}, where the shadow graphs guarantee that we still approximately satisfy the properties of planar graphs that were used the previous paper \cite{CMOS11}: We first prove that we can construct a subset $\DCH$ of $\DDCH{i}$ of linear size such that every copy of $H$ in $\DCH$ has a vertex of constant degree in $G[\DCH]$. Then we use this claim in the proof of Lemma \ref{lemma:many-safe} to show how to construct a subset $\DCH^*$ of $\DDCH{i}$ such that every copy of $H$ in $\DCH^*$ contains a safe vertex.

%---------------------------------------------------------------------------------------------------------------------------------------------------------

\paragraph{Wrapping things up.}
What remains to do is to prove that our construction satisfies the second required property of our tester:

\begin{algo}\vspace*{-0.25in}
\begin{itemize}[$\otimes$]
\item \RLBD{($\G[\DCH],H,\dg,\ld$)} finds a copy of $H$ in $\G[\DCH]$ with probability $\Omega_{\eps,H}(1)$.
\end{itemize}
\end{algo}

We define $\DCH$ to be the set $\DDCH{\Hsize}$ obtained in the final round of our reduction.

We will then prove $\otimes$ by showing the following two properties (proven in Claims \ref{claim:final-prop1} and \ref{claim:final-prop2}), where, informally, \HRLBD\ is an extension of \RLBD\ to hypergraphs, $\HQ_{i}(\DCH)$ denotes the hypergraph corresponding to $\DDCH{i}$, $\M_i$ is the hypergraph corresponding to $H$ in round $i$, and $\P_i$ corresponds to the function assigning vertices contracted in the process to their representatives in the hypergraph $\HQ_{i}(\DCH)$.%\Artur{Here \HRLBD\ is undefined!!!}

%---------------------------------------------------------------------------------------------------------------------------------------------------------
\begin{walgo}\vspace*{-0.2in}
\begin{enumerate}[1.]
\item %(Claim \ref{claim:final-prop1}:)
    %for some $\dg, \ld = \Theta_{\eps,H}(1)$,
    the probability that \HRLBD\,${(\HQ_{\Hsize}(\DCH), \P_{\Hsize}, \M_{\Hsize}, \Hsize^2, 1)}$ finds a copy of $\M_{\Hsize}$ is $\Omega_{\eps,H}(1)$, and
    %\label{final-prop1}
\item %(Claim \ref{claim:final-prop2}:)
    for every $i$, $1 \le i < \Hsize$, %and any $\dg, \ld$,
\begin{itemize}
\item if the probability that
    %\label{final-prop2}
    \HRLBD\,${(\HQ_{i+1}(\DCH),\P_{i+1},\M_{i+1},\dg,\ld)}$ finds a copy of $\M_{i+1}$ is $\Omega_{\eps,H}(1)$,
\item then the probability that \HRLBD\,${(\HQ_i(\DCH),\P_i,\M_i,\Hsize \cdot \dg, 2 \ld)}$ finds a copy of $\M_i$ is $\Omega_{\eps,H}(1)$.
\end{itemize}
\end{enumerate}
\end{walgo}
%---------------------------------------------------------------------------------------------------------------------------------------------------------

The proof of Claim \ref{claim:final-prop1} now exploits that $\M_{\Hsize}$ only consists of selfloops, which can easily be found and the proof of Claim \ref{claim:final-prop2} formalizes our idea that if our random walk enters a hyperedge in $\HQ_{i+1}(\DCH)$ then we perform with constant probability the same operation in $\HQ_{i+1}(\DCH)$ in two steps of our randomized process. Combining the results with our previous considerations yields our main statement: $H$-freeness in planar graphs is constant query-time testable.

%---------------------------------------------------------------------------------------------------------------------------------------------------------

%\section{Exploring graphs to test $H$-freeness: \RBE}
%\label{sec:exploration-testing-H-free}

\junk{
\Artur{Notation:
    \begin{compactitem}
    \item \texttt{$\backslash$eps} denotes $\eps$
    \item \texttt{$\backslash$RBE} denotes \RBE
    \item \texttt{$\backslash$RLBFS} denotes \RLBFS
    \item \texttt{$\backslash$RLBD} denotes \RLBD
    \item \texttt{$\backslash$ld} denotes $\ld$ (depth of BFS)
    \item \texttt{$\backslash$dg} denotes $\dg$ (breadth of BFS)
    \end{compactitem}
}
}
%

%---------------------------------------------------------------------------------------------------------------------------------------------------------

\section{Analysis of \RLBD{} when $G$ is $\eps$-far from $H$-free}
\label{sec:exploration-testing-H-free}

\junk{
\Artur{Notation:
    \begin{compactitem}
    \item \texttt{$\backslash$G[X]} denotes $\G[X]$, the subgraph of $G$ induced by $X$
    \item \texttt{$\backslash$U} denotes $\U$ (also in $\U_1, \dots, \U_k$, with $G \supseteq \U_1 \supseteq \U_2 \supseteq \dots \supseteq \U_k$); note that $\U_i = \G[\DDCH{i}]$
    \item \texttt{$\backslash$DDCH\{i\}} denotes $\DDCH{i}$ (edge-disjoint copies of $H$ in $G$, $\DDCH{1} \supseteq \DDCH{2} \supseteq \dots \supseteq \DDCH{k}$)
    \item \texttt{$\backslash$HQ} denotes $\HQ$ (also in $\HQ_i$, which is a hypergraph constructed from $\DDCH{i}$ by contracting vertices of $H$ in a specific, consistent way (cf. Section \ref{subsec:shrinking-copies-of-H}))
    \item \texttt{$\backslash$M} denotes $\M$ (also in $\M_1, \M_2, \dots, \M_{\Hsize}$, which is a sequence of hypergraphs that are ``shrunk'' copies of $H$, each $\M_i$ with $\Hsize-i+1$ vertices)
    \item \texttt{$\backslash$Hsize} denotes $\Hsize$
    \end{compactitem}
}
}
Because of the arguments from the previous section, the remainder of the paper deals with the main technical challenge of our result: proving Theorem \ref{thm:main-H-freeness-single-call} that in any simple planar graph $G=(V,E)$ that is $\eps$-far from $H$-free, our algorithm \RLBD{} finds with sufficient probability a copy of $H$.

Our analysis relies on the following lemma showing the existence of a special subgraph $\U$ of~$G$:

\begin{lemma}
\label{lemma:ExistenceOfH}
For every $\eps \in (0,1)$, there are $\dg = \dg(\eps,H)$ and $\ld = \ld(\eps,H)$, such that for every simple planar graph $G = (V,E)$ that is $\eps$-far from $H$-free, there is a subgraph $\U$ of $G$ with the following properties:
\begin{enumerate}[(a)]
\item if \RLBD{($\U,H,\dg,\ld$)} finds a copy of $H$ in $\U$ with probability $\Omega_{\eps,H}(1)$, then \RLBD{($G,H,\dg,\ld$)} finds a copy of $H$ in $G$ with probability $\Omega_{\eps,H}(1)$, and
    \label{part-a-lemma:ExistenceOfH}
\item \RLBD{($\U,H,\dg,\ld$)} finds a copy of $H$ in $\U$ with probability $\Omega_{\eps,H}(1)$.
    \label{part-b-lemma:ExistenceOfH}
\end{enumerate}
\end{lemma}

Observe that if such a subgraph $\U$ as promised in Lemma \ref{lemma:ExistenceOfH} always exists, then these properties immediately imply that \RLBD{$(G,H,\dg,\ld)$} finds a copy of $H$ in $G$ with probability $\Omega_{\eps,H}(1)$ and therefore, by the discussion above, Theorems \ref{thm:main-H-freeness} and \ref{thm:main-H-freeness-single-call} follow.

In order to prove Lemma \ref{lemma:ExistenceOfH}, we will show that for any simple planar graph $G$ that is $\eps$-far from $H$-free, there exists a set $\DCH$ of edge-disjoint copies of $H$ in $G$ for which $\G[\DCH]$, the subgraph of $G$ induced by the edges of $\DCH$,
%\footnote{Throughout the paper, for any set of edge-disjoint subgraphs $\mathcal{S}$ of $G$, we write $\G[\mathcal{S}]$ to denote the graph with vertex set $V$ and edge set being the set of edges from the sets in $\mathcal{S}$.},
satisfies the properties of graph $\U$ in Lemma \ref{lemma:ExistenceOfH}. The construction of the set $\DCH$ and the analysis of its properties form the main technical contribution of our paper. While part (\ref{part-a-lemma:ExistenceOfH}) in Lemma \ref{lemma:ExistenceOfH} is rather easy to achieve and to analyze (thanks to Lemma \ref{lemma:transformation} in Section \ref{subsubsec:condition-a-prime}), the main challenge of our construction is in ensuring part (\ref{part-b-lemma:ExistenceOfH}) in Lemma \ref{lemma:ExistenceOfH}. For that, we use a rather elaborate construction to gradually find a sequence $\DDCH{1} \supseteq \DDCH{2} \supseteq \dots \supseteq \DDCH{\Hsize}$ of sets of edge-disjoint copies of $H$ in $G$, with $|\DDCH{\Hsize}| = \Omega_{\eps,H}(|V|)$, such that the final set $\DDCH{\Hsize}$ is the set $\DCH$ that defines $\U = \G[\DDCH{\Hsize}]$ in Lemma~\ref{lemma:ExistenceOfH}.

The construction of the sequence $\DDCH{1} \supseteq \DDCH{2} \supseteq \dots \supseteq \DDCH{\Hsize}$ of sets of edge-disjoint copies of $H$ in $G$, with $|\DDCH{\Hsize}| = \Omega_{\eps,H}(|V|)$, for which we could easily argue that \RLBD{($\G[\DDCH{\Hsize}],H,\dg,\ld$)} finds a copy of $H$ in $\G[\DDCH{\Hsize}]$ with probability $\Omega_{\eps,H}(1)$, is the most challenging and technical contribution of our paper. We begin with a simple construction of $\DDCH{1}$ which is a set of $\Omega_{\eps,H}(|V|)$ edge-disjoint copies of $H$ in $G$ (cf. Lemma \ref{lemma:edge-disjoint-copies-H-free}). Then our construction is iterative: we design a reduction that takes set $\DDCH{i}$ of $\Omega_{\eps,H}(n)$ edge-disjoint copies of $H$ and we construct from it another set $\DDCH{i+1} \subseteq \DDCH{i}$ with $|\DDCH{i+1}| = \Omega_{\eps,H}(|\DDCH{i}|)$ for which we simplify the structure of $\G[\DDCH{i+1}]$ with respect to that of $\G[\DDCH{i}]$. To guide our process, we associate with each $\DDCH{i}$ a certain \emph{hypergraph} $\HQ_i(\DDCH{i})$ that is constructed from $\DDCH{i}$ by contracting vertices of $H$ in a specific, consistent way (cf. Section \ref{subsec:shrinking-copies-of-H}). \emph{The purpose of $\HQ_i(\DDCH{i})$ is to model the copies of $H$ by a hypergraph on a smaller number of vertices, by contracting vertices (and incident edges) which are known to be visited by \RLBFS{} via other means.} %, though a single step of a (weighted) random bounded-BFS in $\HQ_i$ (cf. Section \ref{sec:exploration-in-hypergraphs}) will correspond to multiple steps in \RLBD{} for $\DDCH{i}$.
We will construct a sequence of hypergraphs $\HQ_1(\DDCH{1}), \HQ_2(\DDCH{2}), \dots, \HQ_{\Hsize}(\DDCH{\Hsize})$ that correspond to sets $\DDCH{1}, \DDCH{2}, \dots, \DDCH{\Hsize}$, and a sequence of hypergraphs $\M_1, \M_2, \dots, \M_{\Hsize}$ that are ``shrunk'' copies of $H$, each $\M_i$ with $\Hsize-i+1$ vertices, such that, informally, for our algorithm of selecting $\DDCH{1}, \DDCH{2}, \dots, \DDCH{\Hsize}$, the following conditions holds:
\begin{itemize}
\item the probability of finding by \RLBFS{} a copy of $H$ in $\G[\DDCH{1}]$ is the same as the probability of finding by \RLBFS{} a copy of $\M_1$ in $\HQ_1(\DDCH{1})$,
\item the probability of finding by \RLBFS{} a copy of $\M_{i+1}$ in $\HQ_{i+1}(\DDCH{i+1})$ is similar to the probability of finding by \RLBFS{} a copy of $\M_i$ in $\HQ_i(\DDCH{i})$, and
\item using the fact that $\M_{\Hsize}$ has a single vertex, one can easily estimate the probability of finding by \RLBFS{} a copy of $\M_{\Hsize}$ in $\HQ_{\Hsize}(\DDCH{\Hsize})$.
\end{itemize}
With these three properties at hand, the main theorem will follow.

One central feature of our analysis via the study of hypergraphs is to ensure that the underlying hypergraphs have some basic planar graphs-like properties. (In particular, informally, in our analysis we would like to argue that there is always a constant fraction of low-degree vertices.) While we do not have a useful characterization of planar hypergraphs, we will be able to model some planarity-like properties of the hypergraphs using some special graph reduction (via \emph{shadow graphs}), see Lemma \ref{lemma:central-small-degrees} and Appendix \ref{subsec:planarization-of-hypergraphs}.

In the following sections we will develop this framework in details, finalizing it in Section \ref{sec:final-proof} that proves the desired properties above.

%---------------------------------------------------------------------------------------------------------------------------------------------------------

\subsection{Auxiliary technical tools}
\label{subsec:auxiliary-tools}

We begin with three auxiliary tools in our analysis, the study of the problem of finding \emph{colored} copies of $H$ in $G$ (Section \ref{subsubsec:coloring-H}), a reduction simplifying condition (a) of Lemma \ref{lemma:ExistenceOfH} (Section \ref{subsubsec:condition-a-prime}), and extension of the testing and graph exploration framework to hypergraphs (Section \ref{sec:exploration-in-hypergraphs}).

%---------------------------------------------------------------------------------------------------------------------------------------------------------

\subsubsection{Auxiliary tools: Finding \emph{colored copies} of $H$ in $G$}
\label{subsubsec:coloring-H}

To simplify the analysis, we will %follow the approach from \Artur{Legacy. Will require some major rewriting.}Section~\ref{section-tree-testing} and
consider colored copies of $H$ in $G$. Let us \emph{color} all vertices of $H$ using $\Hsize$ colors, one color for each vertex (without loss of generality, the colors are $\{1, 2, \dots, \Hsize\}$). While the coloring is not needed by the algorithm, it will simplify the analysis. With this in mind, instead of showing that our algorithm \RLBD{} finds with sufficient probability a copy of $H$, we will show (cf. Lemma~\ref{lemma:edge-disjoint-copies-H-free}) that there is a coloring $\chi$ of vertices of $G$ such that \RLBD{} finds (with sufficient probability) a colored copy of $H$, that is, a copy of $H$ in $G$ with colors of the vertices in the copy consistent with the coloring $\chi$. (While this statement sounds trivial, since once we found a copy of $H$ in $G$ we can always color vertices of $G$ to be consistent with the coloring of $H$, the colors will be helpful in our analysis.) Therefore, from now on, whenever we will aim to find a copy of $H$ we will mean to find a colored copy of $H$ consistent with given coloring~$\chi$.

Let us notice one immediate implication of this assumption: if $\DDCH{i}$ and $\chi$ are fixed, then one can think about every edge $e$ as a \emph{labeled edge}, since the colors of its endpoints define a unique edge in $H$ that $e$ corresponds too. We will use this property implicitly throughout the paper, without mentioning it anymore.

%---------------------------------------------------------------------------------------------------------------------------------------------------------

\subsubsection{Auxiliary tools: Simplifying condition (\ref{part-a-lemma:ExistenceOfH}) of Lemma \ref{lemma:ExistenceOfH}: (via edge-disjoint copies of $H$)}
\label{subsubsec:condition-a-prime}

We show that one can simplify condition (\ref{part-a-lemma:ExistenceOfH}) of Lemma \ref{lemma:ExistenceOfH} for the special case when the subgraph $\U$ of $G$ is a union of a linear number of edge-disjoined colored copies of $H$ (a similar approach has been also used in \cite{CMOS11}). That is, if there is a graph $\G[\DCH]$ with a \emph{linear number of edge-disjoint colored copies of $H$}, then Lemma \ref{lemma:transformation} shows that there is always a subset $\DCH' \subseteq  \DCH$ with cardinality $|\DCH'| = \Omega_{\eps,H} (|\DCH|)$ such that the graph $\G[\DCH']$ satisfies property (\ref{part-a-lemma:ExistenceOfH}).

\begin{restatable}{lemma}{LPaI}
\emph{\textbf{(Transformation to obtain property (\ref{part-a-lemma:ExistenceOfH}))}}
\label{lemma:transformation}
Let $G = (V,E)$ be a simple planar graph. Let $\DCH$ be a set of $\Omega_{\eps,H}(|V|)$ edge-disjoint colored copies of $H$ in $G$. Then there exists a subset $\DCH' \subseteq \DCH$, $|\DCH'| = \Omega_{\eps,H}(|V|)$, such that the graph $\G[\DCH']$ satisfies condition (\ref{part-a-lemma:ExistenceOfH}) of Lemma \ref{lemma:ExistenceOfH}.
\end{restatable}

The proof of Lemma \ref{lemma:transformation}, as a natural extension of the approach from \cite{CMOS11}, is deferred to Appendix \ref{sec:condition-a-prime}.

%---------------------------------------------------------------------------------------------------------------------------------------------------------

\subsubsection{Traversing hypergraphs and testing hypergraph $\M$-freeness}
\label{sec:exploration-in-hypergraphs}

\junk{
\Artur{Notation:
    \begin{compactitem}
    \item \texttt{$\backslash$HRLBFS} denotes \HRLBFS
    \item \texttt{$\backslash$HRLBD} denotes \HRLBD
    \item \texttt{$\backslash$P} denotes representative function $\P: V \rightarrow V$
    \end{compactitem}
}
}
In Section \ref{sec:exploration-testing-H-free}, we described two central algorithms used for testing $H$-freeness:
\RLBFS, %\RLBFS\,${(G,\dg,\ld)}$
%\RBE, %\RBE\,$(G,H,\eps)$,
and \RLBD.
%\RLBD\,${(G,H,\dg,\ld)}$.
Both these algorithms were presented in a form required to test $H$-freeness in a graph. However, in our transformations we will apply the same algorithms to hypergraphs, to test whether a hypergraph $\HQ$ (in a form of $\HQ_i(\DDCH{i})$, as defined in Section \ref{subsec:shrinking-copies-of-H}) is $\M$-free, where $\M$ is a fixed hypergraph (which in our applications will be $\M_i$, as defined in Section \ref{subsec:shrinking-H}). While the modifications are rather straightforward, for the sake of completeness, we will describe below these algorithms to be run on a hypergraph. Furthermore, in our algorithms for hypergraphs we will have one additional parameter, a \emph{representative function} $\P: V \rightarrow V$, which describes the way how the edges have been contracted (cf. Definition \ref{def:canonical-representative-function} and Appendix~\ref{subsec:planarization-of-hypergraphs}). %\Artur{Double check if it makes sense and if it's the right reference.}.
The idea behind the representative function $\P$ is that any vertex $u$ that either is in the hypergraph $\HQ$ or which does not belong to any set of copies of $H$ has $\P(u) = u$, but any other vertex $u$ from $G$ that has been contracted and now does not appear in $\HQ$, has $\P(u)$ equal to its representative in $\HQ$. In the latter case, the intuition is that the representative is a vertex in $\HQ$ that with probability $\Omega_{\eps,H}(1)$ can be reached from $u$ in $O_{\eps,H}(1)$ steps, if \RLBFS{} (run in $G$) started at~$u$.

\begin{remark}\label{remark:use-of-V-in-H}
Let us remark that in \HRLBD{} and \HRLBFS{} below we use the input graph $G$ implicitly, since in \HRLBFS{} we directly refer here to the set $V$, which is the vertex set of $G$, and we do so indirectly via the use of $\P$, whose domain and range are $V$. Further, in our applications we will always have that $V(\HQ) \subseteq V$.
\end{remark}

%---------------------------------------------------------------------------------------------------------------------------------------------------------
\begin{walgo}
\HRLBFS\,${(\HQ,\P,\dg,\ld)}$:
\begin{itemize}
%\item Pick a random vertex $v \in V$, such that any $v \in V$ is chosen with probability $\frac{|P^{(-1)}{}^*(v)|}{|V|}$; let $L_0 = \{v\}$.
\item Pick a vertex $v \in V$ i.u.r., and let $L_0 = \{\P(v)\}$
    \\
    {\small (i.e., $L_0$ has a randomly selected vertex, such that any $u \in V$ is chosen with probability $\frac{|\P^{(-1)}(u)|}{|V|}$)}.
\item If $v$ is a vertex of $\HQ$ then for $\ell = 1$ to $\ld$ do:
    \begin{itemize}[$\diamond$]
    \item Let $L_{\ell} = \emptyset$ and $\mathcal{E}_{\ell} = \emptyset$.
    \item For every $u \in L_{\ell-1}$ do:
        \begin{itemize}[$\circ$]
        \item Choose $\dg$ edges incident to $u$ in $\HQ$ i.u.r.; call them $\mathcal{E}_{\ell,u}$.
        \item Let $\Gamma_u$ be the set of vertices in $\mathcal{E}_{\ell,u}$.
        \item Set $L_{\ell} = L_{\ell} \cup \Gamma_u$ and $\mathcal{E}_{\ell} = \mathcal{E}_{\ell} \cup \mathcal{E}_{\ell,u}$.
        \end{itemize}
    \item $L_{\ell} = L_{\ell} \setminus \bigcup_{i=0}^{\ell-1} L_i$.
    \end{itemize}
\item \textbf{Return} the edges $\bigcup_{\ell=1}^{\ld} \mathcal{E}_{\ell}$.
\end{itemize}
\end{walgo}
%---------------------------------------------------------------------------------------------------------------------------------------------------------

%---------------------------------------------------------------------------------------------------------------------------------------------------------
\begin{walgo}
\HRLBD\,${(\HQ,\P,\M,\dg,\ld)}$:
\begin{itemize}
\item Run \HRLBFS\,${(\HQ,\P,\dg,\ld)}$ and let $\mathcal{E}$ be the resulted set of edges.
\item If the sub-hypergraph of $\HQ$ induced by the edges $\mathcal{E}$ contains a copy of $\M$, then \textbf{reject}.
\item  If not, then \textbf{accept}.
\end{itemize}
\end{walgo}
%---------------------------------------------------------------------------------------------------------------------------------------------------------

\section{Finding the first set $\DDCH{1}$ of edge-disjoint colored copies of $H$}
\label{subsec:constructing-H1}

We now proceed with a simple construction that for a given graph $G$ that is $\eps$-far from $H$-free, finds a set $\DDCH{1}$ of $\Omega_{\eps,H}(|V|)$ edge-disjoint colored copies of $H$ in $G$.

\begin{lemma}
\label{lemma:edge-disjoint-copies-H-free}
If $G$ is $\eps$-far from $H$-free, then one can color vertices of $G$ with $\Hsize$ colors $\chi$ such that $G$ has a set $\DCH$ of at least %$\lceil\frac{\eps \cdot |V|}{|E(H)|}\rceil \cdot \frac{1}{\Hsize^{\Hsize}}$
$\frac{\eps}{|E(H)| \cdot \Hsize^{\Hsize}} \cdot |V|$ edge-disjoint colored copies of $H$.
\end{lemma}

\begin{proof}
%\Artur{It is misleading to compare this proof to the proof from \cite{CMOS11}. In \cite{CMOS11} we needed planarity, here we don't.}
%We follow the simple arguments developed earlier in \cite{CMOS11} for testing bipartiteness, extending them to deal with colored vertices and colored copies of arbitrary graph $H$.
%
We first find the copies of $H$ without considering the coloring of $V$ and $V(H)$, and then we will prove the existences of the relevant coloring~$\chi$.

We find edge-disjoint copies of $H$ in $G$ one by one. Suppose that we have already found in $G$ a set of $k$ edge-disjoint copies of $H$, where $k < \frac{\eps |V|}{|E(H)|}$. Then, since $G$ is $\eps$-far from $H$-free, the graph obtained from $G$ by removal of the $k$ copies of $H$ found already (which removes $k |E(H)| < \eps |V|$ edges from $G$) cannot be $H$-free, and hence $G$ must contain a copy of $H$. This copy would be edge-disjoint with all copies found before, what by induction shows that $G$ has at least $\frac{\eps \cdot |V|}{|E(H)|}$ edge-disjoint copies of $H$.

%\Artur{For our record, we first used it in the legacy Lemma \ref{lemma:edge-disjoint-colored-copies-T} (Section \ref{subsec:traverse-G-for-T-planar}, might be commented now).}
Let $H_1, \dots, H_{\ell}$ be the edge-disjoint copies of $H$ in $G$, with $\ell \ge \frac{\eps \cdot |V|}{|E(H)|}$. Let us consider a uniformly random coloring of vertices of $G$ (with $\Hsize$ colors) and let $X_i$ be the indicator random variable that $H_i$ has all vertices of the same color as in $H$; let $X = \sum_{i=1}^{\ell} X_i$. Clearly, for every $i$, $\Pr{X_i = 1} = \Ex{X_i} = \frac{1}{\Hsize^{\Hsize}}$. Therefore, $\Ex{X} = \Ex{\sum_{i=1}^{\ell} X_i} = \sum_{i=1}^{\ell} \Ex{X_i} = \frac{\ell}{\Hsize^{\Hsize}}$. This implies that there is a coloring of vertices of $G$ that has at least $\frac{\ell}{\Hsize^{\Hsize}} \ge %\frac{\lceil\frac{\eps \cdot |V|}{|E(H)|}\rceil}{\Hsize^{\Hsize}} \ge
\frac{\eps \cdot |V|}{|E(H)| \cdot \Hsize^{\Hsize}}$ edge-disjoint colored copies of $H$. Therefore, there is a coloring $\chi$ with this property, that is, after we color vertices of $G$ using $\chi$, then $G$ will have at least $\frac{\eps \cdot |V|}{|E(H)| \cdot \Hsize^{\Hsize}}$ edge-disjoint colored copies of $H$ that form the required set $\DCH$.
\end{proof}

Using the result from Lemma \ref{lemma:edge-disjoint-copies-H-free}, from now on, %throughout the entire paper,
we will assume that the vertices of $G$ are colored using $\chi$ (the coloring from Lemma \ref{lemma:edge-disjoint-copies-H-free}) so that $G$ has at least $\Omega_{\eps,H}(|V|)$ edge-disjoint colored copies of~$H$.

%---------------------------------------------------------------------------------------------------------------------------------------------------------

\section{Constructing $\DDCH{i+1}$ from $\DDCH{i}$}
\label{subsec:Ui->Ui+1}

\junk{\Artur{Some random text that could be used here: As we said above, in our construction we will gradually shrink in $G$ all copies of $H$ into smaller sub-hypergraphs $\HQ_1(\DDCH{1}), \HQ_2(\DDCH{2}), \dots$. Each shrinking will follow the same rule in all copies of $H$, and will choose vertex $v_i$ (to be contracted, as defined in Section \ref{subsec:shrinking-H}) on the basis of the current structure of $\DDCH{i}$. Our key property (cf. Lemma \ref{lemma:small-vertices-new}) is that \emph{for every $i$, hypergraph $\HQ_i(\DDCH{i})$ is consistent for $\DDCH{i}$}. We will prove this claim by showing that assuming that $\HQ_i(\DDCH{i})$ is consistent for $\DDCH{i}$, there is always a subset $\DCH \subseteq \DDCH{i}$ of linear size (with $|\DCH| = \Omega_{\eps,H}(|\DDCH{i}|)$) for which there is a color $\cc$ such that every vertex of color $\cc$ in $V(\HQ_i(\DDCH{i}))$ is safe with respect to $\DDCH{i}$ and $\HQ_i(\DDCH{i})$. Using this claim, we will be able to perform an appropriate construction of the hypergraph $\HQ_{i+1}(\DDCH{i+1})$ by first setting $\DDCH{i+1} := \DCH$ and then taking $v_i$ with $\chi(v_i) = \cc$, and contracting (as $v_i$ in $\M_i$) all vertices of color $\cc$ in all copies of $H$ in $\DDCH{i+1}$.
\\
As we described in Section \ref{subsec:shrinking-copies-of-H}, the central idea behind our analysis is to study the relation between random exploration algorithms in hypergraphs $\HQ_i(\DDCH{i})$ vs $\HQ_i(\DDCH{i+1})$. In this context, one key task in our analysis is the choice of the order of vertices in $H$, $v_1, \dots, v_{\Hsize}$, and another is the selection of $\DDCH{i+1}$ as a subset of $\DDCH{i}$. Given $\DDCH{i}$ and assuming that $\HQ_i(\DDCH{i})$ is a hypergraph consistent for $\DDCH{i}$, we will be able to determine $v_i$ and to find a large set $\DDCH{i+1} \subseteq \DDCH{i}$, such that the hypergraph $\HQ_i(\DDCH{i+1})$ will be consistent for $\DDCH{i+1}$ (see Lemma \ref{lemma:small-vertices-new}). The proof of this claim relies on some planarity properties of the original graph $G$, and hence, indirectly, of the hypergraphs considered. However, since the notion of planar hypergraphs is not very suitable in our context, we will rely on some planar graph representation of them.
\\
In what follows, we will consider a set $\DDCH{i}$ of edge-disjoint colored copies of $H$ in $G$ and the corresponding hypergraph $\HQ_i(\DDCH{i})$. We will show in Lemma \ref{lemma:small-vertices-new} that if $\HQ_i(\DDCH{i})$ is consistent for $\G[\DDCH{i}]$, then a constant fraction of non-isolated vertices in $\HQ_i(\DDCH{i})$ will have a small number of distinct neighbors. Our proof uses some basic transformation of $\HQ_i(\DDCH{i})$ (as it has been constructed by the algorithm above) into planar graphs.
}}
The construction of $\DDCH{1}$ from Section \ref{subsec:constructing-H1} is rather simple, but it is significantly more complex to define $\DDCH{2}$, and then $\DDCH{3}, \dots, \DDCH{\Hsize}$. In what follows, we will first present key intuitions in Section \ref{subsubsec:gadgets}, then describe our framework in Sections \ref{subsec:shrinking-H} and \ref{subsec:shrinking-copies-of-H}, and present details of the construction of $\DDCH{i+1}$ in Section \ref{subsec:main-reduction}.

%---------------------------------------------------------------------------------------------------------------------------------------------------------

While our main focus is on the sets $\DDCH{i}$ of edge-disjoint colored copies of $H$ in $G$, in our analysis we will analyze these sets and the relevant graphs $\G[\DDCH{i}]$ via their suitable \emph{hypergraph representation}. Indeed, to prove that \RLBFS\ finds a copy of $H$, we will consider a hypergraph induced by ``shrunk'' copies of $H$ defining $\DDCH{i}$. The idea of this construction is two-folded:%\Artur{This sounds lame and should be rewritten.}
\begin{itemize}
\item on one hand, using the hypergraph representation it will be easier to argue a lower bound for the probability that a copy of $H$ is found, and
\item on the other hand, the hypergraph representation will allow us to combine distinct colored copies of $H$ (or the subgraph of $H$) that are undistinguishable to \RLBFS.
\end{itemize}

%---------------------------------------------------------------------------------------------------------------------------------------------------------

\subsection{Overview: Gadgets, hypergraph representation and their use}
\label{subsubsec:gadgets}

Our analysis relies on special structures (gadgets) in the input graph and then representing these gadgets in a succinct way using \emph{hypergraphs}.

Let $G^*$ be a subgraph of $G$ that has a copy of $H$. Consider a subgraph $H_1$ of $H$ and let $u_1, \dots, u_{\ell}$ be the vertices in the copy of $H_1$ in $G^*$ that separate $G \setminus H_1$ from $H_1$, so that (cf. Figure \ref{fig:ex-gadgets-1}):
%
%\begin{enumerate}[\sl (a)]
%\begin{inparaenum}[\sl (a)]
\begin{compactenum}[\sl \qquad (a)]
\item every vertex from $\{u_1, \dots, u_{\ell}\}$ is adjacent in $G^*$ to some vertex $H_1 \setminus \{u_1, \dots, u_{\ell}\}$,
\item every vertex in $H_1 \setminus \{u_1, \dots, u_{\ell}\}$ is adjacent in $G^*$ only to vertices from $H_1$, and
\item $\{u_1, \dots, u_{\ell}\}$ forms an independent set in $H_1$.
\end{compactenum}
%\end{inparaenum}
%\end{enumerate}
%
Then, we can construct a gadget to represent that copy of $H_1$ by removing from $H_1$ all vertices and edges from $H_1 \setminus \{u_1, \dots, u_{\ell}\}$ and replacing them by a single \emph{hyperedge} $\{u_1, \dots, u_{\ell}\}$.

\junk{
\begin{figure}[t]
\centerline
{(a) \includegraphics[width=.3\textwidth]{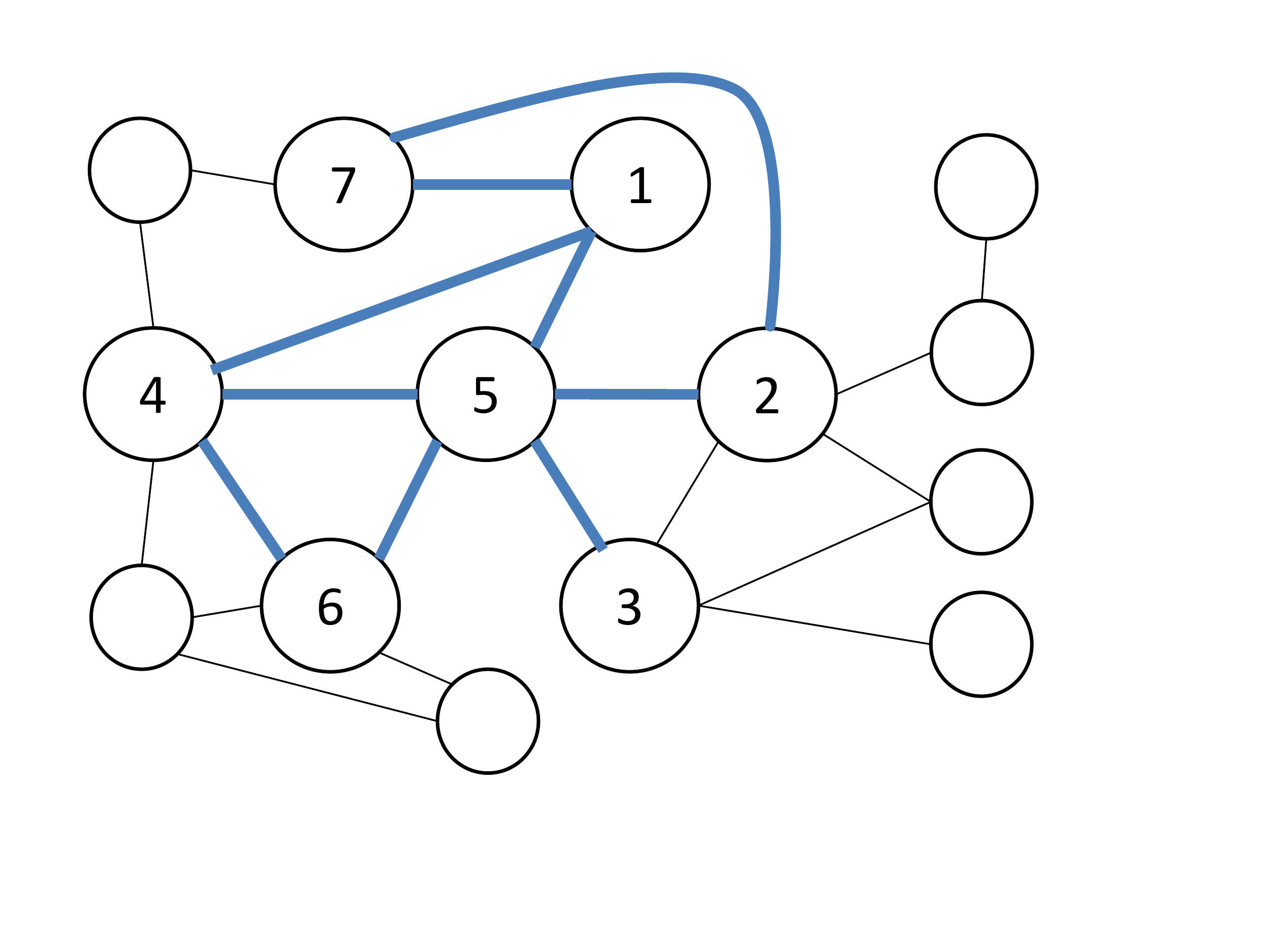}
(b) \includegraphics[width=.3\textwidth]{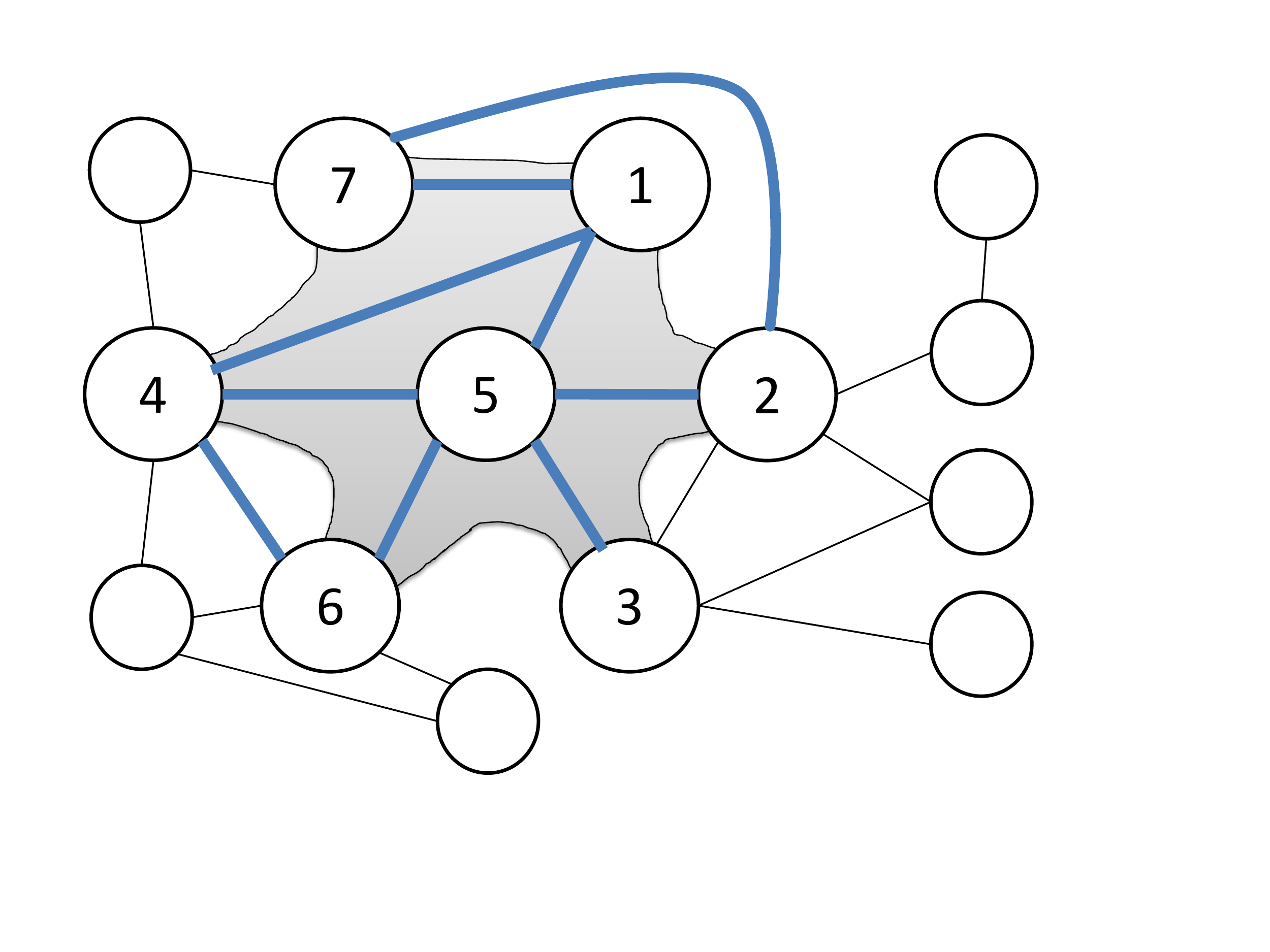}
(c) \includegraphics[width=.3\textwidth]{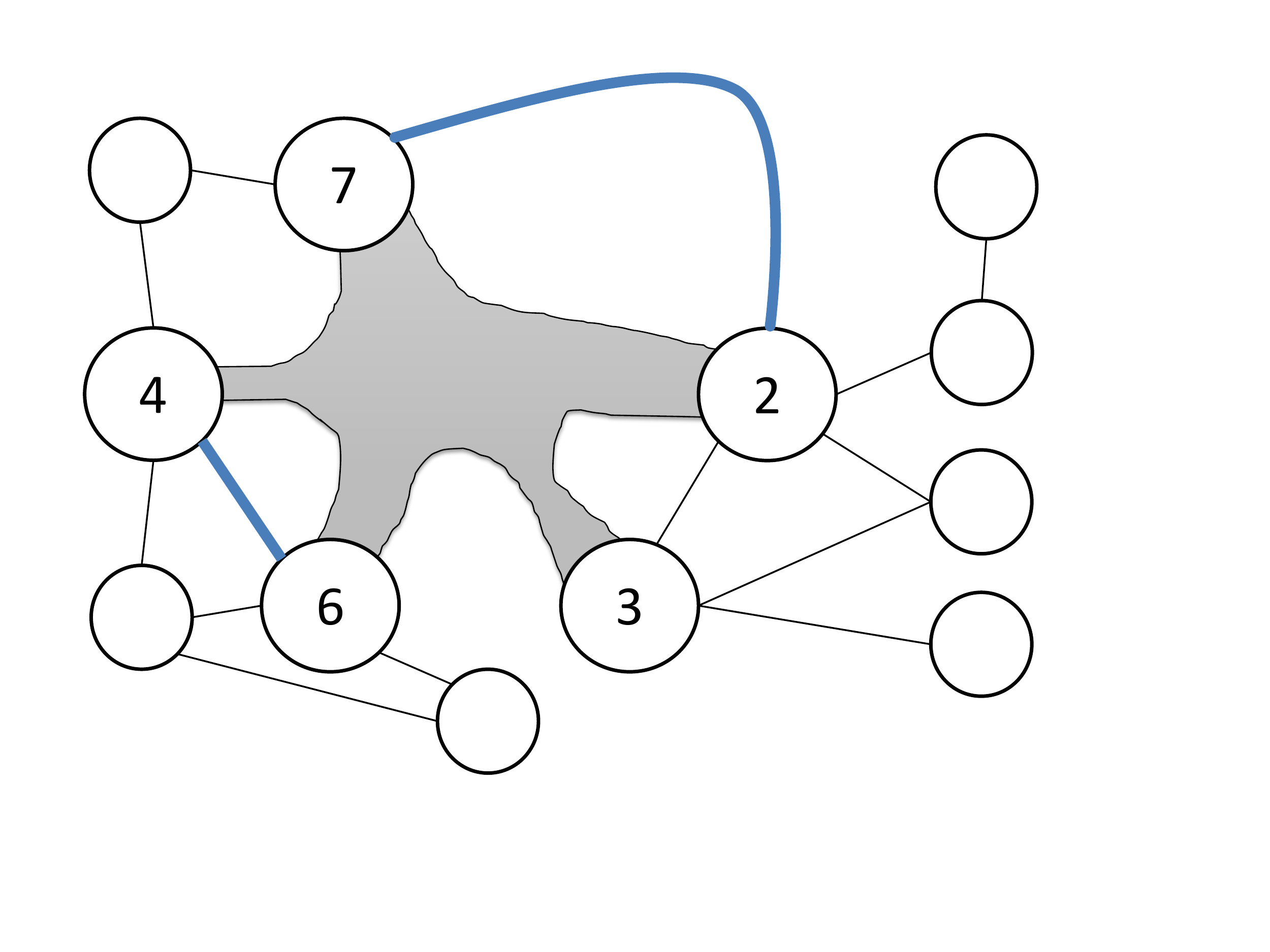}
}
\caption{\small (a) Consider a part of the input graph $G$ with numbered vertices corresponding to the colored vertices in a copy of $H$ in $G$ and thick edges corresponding to the edges in that copy of $H$.
(b) We have a subgraph $H_1$ of $H$ consisting of the vertices and edges marked by the grey area, with vertices $\{2,3,4,6,7\}$ separating $H_1$ from the rest of $G$.
(c) The gadget obtained by removing internal vertices $\{1,5\}$ and replacing $H_1$ by a hyperedge connecting vertices $\{2,3,4,6,7\}$.}
\label{fig:ex-gadgets-1}
\end{figure}
}

We will be using this construction of gadgets to model the following scenario:
\begin{itemize}
\item \emph{when entering (in \RLBFS) $H_1$ via any single edge incident to any vertex from the separator $u_1, \dots, u_{\ell}$ is sufficient to visit (with constant probability) all edges in $H_1$.}
\end{itemize}
Therefore, for the analysis, this will correspond to the situation that
\begin{itemize}
\item there is a hyperedge $\{u_1, \dots, u_{\ell}\}$, and by visiting this hyperedge (in the hypergraph), the algorithm will visit (with constant probability) all edges in $H_1$ (in the original graph), and will be able to continue the search from \emph{all} separating vertices $u_1, \dots, u_{\ell}$.
\end{itemize}

Furthermore, the gadgets can be also helpful in the analysis of ``substitutable'' copies of a subgraph of $H$. Suppose that for a subgraph $H_1$ of $H$, the separator (as defined above) is identical in multiple copies, that is, vertices $u_1, \dots, u_{\ell}$ form the separator in multiple edge-disjoint copies of $H_1$. Then, we have multiple hyperedges $\{u_1, \dots, u_{\ell}\}$ and their multiplicity represents the fact that \emph{to find a copy of $H_1$ it is enough to visit just one of the hyperedges $\{u_1, \dots, u_{\ell}\}$}. In particular, if $u_1$ is incident to multiple copies of the identical hyperedge $\{u_1, \dots, u_{\ell}\}$, then the probability that the process will visit $H_1$ starting from $u_1$ increases with this multiplicity. And so, if the multiplicity is of order $\deg_G(u_1)$, then after reaching vertex $u_1$, the \HRLBFS\ algorithm (cf. Section \ref{sec:exploration-in-hypergraphs}) will visit the entire $H_1$ with a constant probability.

The central idea behind the gadgets as described above is to use them repeatedly to transform a subgraph of $G$ into a sub-hypergraph representing a smaller subgraph of $G$ for which we can easily analyze the \HRLBFS\ algorithm.

%---------------------------------------------------------------------------------------------------------------------------------------------------------

\subsection{The process of shrinking $H$ and hypergraph representation of $H$ by $\M_i$}
\label{subsec:shrinking-H}

\junk{\Artur{Notation:
    \begin{compactitem}
    \item \texttt{$\backslash$e} denotes $\e$, an edge/hyperedge
    \item \texttt{$\backslash$N\_i} denotes $\N_i$, the set of all neighbors of $v_i$ in $\M_i$
    \item \texttt{$\backslash$mathcal\{E\}\_i} denotes $\mathcal{E}_i$, the set of edges/hyperedges incident to vertex $v_i$ in $\M_i$
    \item \texttt{$\backslash$lab} denotes $\lab$, label of an edge/hyperedge, $\lab(\e)$
    \item \texttt{$\backslash$clab} denotes $\clab$, colored label of an edge/hyperedge, $\clab(\e)$ (note that $\clab(\e) = \{\chi(u): u \in \lab(\e)\}$)
    \end{compactitem}
}}
We will begin with an iterative procedure that gradually shrinks $H$ into a single vertex. This procedure processes $H$ and its contractions in a form of a \emph{hypergraph}. (See also Figures \ref{fig:ex-hypergraph-1}--\ref{fig:ex-hypergraph-3}.)

%---------------------------------------------------------------------------------------------------------------------------------------------------------

\begin{figure}[t]
\centerline
{(a) \includegraphics[width=.3\textwidth]{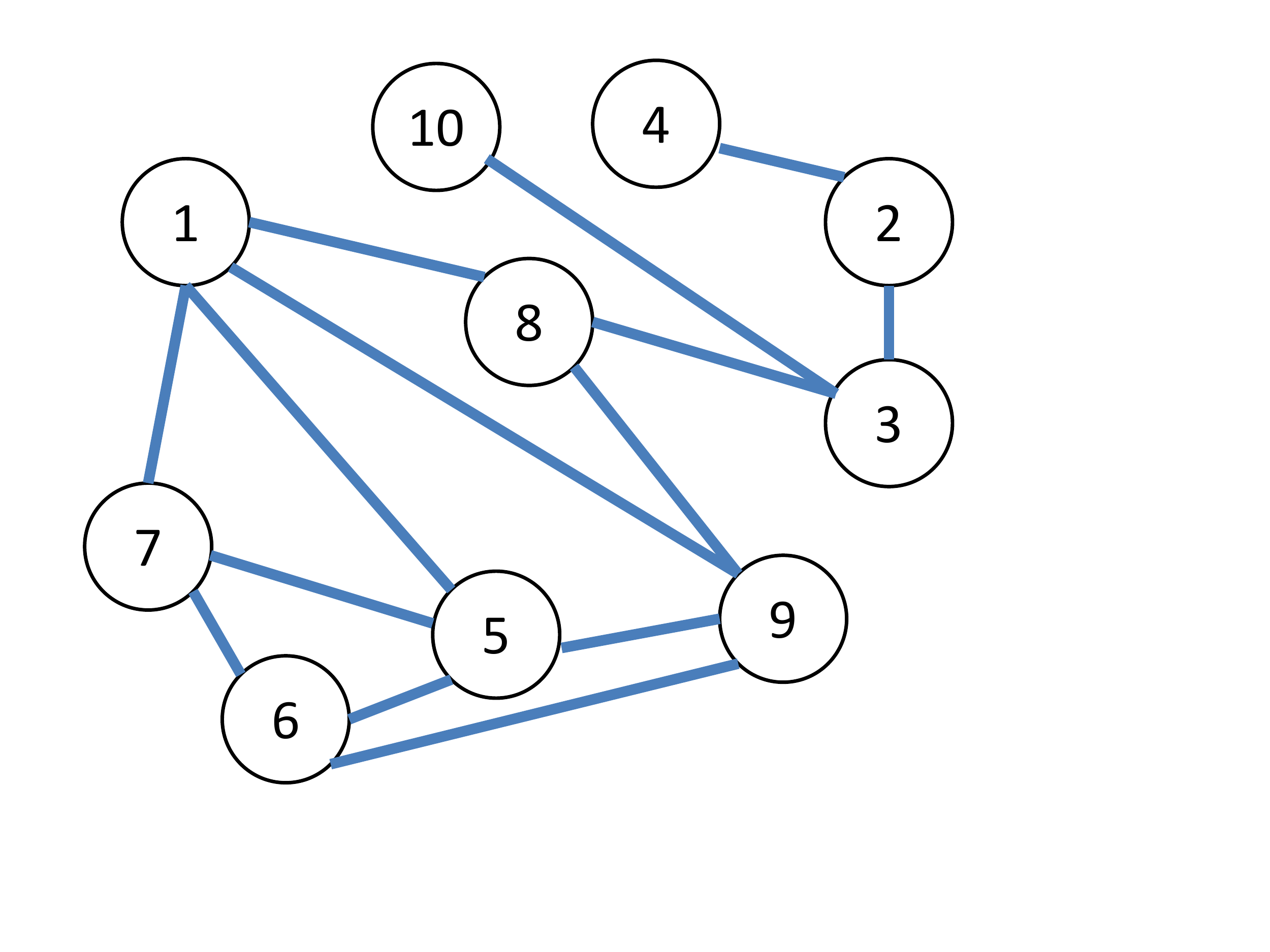}
(b) \includegraphics[width=.3\textwidth]{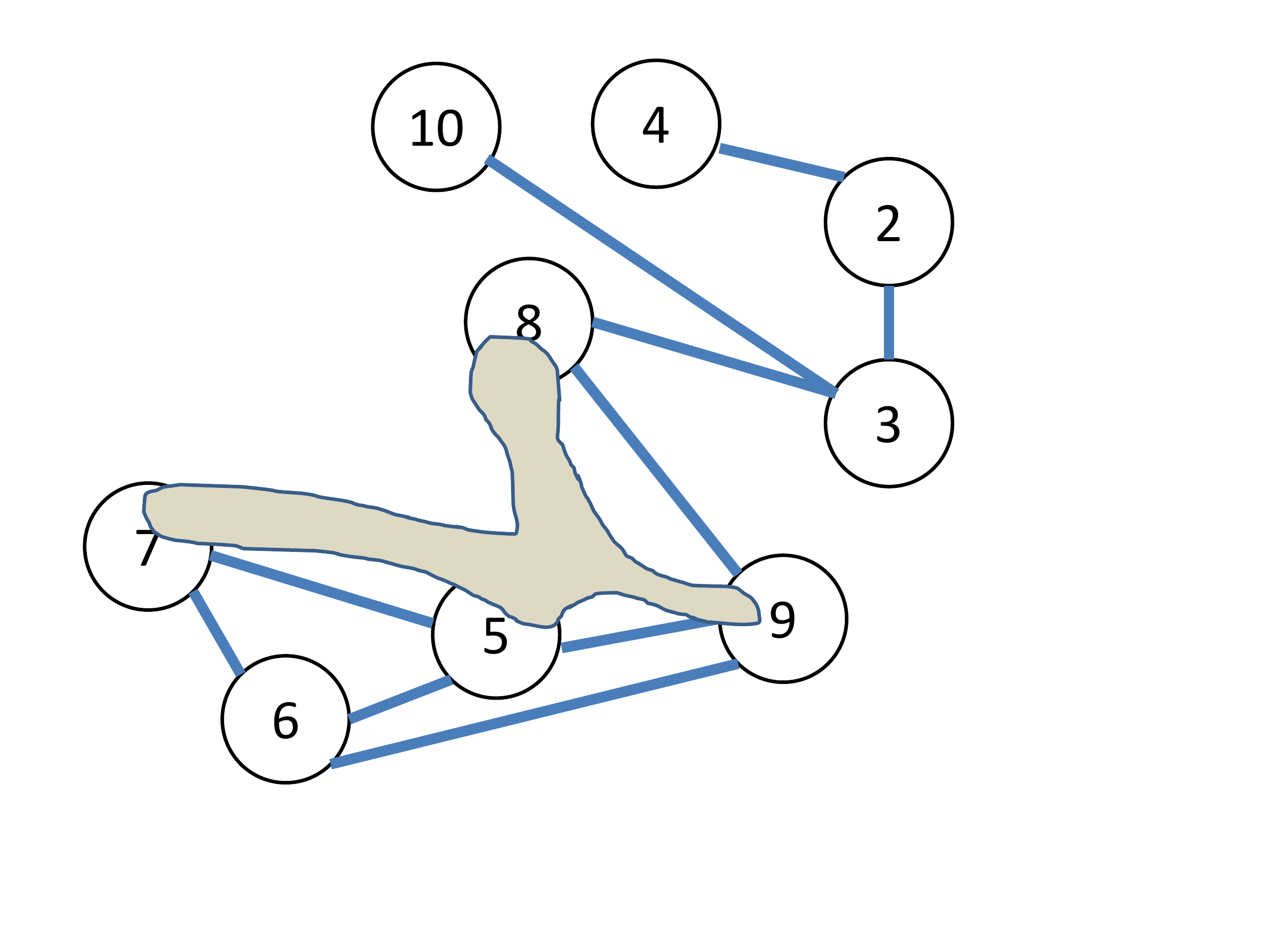}
(c) \includegraphics[width=.3\textwidth]{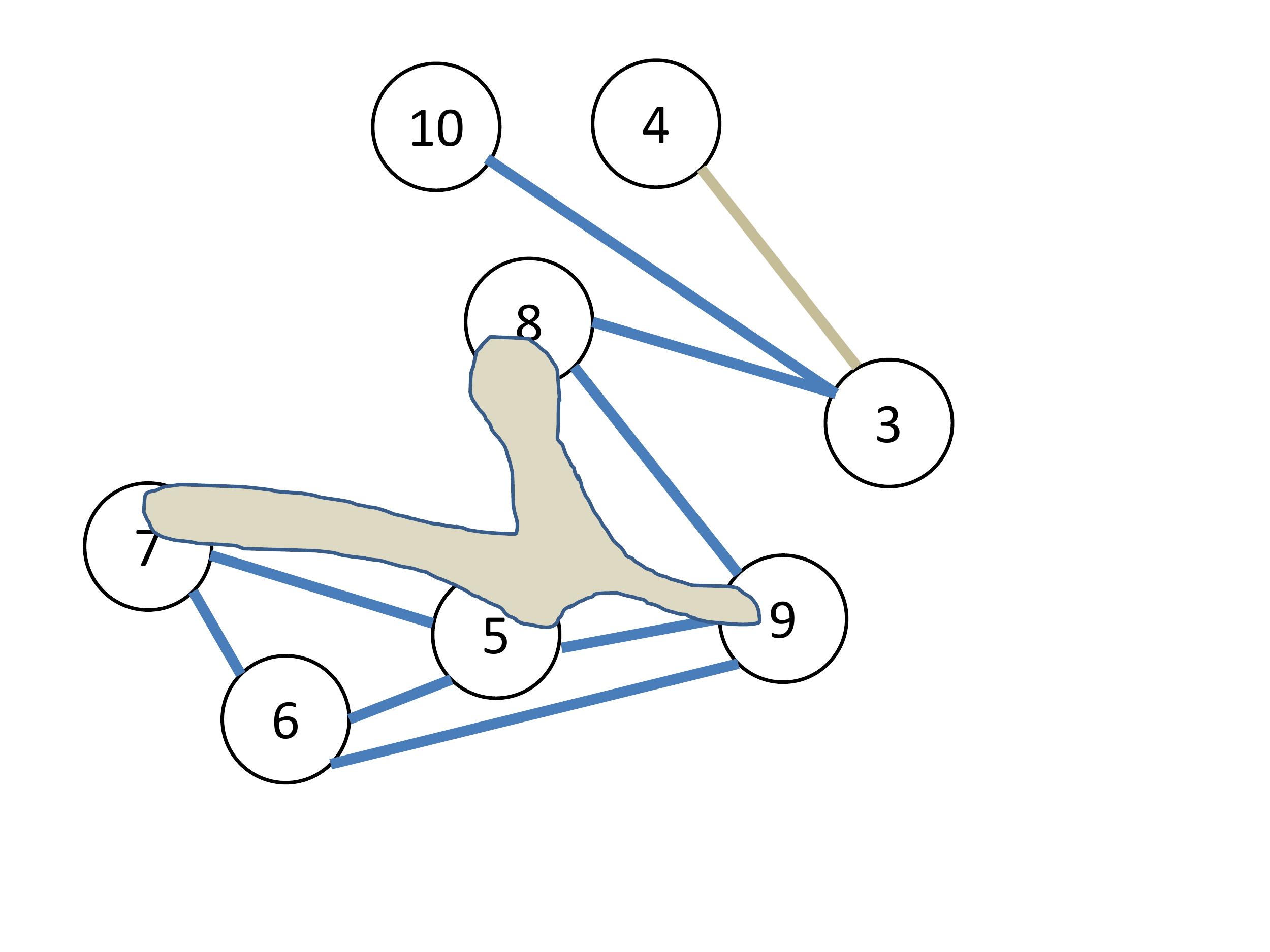}
}
\caption{\small Consider the input graph $G$ in Figure (a) and consider the process of shrinking $G$, as described in Section \ref{subsec:shrinking-H}.
(b) presents contraction of vertex 1 and adding of hyperedge $\{5,7,8,9\}$ (with label $\{1\}$).
(c) After contracting vertex 2 and adding hyperedge $\{3,4\}$ (with label $\{2\}$).
}
\label{fig:ex-hypergraph-1}
\end{figure}

%---------------------------------------------------------------------------------------------------------------------------------------------------------

%\subsubsection{Shrinking process}
%
Let us consider an arbitrary numbering of the vertices of $H$, $v_1, v_2, \dots, v_{\Hsize}$; this order is not known in advance and is independent of the coloring of $H$ (in fact, the order will be determined by the structure of $G$, and finding the right order $v_1, v_2, \dots, v_{\Hsize}$ is the central part of our analysis in the next sections,
%\Artur{It's essential to make this clear to the reader --- that selecting the right order $v_1, v_2, \dots, v_{\Hsize}$ is the essential (and most challenging) part of our analysis.},
finalized in Lemma \ref{lemma:defining-Qi+1}). In our analysis, we will perform a sequence of transformations on $H$, each transformation converting some hypergraph $\M_i$ corresponding to $H$ into some other hypergraph $\M_{i+1}$ corresponding to $H$, $1 \le i \le \Hsize-1$ (cf. Figures \ref{fig:ex-hypergraph-1}--\ref{fig:ex-hypergraph-3}), such that:

%---------------------------------------------------------------------------------------------------------------------------------------------------------
\begin{walgo}\vspace*{-0.3in}
\begin{itemize}
\item $\M_1 := H$, and
%\item each $\M_i$ has $\Hsize - i + 1$ vertices,
\item $\M_{i+1}$ is obtained from $\M_i$ by contracting vertex $v_i$ to its neighbors as follows:
    \begin{itemize}[$\diamond$]
    \item let $\N_i$ be the set of all neighbors of $v_i$ in $\M_i$; contract $v_i$ to its neighbors by removing $v_i$ from $\M_i$ and then adding a new hyperedge consisting of vertices in $\N_i$.
    \end{itemize}
\end{itemize}
\end{walgo}
%---------------------------------------------------------------------------------------------------------------------------------------------------------

%---------------------------------------------------------------------------------------------------------------------------------------------------------

\begin{figure}[t]
\centerline
{
(d) \includegraphics[width=.3\textwidth]{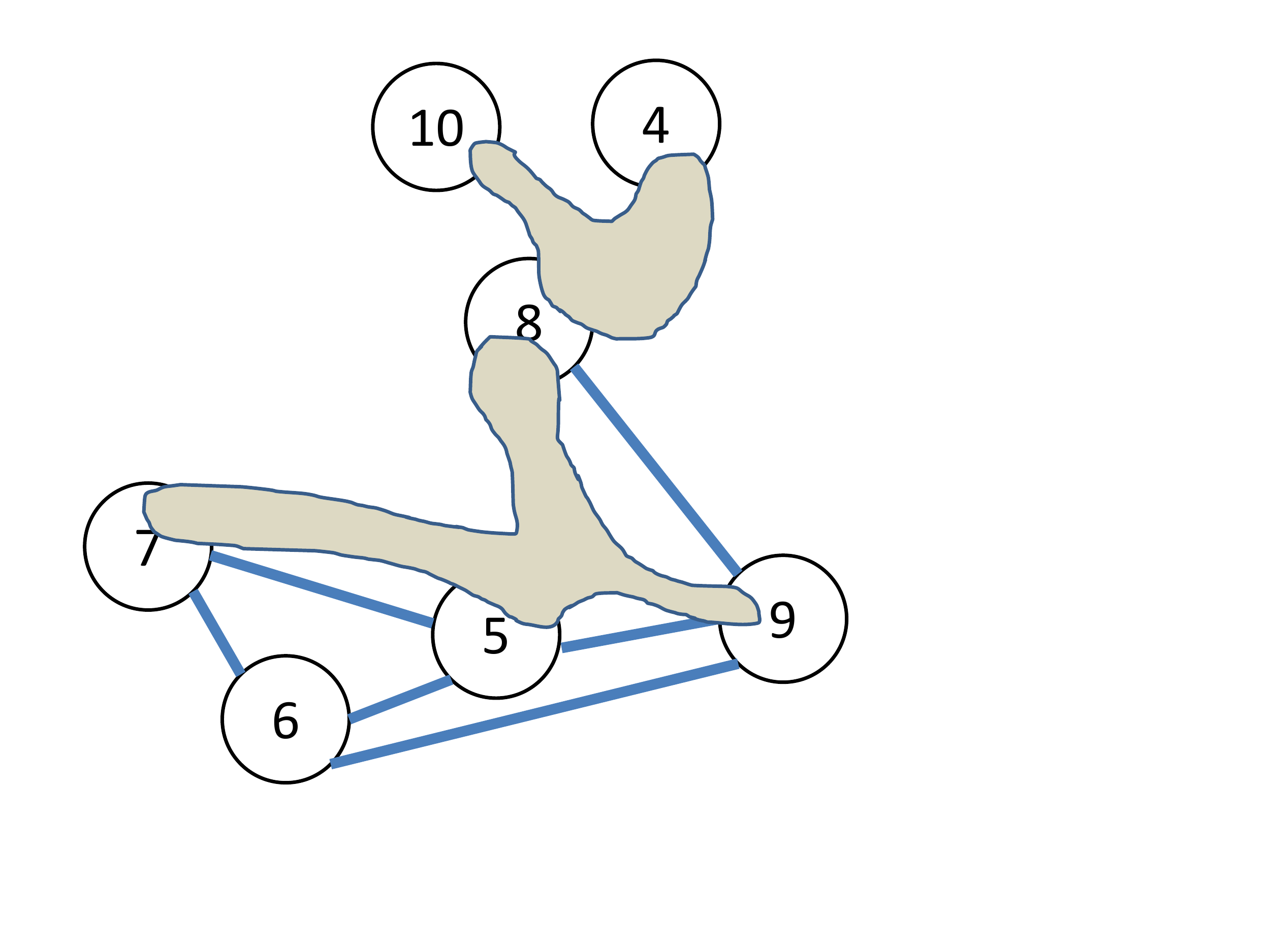}
(e) \includegraphics[width=.3\textwidth]{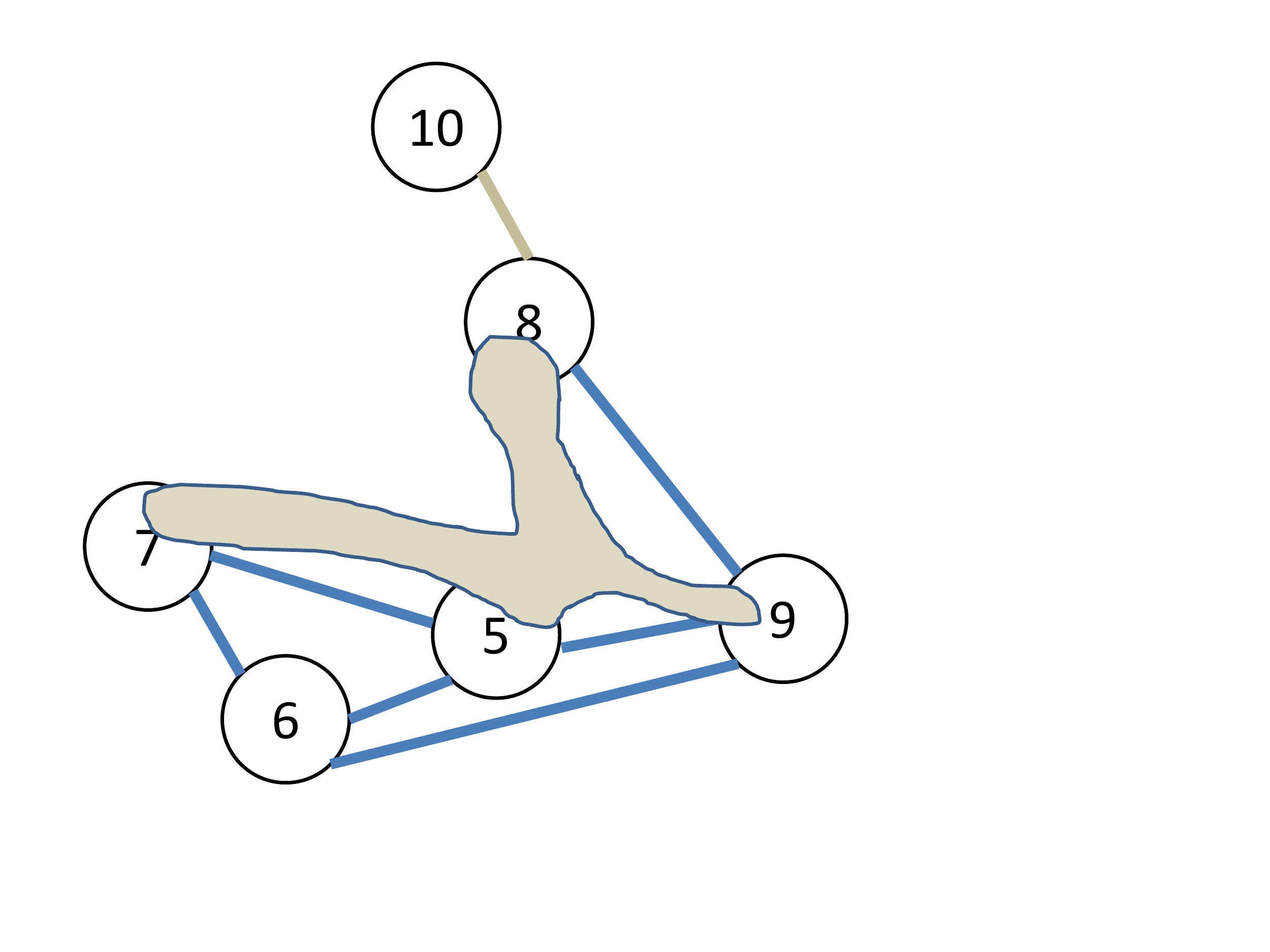}
(f) \includegraphics[width=.3\textwidth]{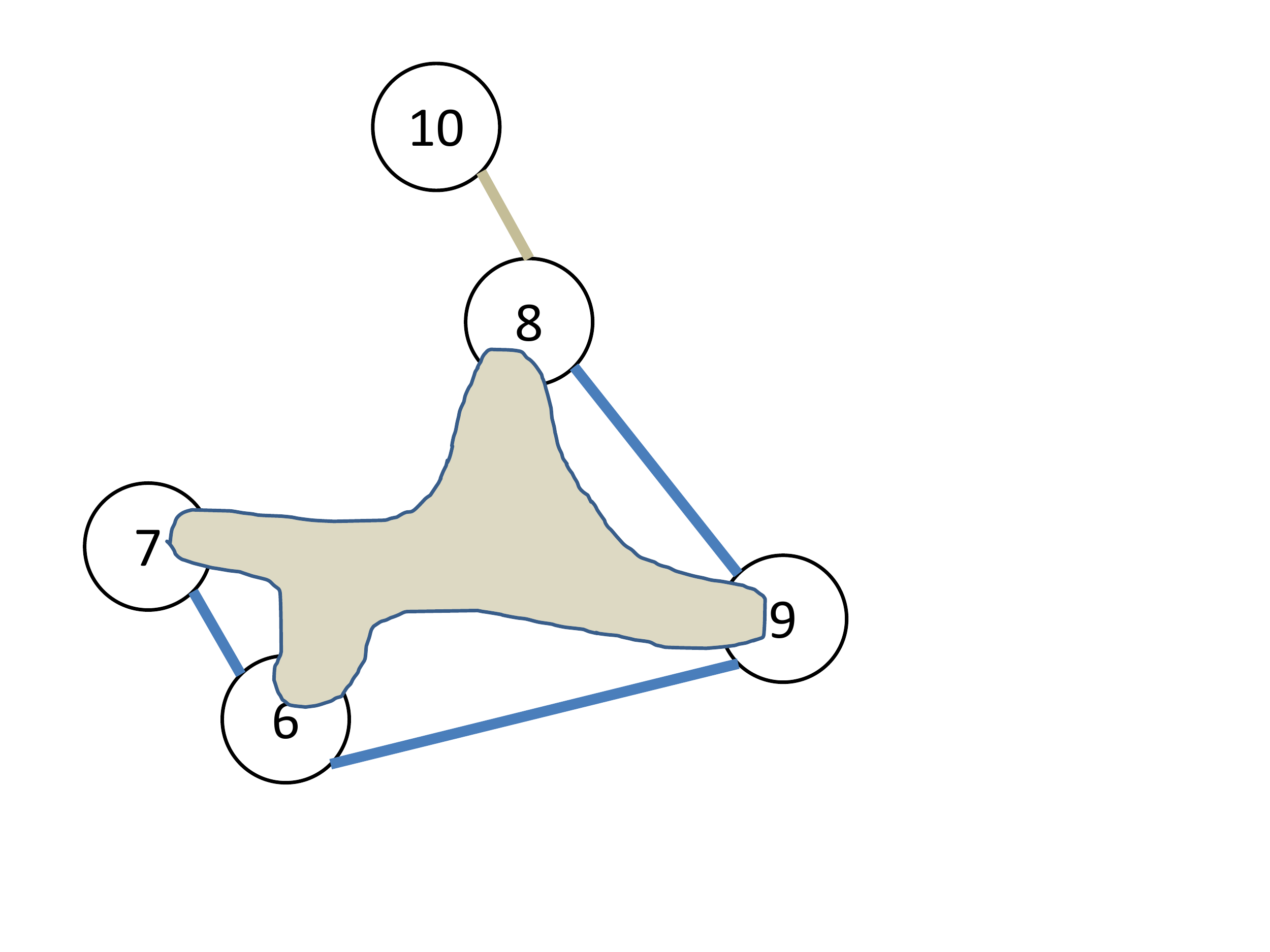}
}
\caption{\small Continuing the example from Figure \ref{fig:ex-hypergraph-1}: (d) After contracting vertex 3 and adding hyperedge $\{4,8,10\}$ (with label $\{2,3\}$).
(e) After contracting vertex 4 and adding hyperedge $\{8,10\}$ (with label $\{2,3,4\}$).
(f) After contracting vertex 5 and adding hyperedge $\{6,7,8,9\}$ (with label $\{1,5\}$).
}
\label{fig:ex-hypergraph-2}
\end{figure}

%---------------------------------------------------------------------------------------------------------------------------------------------------------

%---------------------------------------------------------------------------------------------------------------------------------------------------------

\begin{figure}[t]
\centerline
{(g) \includegraphics[width=.21\textwidth]{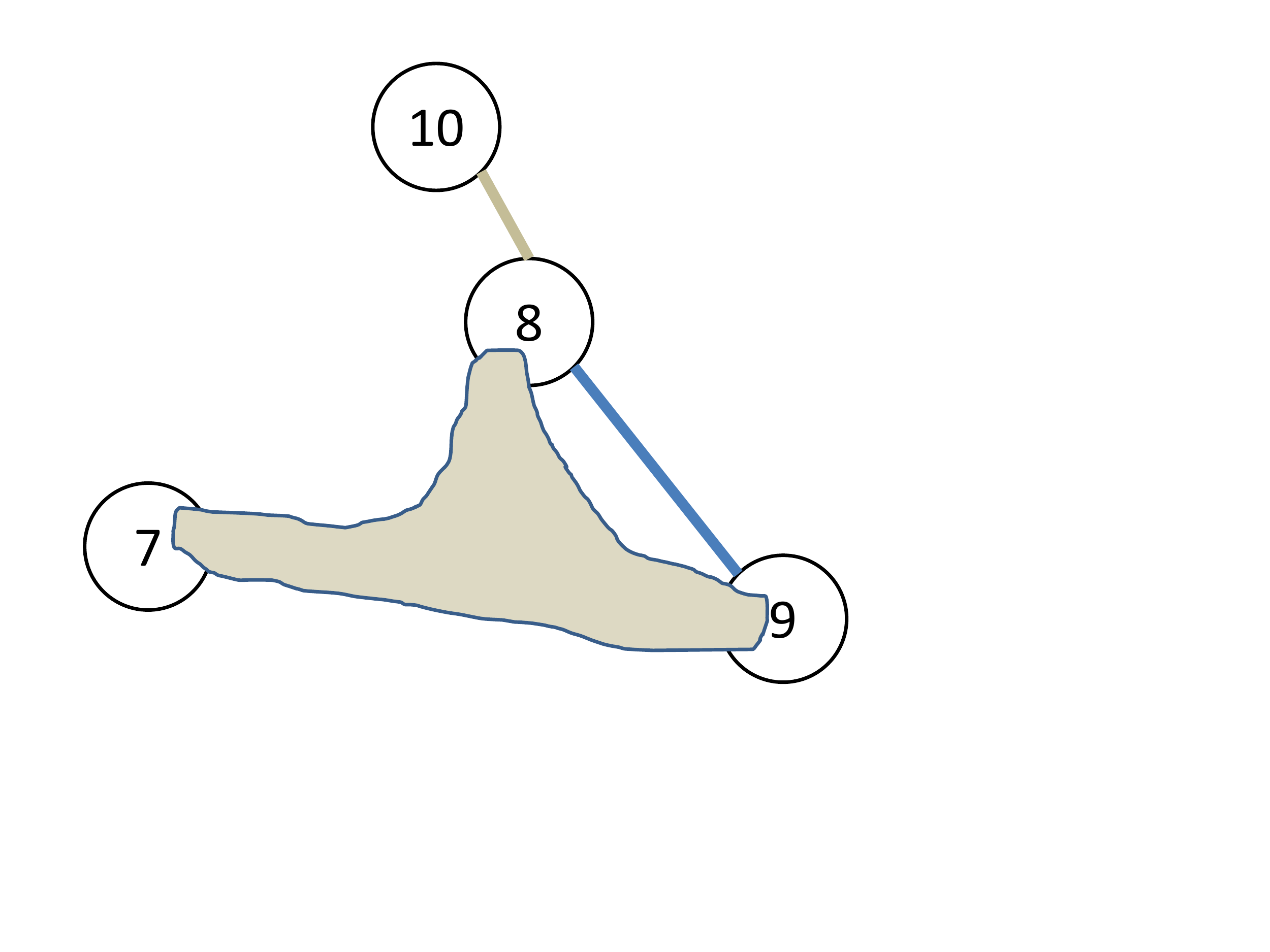}
(h) \includegraphics[width=.21\textwidth]{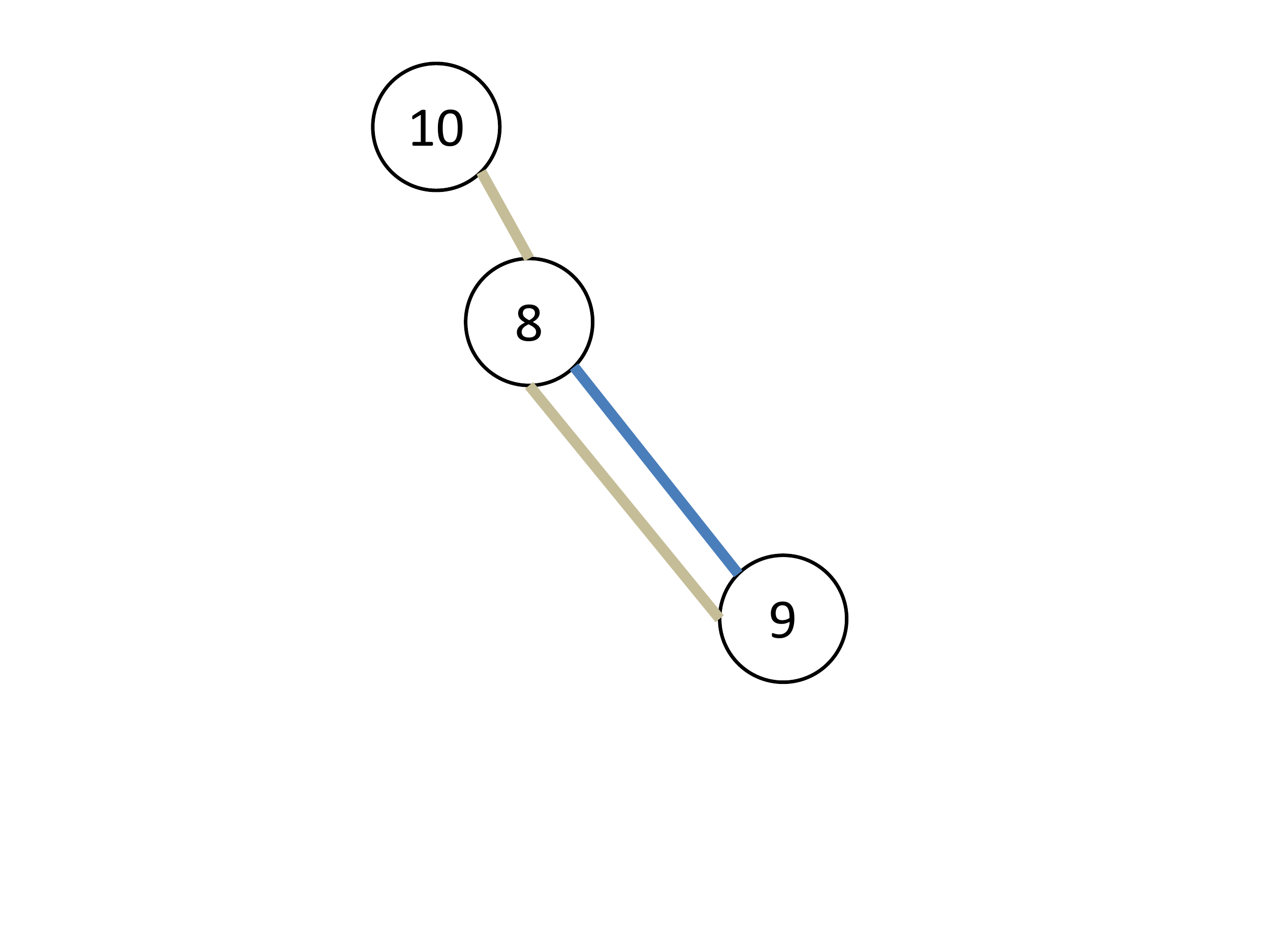}
(i) \includegraphics[width=.21\textwidth]{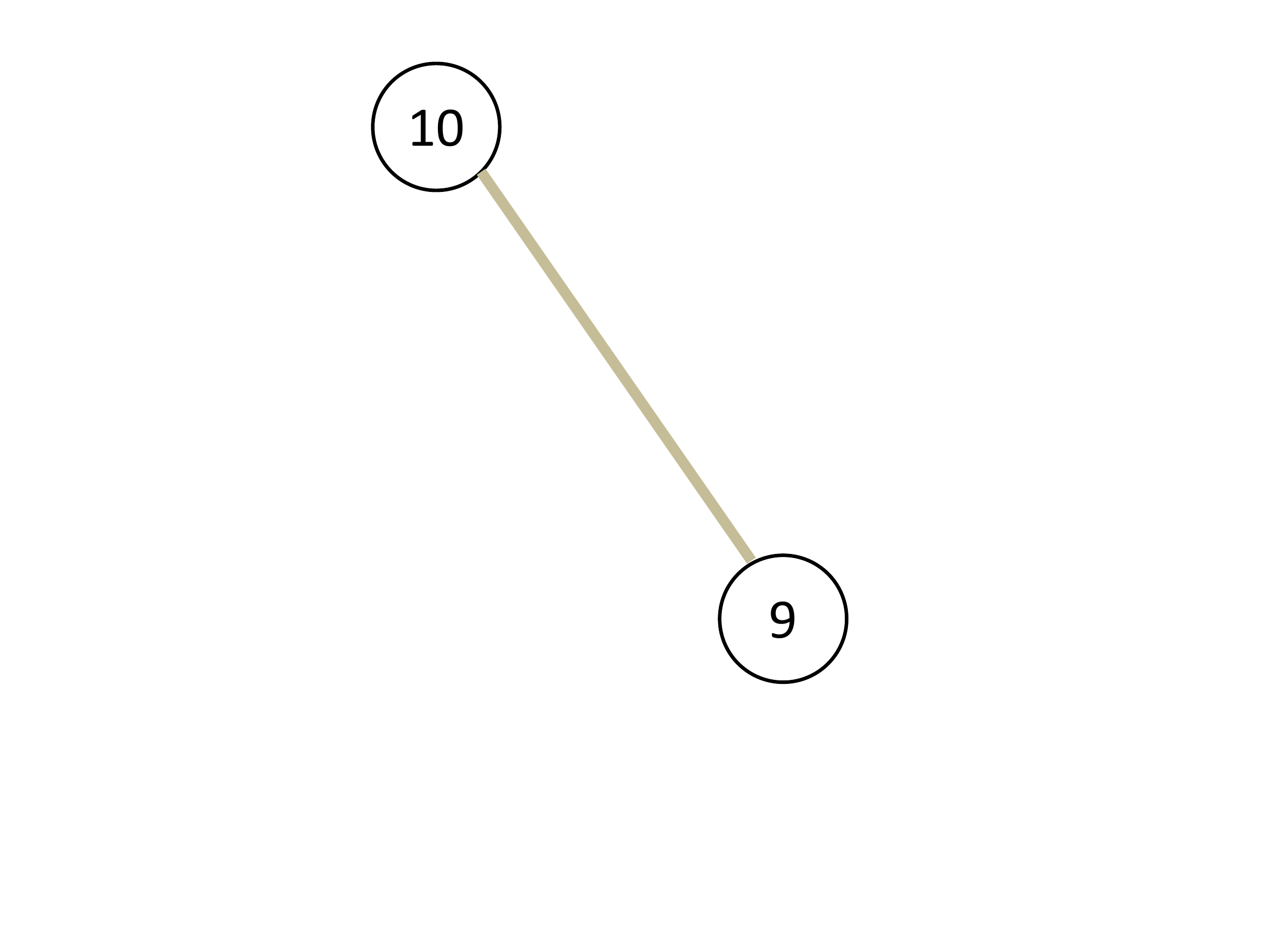}
(j) \includegraphics[width=.21\textwidth]{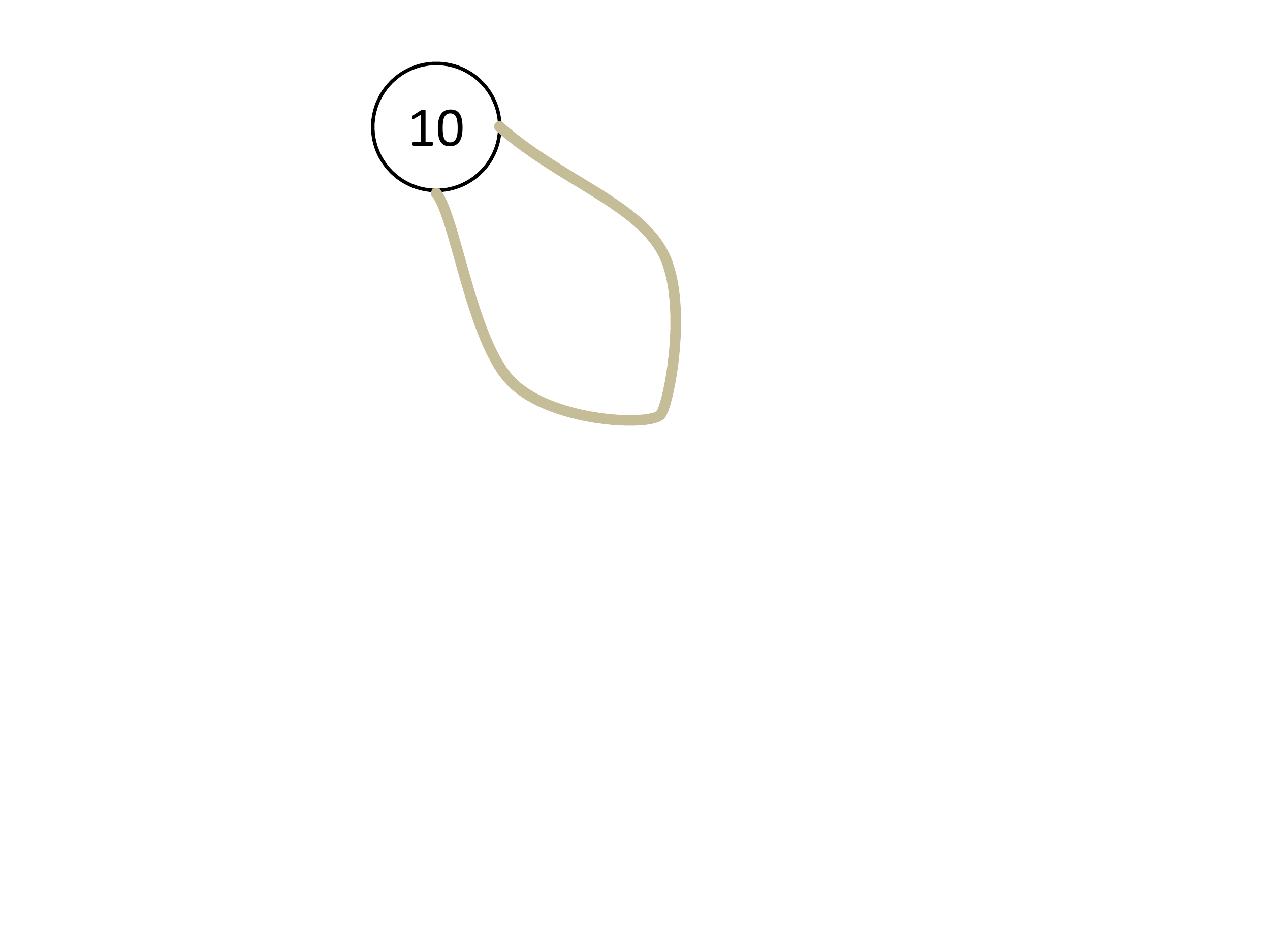}
}
\caption{\small Continuing the example from Figures \ref{fig:ex-hypergraph-1} and \ref{fig:ex-hypergraph-2}:  (g) After contracting vertex 6 and adding hyperedge $\{7,8,9\}$ (with label $\{1,5,6\}$).
(h) After contracting vertex 7 and adding hyperedge $\{8,9\}$ (with label $\{1,5,6,7\}$; note that as the result, we have two parallel edges between 8 and 9, but each of these edges is different, one corresponds to a direct edge between 8 and 9 with label $\emptyset$, and another corresponds to the gadget with separator $\{8,9\}$ and internal vertices $\{1,5,6,7\}$ (as shown by the label)).
(i) After contracting vertex 8 and adding hyperedge $\{9,10\}$ (with label $\{1,2,3,4,5,6,7,8\}$).
(j) After contracting vertex 9 and hyperedge $\{10\}$ (with label $\{1,2,3,4,5,6,7,8,9\}$).
}
\label{fig:ex-hypergraph-3}
\end{figure}

%---------------------------------------------------------------------------------------------------------------------------------------------------------

We will want to maintain information about all vertices which have been contracted to create a given hyperedge (e.g., in Figure~\ref{fig:ex-gadgets-1}, these would be vertices $\{1,5\}$) and so we will \emph{label} the hyperedges. We will denote the label of an edge $\e$ by $\lab(\e)$. A regular edge $e$ (original edge from $E(H)$) has an empty label, i.e., $\lab(e) = \emptyset$, and if $\mathcal{E}_i$ denotes the set of edges/hyperedges incident to vertex $v_i$ in $\M_i$, then the new hyperedge $\N_i$ obtained by contraction of $v_i$ will have label $\lab(\N_i) = \{v_i\} \cup \bigcup_{\e \in \mathcal{E}_i} \lab(\e)$ (i.e., its label is the union of $\{v_i\}$ and the union of the labels of the edges in $\mathcal{E}_i$).

Furthermore, we will also have \emph{colored label} $\clab$ of any edge $\e$, defined as the set of the colors of the vertices defining the label of $\e$, that is, $\clab(\e) = \{\chi(u): u \in \lab(\e)\}$. (Note that if $\lab(\e) = \emptyset$ then $\clab(\e) = \emptyset$.)

We will also use the following notion.

\begin{definition}
\label{def:modeling-hyperedge-in-Mi+1}
If in our construction, in $\M_i$, we had edges $\e_1, \dots, \e_{\ell}$ incident to $v_i$, then we will say that the newly created \emph{hyperedge $\N_i$ in $\M_{i+1}$ is modeled by edges $\e_1, \dots, \e_{\ell}$ in $\M_i$}.
\end{definition}

\medskip

In Appendix \ref{proofs-subsec:shrinking-H} we will present some basic properties of the process of shrinking $H$ and hypergraph representation of $H$ by $\M_i$, as defined in this section.%Section \ref{subsec:shrinking-H}

\begin{figure}[t]
\centerline
{(a) \includegraphics[width=.22\textwidth]{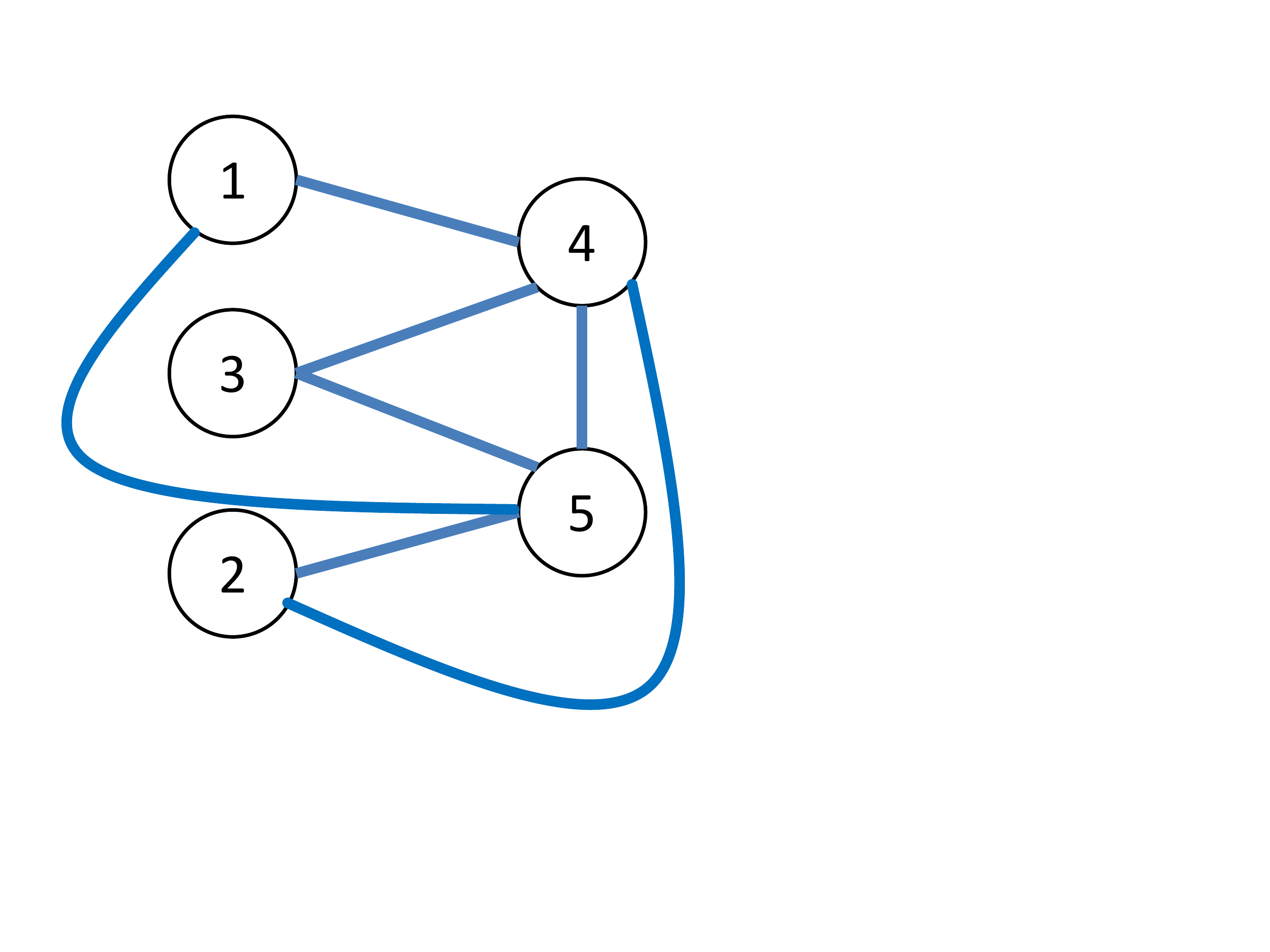}
(b) \includegraphics[width=.22\textwidth]{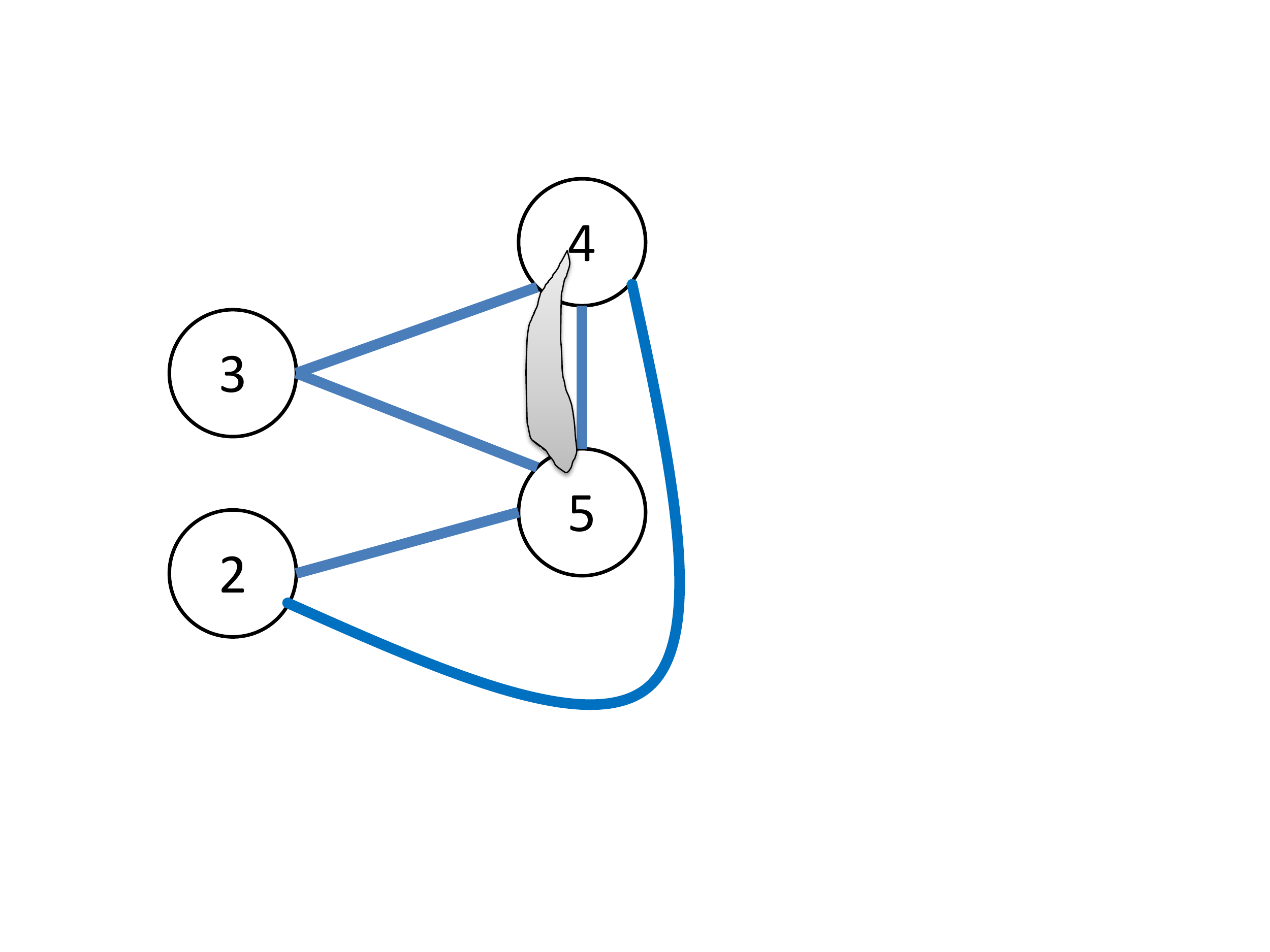}
(c) \includegraphics[width=.22\textwidth]{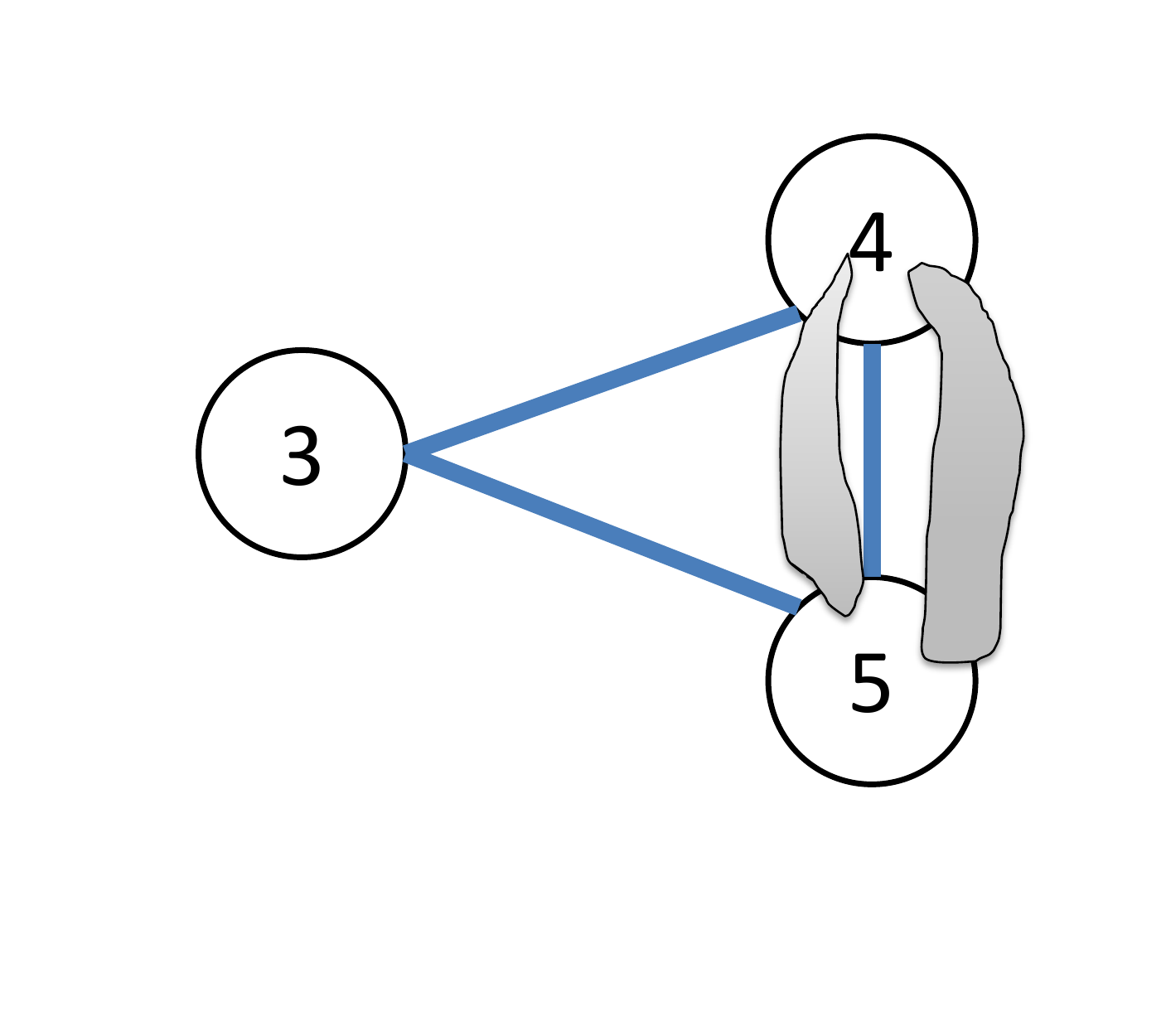}
(d) \includegraphics[width=.22\textwidth]{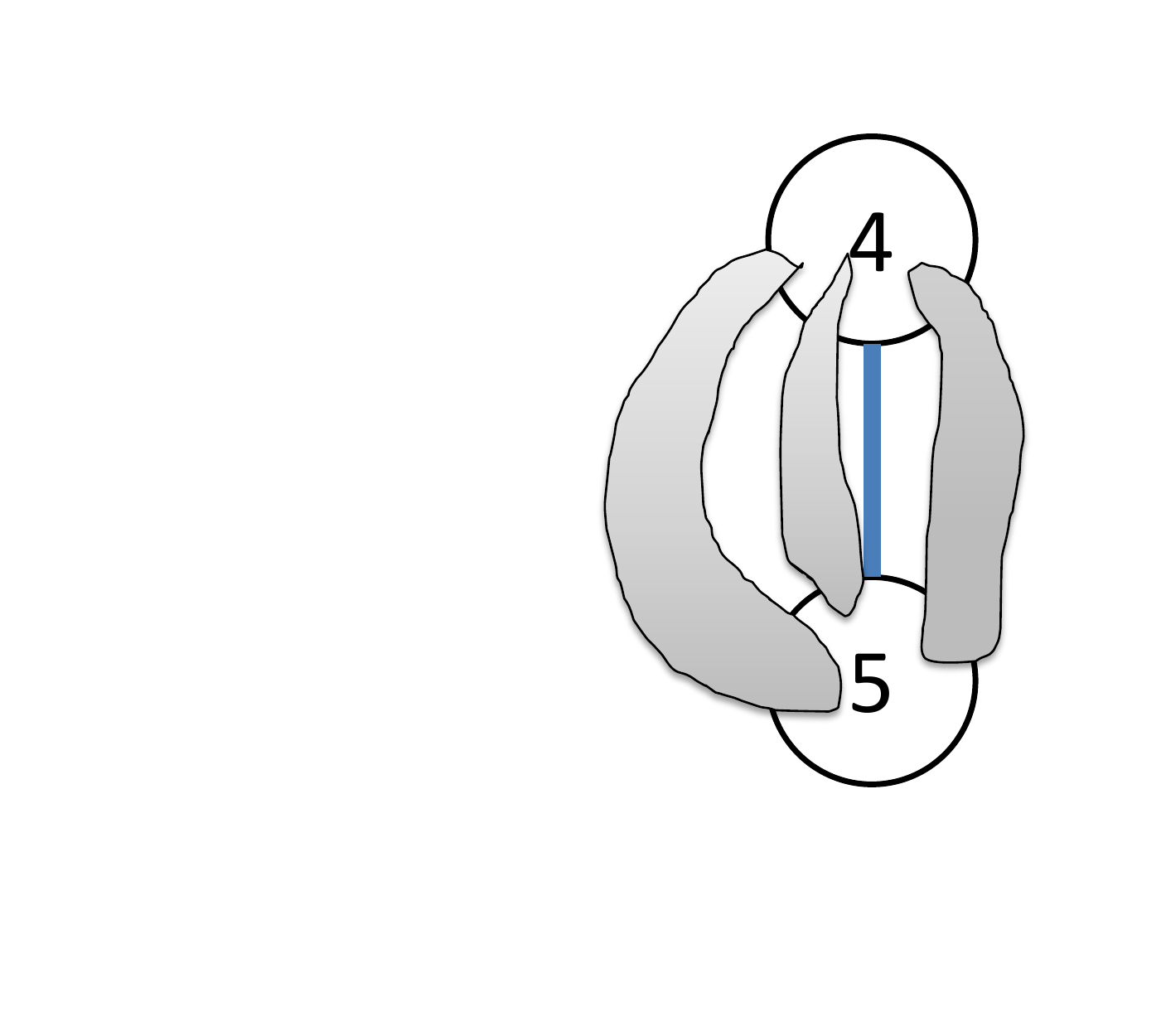}
}
\caption{\small Construction of hypergraphs (a) $\M_1$, (b) $\M_2$, (c) $\M_3$, (d) $\M_4$, with multiple hyperedges~$\{4,5\}$.}
\label{fig:ex-hypergraph-parallel}
\end{figure}

Let us note that the construction above allows ``\emph{selfloops},'' that is, hyperedges consisting of a single vertex, and that it allows multiple copies of hyperedges on the same vertex set (see, e.g., Figure \ref{fig:ex-hypergraph-3} (h) or Figure \ref{fig:ex-hypergraph-parallel}, and one could have many copies of hyperedges even with more than two vertices). An important feature of the latter case is that all these hyperedges will be considered as different hyperedges, since they correspond to different subgraphs of $H$ and have different labels. Note also that all labels are disjoint %in $\M_i$
(i.e., $\lab(\e_1) \cap \lab(\e_2) = \emptyset$ for any distinct hyperedges $\e_1, \e_2$ in $\M_i$).

%---------------------------------------------------------------------------------------------------------------------------------------------------------

\subsection{Shrinking copies of $H$ in $G$ (via safe vertices and consistent hypergraphs)}
\label{subsec:shrinking-copies-of-H}

\junk{
\Artur{\textbf{This section requires revision.} It's possibly the most important section in the paper that the readers must understand in details, and it seems to me that the text and presentation require some more efforts.}
}
\junk{
\Artur{Notation:
    \begin{compactitem}
%    \item \texttt{$\backslash$e} denotes $\e$, an edge/hyperedge
    \item \texttt{$\backslash$HQ} denotes $\HQ$ (also in $\HQ_i(\DDCH{i})$, which is a hypergraph constructed from $\DDCH{i}$ by contracting vertices of $H$)
    \item \texttt{$\backslash$N} denotes $\N$ (like in $\N_i^{\h}\langle u \rangle$, which is the set of neighbors of $u$ in $\h$ in the hypergraph $\HQ_i(\DDCH{i})$)
    \end{compactitem}
}
}
The central idea of our analysis is to mimic the corresponding transformation of $H$ (as described in Section \ref{subsec:shrinking-H}) in all relevant copies of $H$ in sets $\DDCH{1}, \DDCH{2}, \dots$, and then, instead of searching for a copy of $H$ in $\G[\DDCH{1}], \G[\DDCH{2}], \dots$, to search for copies of $\M_1, \M_2, \dots$ in the corresponding shrunk hypergraphs $\HQ_1(\DDCH{1}), \HQ_2(\DDCH{2}), \dots$. Then, we will argue that finding a copy of $H$ in $G$ is %the same
(almost) as easy
as finding a copy of $\M_1$ in $\HQ_1(\DDCH{1})$, which in turn can be reduced (by paying a small price) to finding a copy of $\M_2$ in $\HQ_2(\DDCH{2})$, and so on, reducing everything to finding a copy of $\M_{\Hsize}$ in $\HQ_{\Hsize}(\DDCH{\Hsize})$. And then, since $\M_{\Hsize}$ has only a single vertex, we would hope that finding its copy in $\HQ_{\Hsize}(\DDCH{\Hsize})$ is easy.

In order to incorporate this approach, we will transform appropriate subgraphs of $G$ into a sequence of hypergraphs, such that after $i$ transformations, every relevant copy of $H$ is shrunk into $\M_{i+1}$. (Let us emphasize that this step relies on the choice of vertex $v_i$ --- which is the same in all copies of $H$ --- to be determined by the structure of $\DDCH{i}$, as described in Lemma \ref{lemma:defining-Qi+1}.) In particular, we will mimic the corresponding transformation on $\DDCH{1}, \DDCH{2}, \dots$ as follows.

We consider a hypergraph, denoted by $\HQ_i(\DDCH{i})$, corresponding to $\DDCH{i}$, which has
%\Artur{We should remember that the hyperedges have labels $\lab(\e)$ and colored labels $\clab(\e)$.}
%
\begin{itemize}
\item vertex set $V(\HQ_i(\DDCH{i})) = V \setminus \{ u \in V: \chi(u) \in \{\chi(v_j): j < i\}\}$ (vertices\footnote{Let us first remind that we are assuming that the vertices of $G$ are colored using $\chi$ so that $G$ has at least $\Omega_{\eps,H}(|V|)$ edge-disjoint colored copies of $H$, as promised by Lemma \ref{lemma:edge-disjoint-copies-H-free}.} in $G$ that have colors of vertices $\{v_i, \dots, v_{\Hsize}\}$, that is, that have not been contracted in $\M_i$ yet), and
\item edge set formed by an edge-disjoint collection of copies of $\M_i$ (we allow hyperedges to have some multiplicity).
\end{itemize}
Then, for some carefully chosen set $\DDCH{i+1} \subseteq \DDCH{i}$, a new hypergraph $\HQ_{i+1}(\DDCH{i+1})$ is obtained from $\HQ_i(\DDCH{i})$ by
%\Artur{We should possibly mention that: ``to construct $\HQ_{i+1}(\DDCH{i+1})$, we need to know $v_i$ (and hence, by induction, $v_1, v_2, \dots, v_i$)''.}
%
\begin{itemize}
\item %first
    removing all hyperedges%\Artur{Is this notion/expression clear and well-defined? (This is a quite critical operation.)}
    corresponding to the edge-disjoint copies of $H$ in $\DDCH{i} \setminus \DDCH{i+1}$ and
\item then taking the set $\DDCH{i+1}$ of copies of $H$ and shrinking them, in the same way as $\M_i$ is transformed into $\M_{i+1}$:%\Artur{Figure/example?}

%---------------------------------------------------------------------------------------------------------------------------------------------------------
\begin{walgo}\vspace*{-0.3in}
\begin{itemize}[$\bullet$]
\item Select a vertex $v_i \in V(H)$.
\item Simultaneously, contract every vertex $u \in V(\HQ_i(\DDCH{i}))$ with $\chi(u) = \chi(v_i)$ as follows:
\begin{itemize}[$\diamond$]
\item for every colored copy $\h$ of $H$ in $\DDCH{i+1}$ that contains vertex $u$:
    \begin{itemize}[$\circ$]
    \item add a new hyperedge consisting of vertices in $\N_i^{\h}\langle u \rangle$, where $\N_i^{\h}\langle u \rangle$ is the set of neighbors of $u$ in $\h$ (in the hypergraph $\HQ_i(\DDCH{i})$) other than $u$ (that is, $u \notin \N_i^{\h}\langle u \rangle$);
    \end{itemize}
\item remove vertex $u$ (with all incident edges from $\HQ_i(\DDCH{i})$).
\end{itemize}
\end{itemize}
\end{walgo}
%---------------------------------------------------------------------------------------------------------------------------------------------------------

\end{itemize}

%One can think of $\HQ_i(\DDCH{i})$ as the hypergraph obtained by transformations $\M_1 \rightarrowtail \dots \rightarrowtail \M_i$, where we consider separately all copies of $H$ defining $\DDCH{i}$, with the global choice of vertices $v_1, \dots, v_{i-1} \in V(H)$, which are the same for each copy of $H$ from $\DDCH{i}$.

Notice that in our construction of $\HQ_{i+1}(\DDCH{i+1})$ we are removing all vertices $u \in V$ with color $\chi(u) = \chi(v_i)$. And so, in particular, $V(\HQ_{i+1}(\DDCH{i+1})) = V \setminus \{ u \in V: \chi(u) \in \{\chi(v_j): j \le i\}\}$.

%---------------------------------------------------------------------------------------------------------------------------------------------------------

\begin{figure}[t]
\centerline
{(a) \includegraphics[width=.3\textwidth]{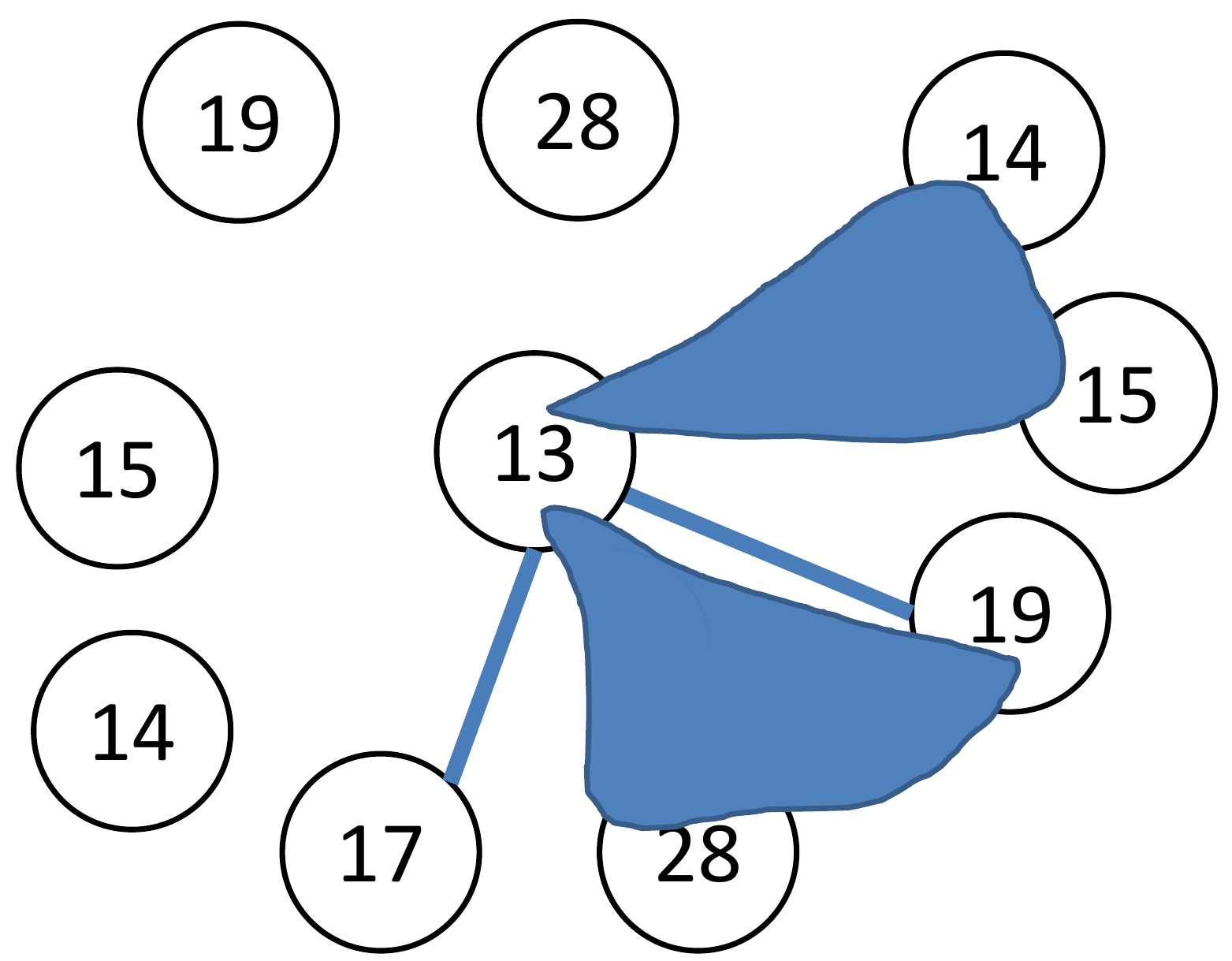}
(b) \includegraphics[width=.3\textwidth]{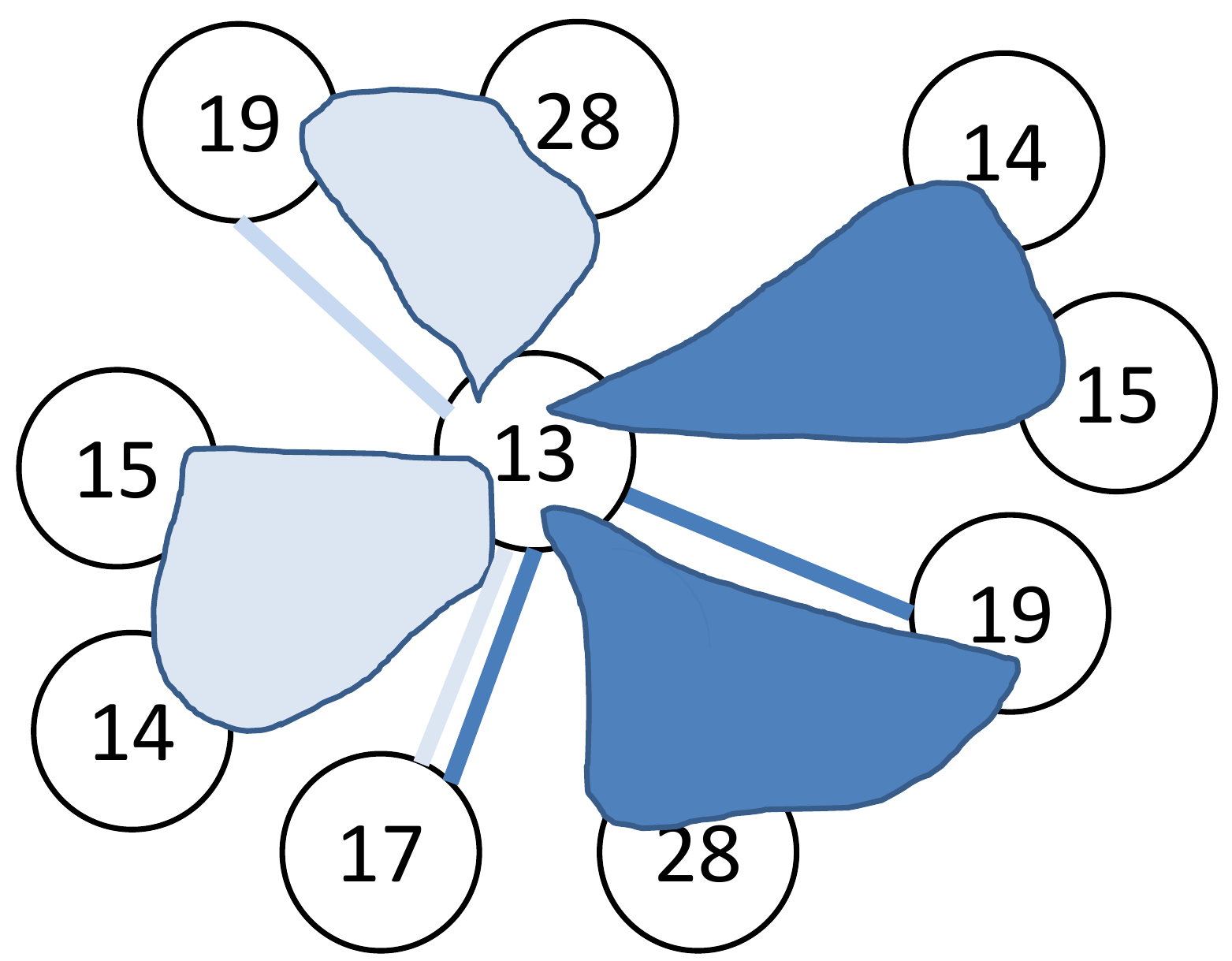}
(c) \includegraphics[width=.3\textwidth]{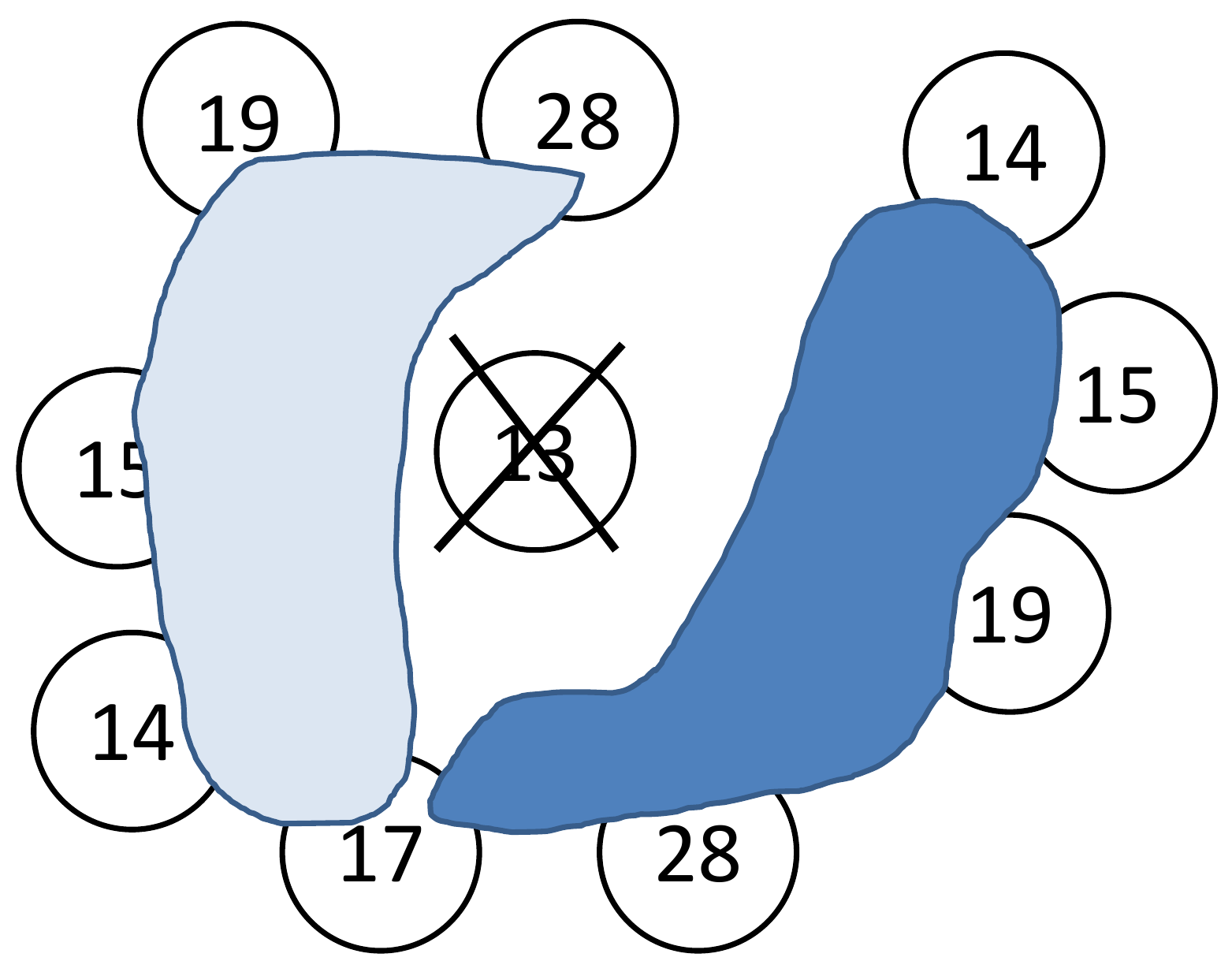}
}
\caption{\small Operation of contracting vertex of color 13 (numbers depicted correspond here to the colors) to define new hypergraph $\HQ_{14}(\DDCH{14})$.
(a) Describes the edges of a \emph{single copy} of $\M_{13}$ incident to vertex of color 13 (in the center) in $\HQ_{13}(\DDCH{13})$.
(b) Describes the edges of \emph{two copies} of $\M_{13}$ incident to vertex of color 13 in $\HQ_{13}(\DDCH{13})$. (Notice that in this case, vertex of color 13 is not safe.)
(c) Describes the situation of contracting vertex of color 13, which creates two new hyperedges, and removal of vertex of color 13 and all incident edges. This defines $\HQ_{14}(\DDCH{14})$.
}
\label{fig:contraction}
\end{figure}

%---------------------------------------------------------------------------------------------------------------------------------------------------------

Furthermore, since we contract only vertices of color $\chi(v_i)$ and since these vertices are independent in $\G[\DDCH{i+1}]$ (indeed, since $\DDCH{i+1}$ is the set of edge-disjoint colored copies of $H$, $\G[\DDCH{i+1}]$ does not have monochromatic edges), the operation above is well defined and the contractions of all vertices of color $\chi(v_i)$ can be performed independently in all copies of $H$ in $\DDCH{i+1}$. This yields an equivalent definition:

%---------------------------------------------------------------------------------------------------------------------------------------------------------
\begin{remark}
\label{remark:independent-shrinking}
The following is an equivalent definition of $\HQ_{i+1}(\DDCH{i+1})$:
\begin{itemize}
\item Start with graph $\G[\DDCH{i+1}]$.
\item For every copy $\h$ of $H$ in $\DDCH{i+1}$, perform the shrinking of $H$ into hypergraph $\M_{i+1}$.
\item Combine all copies of $\M_{i+1}$ obtained in that way.
\item Remove all vertices $u$ with $\chi(u) = \chi(v_j)$ for $j \le i$ that do not belong to any copy of $\DDCH{i+1}$.
\end{itemize}
The fact that this description is correct follows from the fact that the shrinking of different copies of $H$ can be performed independently because of vertex coloring, which ensures that if we contract a vertex $u$ with $\chi(u) = \chi(v_j)$ and create a new edge $\N_j^{\h}\langle u \rangle$, then this construction can be performed independently for different copies of $H$.

Notice that (as formally proven in Claim~\ref{claim:final-hypergraph-invariant} in Appendix \ref{subsec:properties-of-consistent-hypergraphs}) because of the construction above, to define $\HQ_{i+1}(\DDCH{i+1})$, we do not need to consider the constructions of $\HQ_{1}(\DDCH{1}), \HQ_{2}(\DDCH{2}), \dots, \HQ_{i}(\DDCH{i})$ one after another, but we could do it with the constructions of $\HQ_{1}(\DDCH{i+1}), \HQ_{2}(\DDCH{i+1}), \dots, \HQ_{i}(\DDCH{i+1})$, and from $\HQ_{i}(\DDCH{i+1})$ to build $\HQ_{i+1}(\DDCH{i+1})$.

(Note that the vertex set of $\HQ_{i+1}(\DDCH{i+1})$ is $V(\HQ_{i+1}(\DDCH{i+1})) = V \setminus \{ u \in V: \chi(u) \in \{\chi(v_j): j \le i\}\}$. Further, observe that $\HQ_{i+1}(\DDCH{i+1})$ may have (isolated) vertices $u$ that do not belong to any copy of $\DDCH{i+1}$.)
\end{remark}

The construction above maintains a relationship between edges in $\HQ_{i}(\DDCH{i})$ and edges in $\M_i$.

\begin{definition} \textbf{(Corresponding edges)}
\label{def:corresponding-edges}
If $\e$ is an edge in $\HQ_{i}(\DDCH{i})$ then the \emph{corresponding edge to $\e$ in $\M_i$} is edge $\e'$ in $\M_i$ such that the colors of vertices in $\e$ are the same as the colors of vertices in $\e'$ (i.e., $\{\chi(x): x \in \e\} = \{\chi(v_j): v_j \in \e'\}$), and the colored labels of $\e$ and $\e'$ are the same too (i.e., $\clab(\e) = \clab(\e^*)$).
\end{definition}

Notice that every edge in $\HQ_{i}(\DDCH{i})$ has a unique corresponding edge in $\M_i$. Furthermore, for any edge $\e'$ in $\M_i$, the number of edges in $\HQ_{i}(\DDCH{i})$ corresponding to edge $\e'$ in $\M_i$ is exactly $|\DDCH{i}|$.

Next, we can also mimic Definition \ref{def:modeling-hyperedge-in-Mi+1} in the context of our construction here as follows:

\begin{definition} \textbf{(Modeling edges in $\HQ_{i+1}(\DDCH{i+1})$ by edges in $\HQ_i(\DDCH{i})$)}
\label{def:modeling-hyperedge-in-Hi+1}
Let $u$ be a vertex in $\HQ_i(\DDCH{i})$ with $\chi(u) = \chi(v_i)$. Let $\h$ be a colored copy of $H$ in $\DDCH{i+1}$ that contains vertex $u$. Let $\e_1, \dots, e_{\ell}$ be the edges incident to $u$ $\HQ_i(\DDCH{i})$ corresponding to the copy $\h$. Then, we will say that the newly created hyperedge $\N_i^{\h}\langle u \rangle$ in $\HQ_{i+1}(\DDCH{i+1})$ is \emph{modeled by edges $\e_1, \dots, \e_{\ell}$} in $\HQ_i(\DDCH{i})$.
\end{definition}

\medskip

Now, we are ready to formalize the process of finding a colored copy of $\M_i$ in a hypergraph.

\begin{definition} \textbf{(Finding a colored copy of $\M_i$)}
\label{def:finding-colored-copies-of-Mi}
Let $v_i, \dots, v_{\Hsize}$ be the vertices in $\M_i$. We say that \HRLBD\,${(\HQ,\P,\M_i,\dg,\ld)}$ finds a colored copy of $\M_i$ in $\HQ$ if the corresponding algorithm \HRLBFS\,${(\HQ,\P,\dg,\ld)}$ returns a set of edges $\mathcal{E}$, such that
\begin{itemize}
\item the sub-hypergraph of $\HQ$ induced by the edges $\mathcal{E}$ contains vertices $x_i, \dots, x_{\Hsize}$ such that for every edge/hyperedge $\e$ in $\M_i$, $\mathcal{E}$ contains an edge corresponding to $\e$, or equivalently,
    \begin{compactitem}[$\diamond$]
    \item $\chi(x_j) = \chi(v_j)$ for every $j$, $i \le j \le \Hsize$, and
    \item for every edge $\{v_{j_1}, \dots, v_{j_r}\}$ in $\M_i$, $\mathcal{E}$ contains edge $\{x_{j_1}, \dots, x_{j_r}\}$.
    \end{compactitem}
\end{itemize}
\end{definition}

%---------------------------------------------------------------------------------------------------------------------------------------------------------

\subsubsection{Adjusting for planar graphs: safe vertices and consistent hypergraphs}
\label{subsubsec:safe-vertices-and-consistent-hypergraphs}
In our construction we will require more properties from the contractions defining $\HQ_{i+1}(\DDCH{i+1})$. To maintain \emph{some basic properties of planar graphs} (which are required by our analysis), we will want to model the operation of contraction of a vertex $u$ as the standard vertex contraction of $u$ to one of its neighbors, cf. Appendix \ref{subsec:planarization-of-hypergraphs}. For that, we will need an additional, \emph{stronger property}:
\begin{itemize}
\item we want to ensure that all contractions in $\HQ_i(\DDCH{i})$ corresponding to the contraction of $v_i$ in $\M_i$ are \emph{consistent}, that is, the contraction of $u$ is the same in every colored copy of $H$ that contains~$u$ (that is, for every vertex $u$ in with $\chi(u) = \chi(v_i)$, for any two colored copies $\h_1,\h_2$ of $H$ in $\DDCH{i+1}$ containing vertex $u$, we have $\N_i^{\h_1}\langle u \rangle = \N_i^{\h_2}\langle u \rangle$).
\end{itemize}

%---------------------------------------------------------------------------------------------------------------------------------------------------------

\begin{figure}[t]
\centerline
{(a) \includegraphics[width=.35\textwidth]{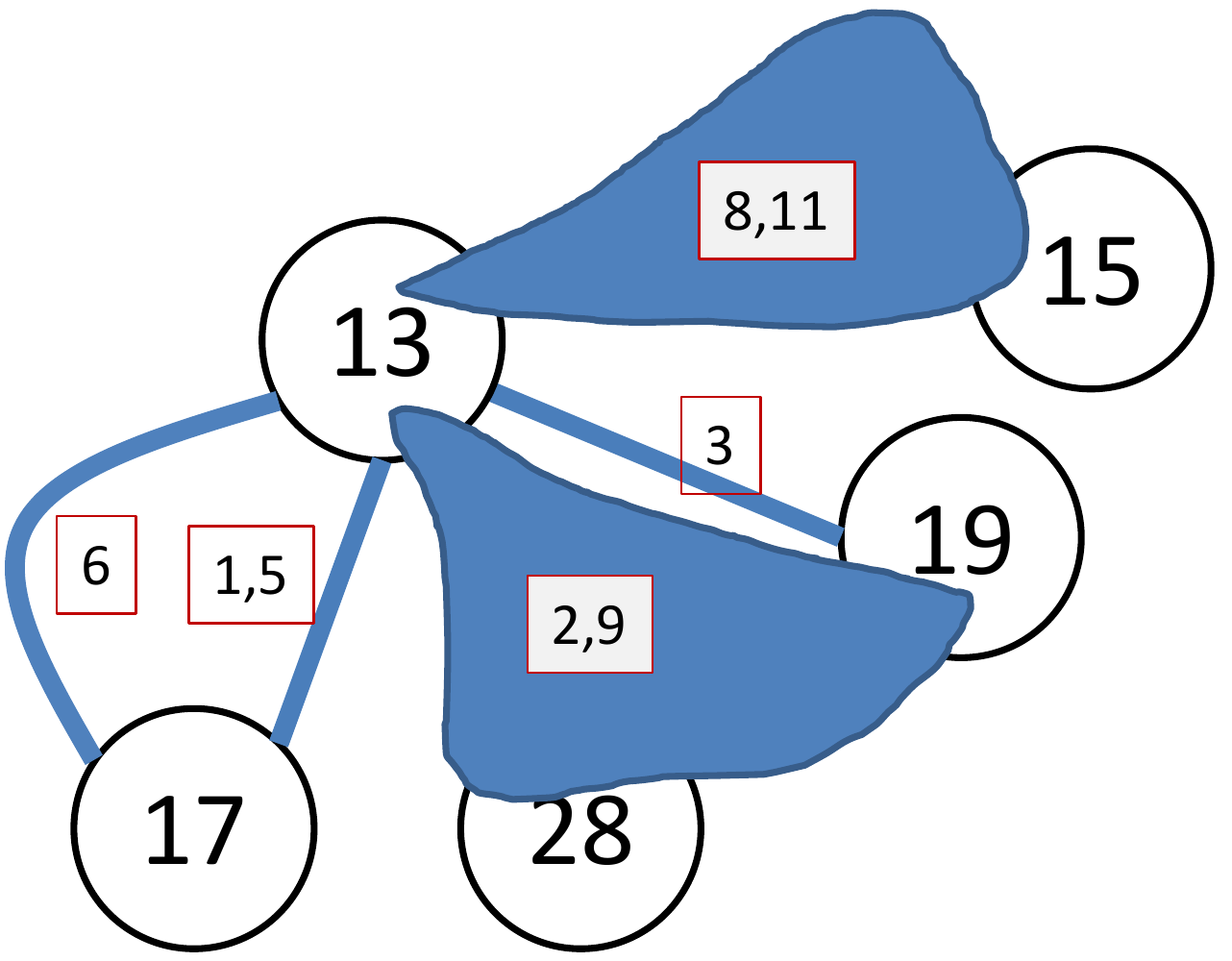}
\qquad
(b) \includegraphics[width=.35\textwidth]{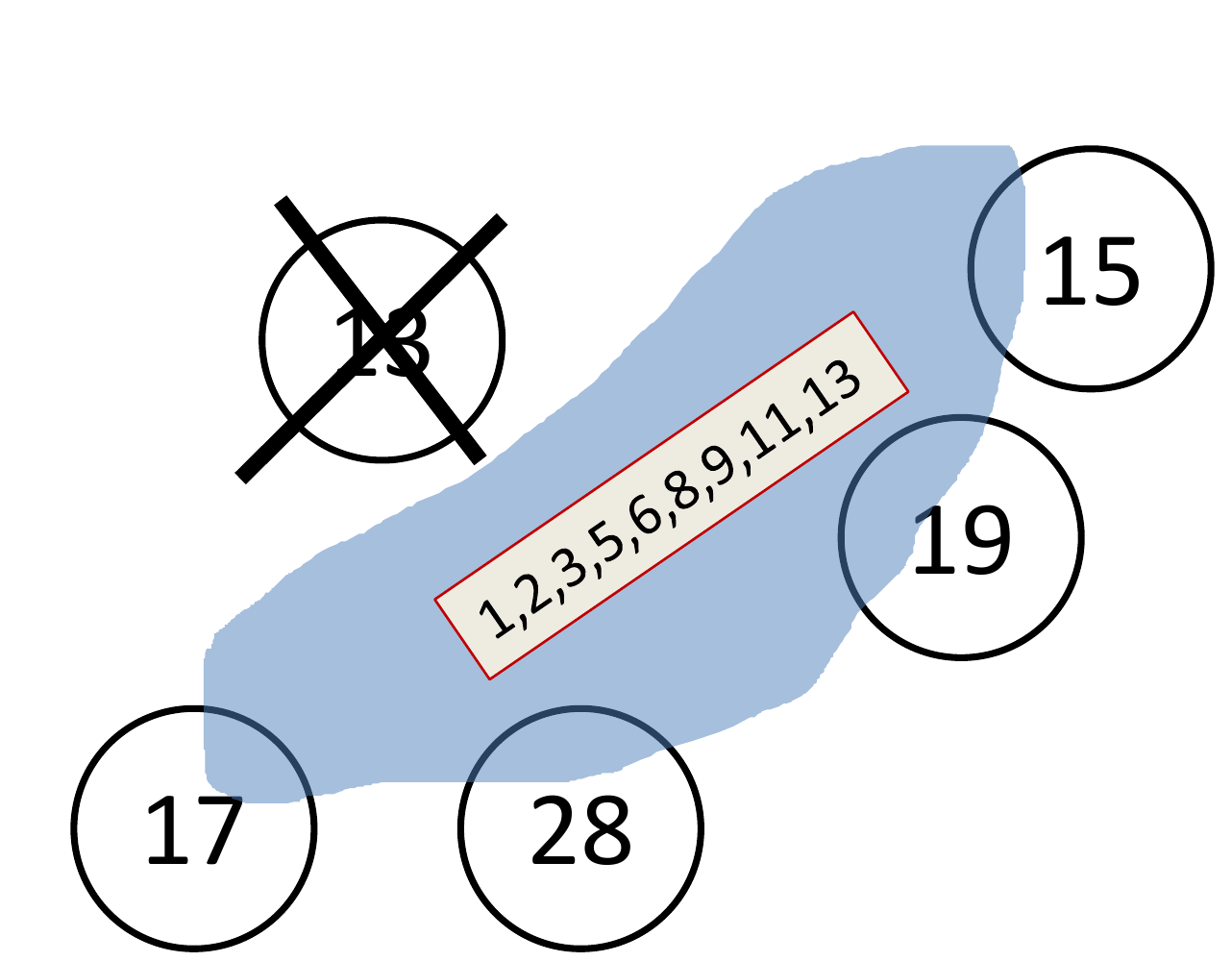}
}
\caption{\small Operation of contracting safe vertex of color 13 (numbers depicted in the vertices correspond to the colors, and number depicted next to the edges correspond to the colored labels of the edges) to define new hypergraph $\HQ_{14}(\DDCH{14})$.
(a) Describes the edges of a single copy of $\M_{13}$ incident to vertex of color 13 (in the center) in $\HQ_{13}(\DDCH{13})$. Since we want vertex of color 13 to be safe, it is possible there are many more copies of $\M_{13}$ incident to that vertex, in which case all of them use identical edges as in the depicted single copy (here identical means: on the same vertex set, with the same colored labels, but with distinct labels).
(b) Describes the situation after contracting vertex of color 13 in $\HQ_{14}(\DDCH{14})$. Notice that the new edge in $\HQ_{14}(\DDCH{14})$ has colored label $\{1,2,3,5,6,8,9,11,13\}$ and is modeled by five edges in $\HQ_{13}(\DDCH{13})$.
}
\label{fig:contraction-safe}
\end{figure}

%---------------------------------------------------------------------------------------------------------------------------------------------------------

To facilitate this property, we will use the following definitions.%\Artur{We will have some figures here, of \emph{safe vertices} and of \emph{consistent hypergraphs}.}

\begin{definition}\textbf{(Safe vertices)}
\label{def:safe-vertices}
Let $\DDCH{i}$ be a set of edge-disjoint colored copies of $H$ in $G$ and let $\DCH \subseteq \DDCH{i}$. We call a vertex $u \in V(\HQ_i(\DDCH{i}))$ \textbf{\emph{safe}} (with respect to $\DCH$ and $\HQ_i(\DDCH{i})$) if for all colored copies $\h \in \DCH$ of $H$ that contain $u$, the sets $\N_i^{\h}\langle u\rangle$ are the same.
\end{definition}

\begin{remark}
\label{remark-properties-of-safe-vertices}
Note that Definition \ref{def:safe-vertices} means that for every safe vertex $u$ with respect to $\DCH$ and $\HQ_i(\DDCH{i})$, not only all edges incident to $u$ correspond to the edges from $\M_i$ incident to vertex $v$ in $\M_i$ with $\chi(u) = \chi(v)$, but also, if $u$ is incident to $r$ edges in $\HQ_i(\DDCH{i})$ and $v$ is incident to edges $\e_1, \dots, \e_{\ell}$ in $\M_i$, then
\begin{compactenum}[\it (i)]
%\begin{inparaenum}[\it (i)]
\item we can partition the edges incident to $u$ into $\ell$ groups, each group corresponding to one of the edges $\e_1, \dots, \e_{\ell}$ in $\M_i$, each group of the same size $r / \ell$, such that two edges $\e', \e''$ from the same group have the same colored label (i.e., $\clab(\e') = \clab(\e'')$) and are defined by the same vertices (i.e., for every vertex $x$, $x \in \e'$ iff $x \in \e''$);
\item $|\N_i^{\h}\langle u \rangle| = |\bigcup_{j=1}^{\ell} \e_j \setminus \{v\}|$, that is, $u$ has as many neighbors in $\HQ_i(\DDCH{i})$ as $v$ has in $\M_i$;
\item $\{\chi(x): x \in \N_i^{\h}\langle u \rangle\} = \{\chi(x): x \in \bigcup_{j=1}^{\ell} \e_j \setminus \{v\} \}$.
%\Artur{I'm not sure if this is clear; maybe I need to draw some figures. That's easy, but I would prefer to not write a complete and very detailed proof.}
%\end{inparaenum}
\end{compactenum}
\end{remark}

Our next iterative definition extends the notion of safe vertices to the entire hypergraph.

\begin{definition}\textbf{(Consistent hypergraphs)}
\label{def:consistent-hypergraphs}
For any set $\DDCH{1}$ of edge-disjoint colored copies of $H$ in $G$, the hypergraph $\HQ_1(\DDCH{1})$ (which is equal to the graph $\G[\DDCH{1}]$) is called \textbf{\emph{consistent}} (for $\DDCH{1}$).

Let $\DDCH{i}$ be a set of edge-disjoint colored copies of $H$ in $G$ and let $\DDCH{i+1} \subseteq \DDCH{i}$. If hypergraph $\HQ_i(\DDCH{i})$ is consistent for $\DDCH{i}$, then hypergraph $\HQ_{i+1}(\DDCH{i+1})$ obtained from $\HQ_i(\DDCH{i})$ is called \textbf{\emph{consistent}} (for $\DDCH{i+1}$) if every vertex $u \in V(\HQ_i(\DDCH{i}))$ with $\chi(u) = \chi(v_i)$ is safe with respect to $\DDCH{i+1}$ and $\HQ_i(\DDCH{i})$.
\end{definition}

In Appendix \ref{subsec:properties-of-consistent-hypergraphs} we will show some basic properties of consistent hypergraphs used later in our analysis.

%---------------------------------------------------------------------------------------------------------------------------------------------------------

\subsubsection{Central property of consistent hypergraphs via shadow graphs}
With the notion of safe vertices and consistent hypergraphs, we can now present the following central lemma that shows that the neighborhood of vertices in consistent hypergraphs can be modeled by some semi-planar structures, which we will call \emph{shadow graphs}, that are a union of at most $\Hsize$ simple planar graphs.

\begin{restatable}{lemma}{LSG}
\label{lemma:central-small-degrees}
Let $\DDCH{i}$ be a set of edge-disjoint colored copies of $H$ in $G$ and let $\HQ_i(\DDCH{i})$ be a hypergraph consistent for $\DDCH{i}$. Then, there is a simple graph $\Gi$,
%with the vertex set equal to the set of all non-isolated vertices in $\HQ_i(\DDCH{i})$ that is a union of at most $\Hsize$ simple planar graphs such that for any distinct $x, y \in V(\HQ_i(\DDCH{i}))$, $x$ is adjacent to $y$ in $\HQ_i(\DDCH{i})$ if and only if $x$ is adjacent to $y$ in $\Gi$.
%
%\begin{enumerate}[(a)]
\begin{compactenum}[\quad (a)]
%\begin{inparaenum}[(a)]
\item with the vertex set equal to the set of all non-isolated vertices in $\HQ_i(\DDCH{i})$,
\item that is a union of at most $\Hsize$ simple planar graphs, and
\item such that for any distinct $x, y \in V(\HQ_i(\DDCH{i}))$, $x$ is adjacent to $y$ in $\HQ_i(\DDCH{i})$ if and only if $x$ is adjacent to $y$ in $\Gi$.
    %\Artur{As it's defined now, $\Gi$ has vertex set equal to the set of all non-isolated vertices in $\HQ_i(\DDCH{i})$; if we really wanted, we could trivially replace it by $V$, but then one possibly should adjust accordingly some later text.}
%\end{inparaenum}
\end{compactenum}
%\end{enumerate}
%
\end{restatable}

The simple graph $\Gi$ in Lemma \ref{lemma:central-small-degrees} will be called the \emph{shadow graph} of $\HQ_i(\DDCH{i})$.

We consider the characterization provided in Lemma \ref{lemma:central-small-degrees} to be one of the most interesting and highly non-trivial contributions of this paper. This is the key tool that allows us to facilitate the approach presented in the paper. To simply the flow of the paper though, the proof of Lemma \ref{lemma:central-small-degrees} is deferred to Appendix \ref{subsec:planarization-of-hypergraphs}.

%---------------------------------------------------------------------------------------------------------------------------------------------------------

\subsubsection{Finding many safe vertices of the same color}
\label{subsec:lemma-many-safe-vertices}

%---------------------------------------------------------------------------------------------------------------------------------------------------------

The main use of Lemma \ref{lemma:central-small-degrees} is to show that even though the use of the hypergraphs $\HQ_1(\DDCH{1}), \HQ_2(\DDCH{2}),$ $\dots$ looses some basic properties of planar graphs, our use of consistent hypergraphs allows us to apply Lemma \ref{lemma:central-small-degrees} to maintain some weaker, but still similar properties of the hypergraphs $\HQ_1(\DDCH{1}), \HQ_2(\DDCH{2}), \dots$. We begin with the following lemma that shows that the hypergraphs will have a constant fraction of vertices of low degrees. The proof of our next Lemma \ref{lemma:small-vertices-new} extends the approach used earlier in the context of planar graphs from \cite{CMOS11}; we defer the proof to Appendix \ref{sec:proof-lemma:small-vertices-new}.

\begin{restatable}{lemma}{LSV}
\label{lemma:small-vertices-new}
Let $\DDCH{i}$ be a set of edge-disjoint colored copies of $H$ in $G$ and let $\HQ_i(\DDCH{i})$ be a hypergraph consistent for $\DDCH{i}$. Then, there is a set $\DCH \subseteq \DDCH{i}$ of size at least $\frac{|\DDCH{i}|}{4\Hsize+2}$ such that in the hypergraph $\HQ_i(\DCH)$, every copy of $H$ in $\DCH$ has a vertex %in $V(\HQ_i(\DCH))$
with at most $6 \Hsize$ distinct neighbors.
%for which there is a color $\cc$ (with $\cc \in \{1, \dots, \Hsize\} \setminus \{\chi(v_j): j < i\}$) such that every vertex of color $\cc$ in $V(\HQ_i(\DDCH{i}))$ has at most $6 \Hsize$ distinct neighbors in the hypergraph $\HQ_i(\DCH)$.
\end{restatable}

Our next lemma follows the arguments used in a related proof from \cite{CMOS11} and shows that if there is a color with all vertices having a small number of neighbors in $\HQ_i(\DCH)$ for $\DCH \subseteq \DDCH{i}$, then we can always find a large subset of $\DCH$ with all vertices of that color being safe.

\begin{lemma}
\label{lemma:if-low-deg-then-many-safe}
Let $\DDCH{i}$ be a set of edge-disjoint colored copies of $H$ in $G$ such that $\HQ_i(\DDCH{i})$ is a hypergraph consistent for $\DDCH{i}$. Let $\cc$ be a color of a vertex in $\{1, \dots, \Hsize\} \setminus \{\chi(v_j): j < i\}$. Let $\DCH \subseteq \DDCH{i}$ such that every colored copy of $H$ in $\DCH$ has vertex of color $\cc$ with at most $6 \Hsize$ distinct neighbors in $\HQ_i(\DCH)$. Then there is a subset $\DCH' \subseteq \DCH$, $|\DCH'| \ge \frac{|\DCH|}{(6 \Hsize)^{\Hsize}}$, such that every colored copy $\h$ of $H$ in $\DCH'$ has vertex of color $\cc$ safe with respect to $\DCH'$ and $\HQ_i(\DDCH{i})$.
\end{lemma}

\begin{proof}
Let $\cc_1, \dots, \cc_{\ell}$ be the colors of vertices adjacent to vertex of color $\cc$ in $\M_i$ (notice that $\cc$ may be among these colors).
For each non-isolated vertex $u$ in $\HQ_i(\DCH)$ of color $\cc$, for every color $\cc_s$, $1 \le s \le \ell$, select i.u.r. one of its neighbors $u_{\langle s \rangle}$ in $\HQ_i(\DCH)$ of color $\cc_s$. Next, remove from $\DCH$ every copy of $\h$ of $H$ in $\HQ_i(\DCH)$ containing vertex $u$ unless the vertices from this copy incident to $u$ are the selected $\ell$ neighbors $u_{\langle 1 \rangle}, u_{\langle 2 \rangle}, \dots, u_{\langle \ell \rangle}$. Let $\DCH'$ be the set of remaining copies of $H$ in $\HQ_i(\DCH)$.

Our construction ensures that every remaining non-isolated vertex $u$ of color $\cc$ is safe with respect to $\DCH'$ and $\HQ_i(\DDCH{i})$. Furthermore, since every vertex of color $\cc$ has at most $6 \Hsize$ distinct neighbors (taking into account self-loops) in $\HQ_i(\DCH)$, the probability that a fixed copy of $\h$ in $\DCH$ is not deleted by the process above is at least $(6 \Hsize)^{-\ell}$. Therefore the expected size of $\DCH'$ is at least $(6 \Hsize)^{-\ell} \cdot |\DCH|$, and therefore, there exists a set $\DCH'$ of that size that satisfies the lemma.
\end{proof}

With Lemmas \ref{lemma:small-vertices-new} and \ref{lemma:if-low-deg-then-many-safe} at hand, we are now ready to present the main result of this section.

\begin{lemma}
\label{lemma:many-safe}
Let $\DDCH{i}$ be a set of edge-disjoint colored copies of $H$ in $G$ and let $\HQ_i(\DDCH{i})$ be a hypergraph consistent for $\DDCH{i}$. Then, there is color $\cc$ in $\{1, \dots, \Hsize\} \setminus \{\chi(v_j): j < i\}$ and a set $\DCH^* \subseteq \DDCH{i}$ of size at least $\frac{|\DDCH{i}|}{(6 \Hsize)^{\Hsize+2}}$ such that every colored copy $\h$ of $H$ in $\DCH^*$ has vertex of color $\cc$ safe with respect to $\DCH^*$ and $\HQ_i(\DDCH{i})$.
\end{lemma}

\begin{proof}
By Lemma \ref{lemma:small-vertices-new}, there is a set $\widehat{\DCH} \subseteq \DDCH{i}$, $|\widehat{\DCH}| \ge \frac{|\DDCH{i}|}{4\Hsize+2}$, such that every colored copy of $H$ in $\widehat{\DCH}$ has a vertex with at most $6 \Hsize$ distinct neighbors in $\HQ_i(\widehat{\DCH})$. For a color $\cc^* \in \{1, \dots, \Hsize\} \setminus \{\chi(v_j): j < i\}$, let $\widehat{\DCH}_{\cc^*}$ be the subset of $\widehat{\DCH}$ such that every copy of $H$ in $\widehat{\DCH}_{\cc^*}$ has a vertex of color $\cc^*$ with at most $6 \Hsize$ distinct neighbors in the hypergraph $\HQ_i(\widehat{\DCH})$. Since $\bigcup_{\cc^*% \in \{1, \dots, \Hsize\} \setminus \{\chi(v_j): j < i\}
} \widehat{\DCH}_{\cc^*} = \widehat{\DCH}$, there is one color $\cc \in \{1, \dots, \Hsize\} \setminus \{\chi(v_j): j < i\}$ such that $|\widehat{\DCH}_{\cc}| \ge \frac{1}{\Hsize} \cdot |\widehat{\DCH}| \ge \frac{|\DDCH{i}|}{(4\Hsize+2) \cdot \Hsize} \ge \frac{|\DDCH{i}|}{(6\Hsize)^2}$ and every copy of $H$ in $\widehat{\DCH}_{\cc}$ has a vertex of color $\cc$ with at most $6 \Hsize$ distinct neighbors in $\HQ_i(\widehat{\DCH})$, and hence also in $\HQ_i(\widehat{\DCH}_{\cc})$. Therefore, we can take such set $\widehat{\DCH}_{\cc}$ as set $\DCH$ in Lemma \ref{lemma:if-low-deg-then-many-safe}, to conclude that there is a subset $\DCH' \subseteq \widehat{\DCH}_{\cc}$, $|\DCH'| \ge \frac{|\widehat{\DCH}_{\cc}|}{(6 \Hsize)^{\Hsize}} \ge %\frac{|\widehat{\DCH}_{\cc}|}{(6 \Hsize)^{\Hsize} \cdot (4\Hsize+2) \cdot \Hsize} \ge
\frac{|\DDCH{i}|}{(6 \Hsize)^{\Hsize+2}}$, such that every colored copy $\h$ of $H$ in $\DCH'$ has vertex of color $\cc$ safe with respect to $\DCH'$ and $\HQ_i(\DDCH{i})$.
\end{proof}

%---------------------------------------------------------------------------------------------------------------------------------------------------------

\subsection{Constructing set $\DDCH{i+1}$ of edge-disjoint colored copies of $H$ %(and choice of $v_i$)
and $\HQ_{i+1}(\DDCH{i+1})$}
\label{subsec:main-reduction}

Now we are ready to define our construction of the set $\DDCH{i+1}$ of edge-disjoint colored copies of $H$ obtained as a subgraph of $\DDCH{i}$, and with this, to define the hypergraph $\HQ_{i+1}(\DDCH{i+1})$ from $\HQ_i(\DDCH{i})$.

Let $\DDCH{i}$ be a set of edge-disjoint colored copies of $H$ in $G$, where $\HQ_i(\DDCH{i})$ is a hypergraph consistent for $\DDCH{i}$. We apply Lemma \ref{lemma:many-safe} to choose color $\cc$ in $\{1, \dots, \Hsize\} \setminus \{\chi(v_j): j < i\}$ and a set $\DCH^* \subseteq \DDCH{i}$ of size at least $\frac{|\DDCH{i}|}{(6 \Hsize)^{\Hsize+2}}$ such that every colored copy $\h$ of $H$ in $\DCH^*$ has vertex of color $\cc$ safe with respect to $\DCH^*$ and $\HQ_i(\DDCH{i})$ (that is, for every vertex $u$ with $\chi(u) = \cc$, all colored copies $\h \in \DDCH{i+1}$ of $H$ that contain $u$ have identical sets $\N_i^{\h}\langle u\rangle$ in $\HQ_i(\DDCH{i})$). Then, we define $\DDCH{i+1} := \DCH^*$ and select vertex $v_i$ to be the vertex of color $\cc$ in $H$.

With so defined vertex $v_i$, we can immediately construct the hypergraph $\HQ_{i+1}(\DDCH{i+1})$ (from the hypergraph $\HQ_i(\DDCH{i})$). The details of the construction have been presented in Section \ref{subsec:shrinking-copies-of-H}, and it required the choice of set $\DDCH{i+1}$ and of vertex $v_i$ among the vertices in $V(H) \setminus \{v_1, \dots, v_{i-1}\}$.

By Lemma \ref{lemma:many-safe} (cf. Definition \ref{def:consistent-hypergraphs} of consistent hypergraphs), this immediately gives the following lemma.

\begin{lemma}
\label{lemma:defining-Qi+1}
Let $\DDCH{i}$ be a set of edge-disjoint colored copies of $H$ in $G$ and let $\HQ_i(\DDCH{i})$ be a hypergraph consistent for $\DDCH{i}$. Then, the choice of the set $\DDCH{i+1}$ with the vertex $v_i$, as described above, will ensure that $|\DDCH{i+1}| \ge \frac{|\DDCH{i}|}{(6 \Hsize)^{\Hsize+2}}$ and that $\HQ_{i+1}(\DDCH{i+1})$ obtained from $\HQ_i(\DDCH{i})$ is consistent for~$\DDCH{i+1}$.
\end{lemma}

%---------------------------------------------------------------------------------------------------------------------------------------------------------

\subsubsection{Representatives $\P_i$ for $\DCH$ and $\HQ_i(\DCH)$}
\label{subsubsec:reps}

In our analysis, we will be also using the concept of representatives to describe the scenario that a vertex from $V$ has been contracted to some other vertices during the construction of $\HQ_i(\DCH)$ (in some moment, it has been deleted from $\HQ_j(\DCH)$, $1 \le j < i$, and new hyperedges containing all neighbors of this vertex has been formed, in which case of these neighbors is used as a proxy). The canonical representative function plays an important role in our analysis and it is used explicitly in algorithms \HRLBD{} and \HRLBFS.
(For the following definition, let us recall the construction of the hypergraph $\HQ_i(\DDCH{i})$ from Section \ref{subsec:shrinking-copies-of-H}. Let us also notice that the notion of canonical representatives is used solely in the analysis at the end of the process, and since it is not used for the construction of sets $\DDCH{1}, \DDCH{2}, \dots, \DDCH{\Hsize}$ and hypergraphs $\HQ_1(\DDCH{1}), \HQ_2(\DDCH{2}), \dots, \HQ_{\Hsize}(\DDCH{\Hsize})$ and  is used only to model their behavior, it does rely on the final order $v_1, \dots, v_{\Hsize}$ of the vertices in $H$.)

\begin{definition}\textbf{(Canonical representatives)}
\label{def:canonical-representative-function}
Let $\DCH$ be a set of edge-disjoint colored copies of $H$ in $G$. Let $v_1, \dots, v_{\Hsize}$ be an arbitrary order of vertices of $H$ such that for each $i$, $1 \le i \le \Hsize$, the hypergraph $\HQ_i(\DCH)$ is consistent for $\DCH$.
A \textbf{\emph{canonical representative function}} is a sequence of functions $\P_1, \P_2, \dots, \P_{\Hsize}: V \rightarrow V$ such that for every $i$, $1 \le i \le \Hsize$:
\begin{itemize}
\item if $u$ is an isolated vertex in $\G[\DCH]$, then $\P_i(u) = u$ for every $i$;
\item otherwise, if $u$ is a vertex in $\HQ_i(\DCH)$ {\small (i.e., $\chi(u) \notin \{\chi(v_j): 1 \le j < i\}$)}, then $\P_i(u) = u$;
\item otherwise, %\footnote{That is, $u$ is a not vertex in $\HQ_i(\DCH)$ and is not isolated in $\G[\DCH]$, and so, there is an edge $\e$ in $\HQ_i(\DCH)$ with $u \in \lab(\e)$.},
    $\P_i(u) = x$, where
    \begin{inparaenum}[\it (i)]
    \item $x \in \bigcup_{\e: u \in \lab(\e)} \e$ and
    \item for any $x, y \in \bigcup_{\e: u \in \lab(\e)} \e$, if $x \ne y$, $\chi(x) = \chi(v_{j_1})$, and $\chi(y) = \chi(v_{j_2})$, then $j_1 < j_2$.
    \end{inparaenum}
\end{itemize}
We will denote any single $\P_i$ as a \emph{representative function}.
\junk{EQUIVALENT DEFINITION, BUT POSSIBLY LESS INTUITIVE
\begin{itemize}
\item if $u$ is a vertex in $\HQ_i(\DCH)$\footnote{That is, $\chi(u) \notin \{\chi(v_j): 1 \le j < i\}$.}, then $\P_i(u) = u$;
\item if $u$ is not a vertex in $\HQ_i(\DCH)$\footnote{That is, $\chi(u) \in \{\chi(v_j): 1 \le j < i\}$.} and no edge in $\HQ_i(\DCH)$ has $u$ in its label\footnote{That is, $u \notin \lab(\e)$ for every edge $\e$ in $\HQ_i(\DCH)$. Equivalently, it's easy to see that $\chi(u) \in \{\chi(v_j): 1 \le j < i\}$ and $u$ is isolated in $\G[\DCH]$. (In fact, \emph{for every $u \in V$ that is isolated in $\G[\DCH]$, we will have $\P_i(u) = u$ for every $i$}.)}, then $\P_i(u) = u$;
\item otherwise\footnote{That is, $u$ is a not vertex in $\HQ_i(\DCH)$ and there is an edge $\e$ in $\HQ_i(\DCH)$ with $u \in \lab(\e)$.} $\P_i(u) = x$, where
    \begin{inparaenum}[\it (i)]
    \item $x \in \bigcup_{\e: u \in \lab(\e)} \e$ and
    \item for any $x, y \in \bigcup_{\e: u \in \lab(\e)} \e$, if $x \ne y$, $\chi(x) = \chi(v_{j_1})$, and $\chi(y) = \chi(v_{j_2})$, then $j_1 < j_2$.\end{inparaenum}
\end{itemize}
} %END OF JUNK
\end{definition}

The notion of the canonical representative function $\P_1, \P_2, \dots, \P_{\Hsize}: V \rightarrow V$ describes the dependencies between the vertices from $G$ in the construction of the sequence of the hypergraphs $\HQ_1(\DCH), \HQ_2(\DCH), \dots, \HQ_{\Hsize}(\DCH)$. And so, $\P_i(u) = u$ unless vertex $u$ has been contracted during the construction of $\HQ_j(\DCH)$ for $j < i$. If $u$ has been contracted during the construction of $\HQ_j(\DCH)$, then for some colored copy $\h$ of $H$ in $\DCH$ containing $u$, we first added a new hyperedge consisting of vertices in $\N_j^{\h}\langle u \rangle$, and then removed vertex $u$ (with all incident edges from $\HQ_j(\DCH)$). In that case, we will define $\P_j(u) = x$, \footnote{Notice that this notion is well defined only since $u$ is a safe vertex with respect to $\DCH$ and $\HQ_j(\DCH)$, because in that case the neighbors of $u$ in $\HQ_j(\DCH)$ do not depend on the choice of the copy $\h$ of $H$ in $\DCH$ containing $u$ we consider.}where $x$ is the vertex in $\N_j^{\h}\langle u \rangle$ that will be contracted first among all vertices in $\N_j^{\h}\langle u \rangle$ (that is, if $x, y \in \N_j^{\h}\langle u \rangle$ and $\chi(x) = \chi(v_{r_1})$ and $\chi(y) = \chi(v_{r_2})$, then $r_1 \le r_2$). Furthermore, if in some future iteration $s>j$ vertex $x = \P_j(u)$ is contracted, then we will not only set $\P_s(x)$, but we will also update $\P_s(u)$ to be the same as $\P_s(x)$. In fact, we will maintain that for all $k > j$, if $\P_j(u) = x$ then $\P_k(u) = \P_k(x)$.\footnote{Note that function $\P_i$ defines a \emph{forest} on $V$, where in each ``tree'' the root is a vertex $u$ with $\P_i(u) = u$, and the ``leaves'' are formed by vertices $u$ with $\P_i^{(-1)}(u) \ne u$ (that is, for which there is no $v$ with $\P_i(v) = u$).}

\begin{remark}
\label{remark:canonical-representative-function}
Equivalently, one can define $\P_1, \P_2, \dots, \P_{\Hsize}: V \rightarrow V$ recursively as follows:
\begin{itemize}
\item if $u$ is an isolated vertex in $\G[\DCH]$, then $\P_i(u) = u$ for every $i$;
\item otherwise:
    \begin{itemize}[$\diamond$]
    \item $\P_1(u) = u$ for every vertex $u \in V$;
    \item for any $i$, $2 \le i \le \Hsize$, for every $u \in V$:
        \begin{itemize}[$\star$]
        \item if $u$ is a vertex in $\HQ_i(\DCH)$, then $\P_i(u) = u$;
        \item otherwise,
            \begin{itemize}[$\ast$]
            \item if $\P_{i-1}(u)$ has color different than $\chi(v_{i-1})$\footnote{That is, $\P_{i-1}(u)$ is not in $\HQ_{i-1}(\DCH)$.}, then $\P_i(u) = \P_{i-1}(u)$;
            \item else, $\P_i(u)$ is equal to the neighbor of vertex $\P_{i-1}(u)$ in $\HQ_{i-1}(\DCH)$ with the lowest color (that is, $\P_i(u)$ is the neighbor $x$ of $\P_{i-1}(u)$ in $\HQ_{i-1}(\DCH)$ that minimizes $j$ with $\chi(x) = \chi(v_j)$).
            \end{itemize}
        \end{itemize}
    \end{itemize}
\end{itemize}
Let us explain the choice of vertex $x$ in the last case of the definition of $\P_i(u)$. First of all, the choice of $\P_i(u)$ to be a neighbor of vertex $\P_{i-1}(u)$ in $\HQ_{i-1}(\DCH)$ is to ensure that $u$ will belong to the label of the newly created edge incident to that neighbor in $\HQ_{i}(\DCH)$. The choice of the neighbor with the ``lowest color'' is to ensure that that vertex will be the first to be contracted in the later procedure of shrinking $\HQ_j(\DCH)$, and thus, during that construction, the edge containing vertex $u$ will be replaced by another edge. Therefore, our choosing $x$ ensures that if $\chi(u) = \chi(v_r)$, then
\begin{compactitem}
\item for every $i \le r$, $\P_i(u) = u$, and
\item for every $i > r$, $\P_i(u)$ is a vertex in $\HQ_{i}(\DCH)$ and there is a hyperedge $\e$ incident to vertex $\P_i(u)$ such that $u \in \e$.
\end{compactitem}
\end{remark}

%---------------------------------------------------------------------------------------------------------------------------------------------------------

\section{Completing the proof of Lemma \ref{lemma:ExistenceOfH}, and of Theorems \ref{thm:main-H-freeness} and \ref{thm:main-H-freeness-single-call}}
\label{sec:final-proof}

We are now ready to complete the proof of Lemma \ref{lemma:ExistenceOfH}, and with this of Theorems \ref{thm:main-H-freeness} and \ref{thm:main-H-freeness-single-call}.

Let $G = (V,E)$ be a simple planar graph that is $\eps$-far from $H$-free. By our analysis in the previous sections (see Lemma \ref{lemma:defining-Qi+1}), we know that we can order the vertices of $H$ \ $v_1, \dots, v_{\Hsize}$ to define the hypergraphs $\M_1, \dots, \M_{\Hsize}$, so that there are sets $\DDCH{1}, \DDCH{2}, \dots, \DDCH{\Hsize}$ of edge-disjoint colored copies of $H$ in $G$ with $\DDCH{\Hsize} \subseteq \DDCH{\Hsize-1} \subseteq \dots \subseteq \DDCH{1}$ and $|\DDCH{\Hsize}| = \Omega_{\eps,H}(|V|)$, such that for each $i$, $1 \le i \le \Hsize$, the hypergraph $\HQ_i(\DDCH{i})$ is consistent for $\DDCH{i}$.

Let us first apply Lemma \ref{lemma:transformation} to the set $\DDCH{\Hsize}$ of edge-disjoint colored copies of $H$ in $G$ to obtain a subset $\DCH \subseteq \DDCH{\Hsize}$ with $|\DCH| = \Omega_{\eps,H}(|V|)$, such that the graph $\G[\DCH]$ satisfies condition (\ref{part-a-lemma:ExistenceOfH}) of Lemma \ref{lemma:ExistenceOfH}. Therefore, we only have to show that condition (\ref{part-b-lemma:ExistenceOfH}) of Lemma \ref{lemma:ExistenceOfH} holds too, that is, we have to show that if $G = (V,E)$ is a simple planar graph that is $\eps$-far from $H$-free, then

%---------------------------------------------------------------------------------------------------------------------------------------------------------
\begin{algo}\vspace*{-0.25in}
\begin{itemize}[$\otimes$]
\item \RLBD{($\G[\DCH],H,\dg,\ld$)} finds a copy of $H$ in $\G[\DCH]$ with probability $\Omega_{\eps,H}(1)$.
\end{itemize}
\end{algo}
%---------------------------------------------------------------------------------------------------------------------------------------------------------

$\DCH$ is a set of edge-disjoint colored copies of $H$ in $G$ such that $|\DCH| = \Omega_{\eps,H}(|V|)$, and (by Claim \ref{claim:consistent-for-subset}) such that for each $i$, $1 \le i \le \Hsize$, the hypergraph $\HQ_i(\DCH)$ is consistent for $\DCH$. Let us take the canonical representative function $\P_1, \P_2, \dots, \P_{\Hsize}: V \rightarrow V$, cf. Definition \ref{def:canonical-representative-function}.

We will prove $\otimes$ by showing the following two properties (proven below as Claims \ref{claim:final-prop1} and \ref{claim:final-prop2}):

%---------------------------------------------------------------------------------------------------------------------------------------------------------
\begin{walgo}\vspace*{-0.2in}
\begin{enumerate}[1.]
\item %(Claim \ref{claim:final-prop1}:)
    %for some $\dg, \ld = \Theta_{\eps,H}(1)$,
    the probability that \HRLBD\,${(\HQ_{\Hsize}(\DCH), \P_{\Hsize}, \M_{\Hsize}, \Hsize^2, 1)}$ finds a copy of $\M_{\Hsize}$ is $\Omega_{\eps,H}(1)$, and
    \label{final-prop1}
\item %(Claim \ref{claim:final-prop2}:)
    for every $i$, $1 \le i < \Hsize$, %and any $\dg, \ld$,
\begin{itemize}
\item if the probability that
    \label{final-prop2} \HRLBD\,${(\HQ_{i+1}(\DCH),\P_{i+1},\M_{i+1},\dg,\ld)}$ finds a copy of $\M_{i+1}$ is $\Omega_{\eps,H}(1)$,
\item then the probability that \HRLBD\,${(\HQ_i(\DCH),\P_i,\M_i,\Hsize \cdot \dg, 2 \ld)}$ finds a copy of $\M_i$ is $\Omega_{\eps,H}(1)$.
\end{itemize}
\end{enumerate}
\end{walgo}
%---------------------------------------------------------------------------------------------------------------------------------------------------------

Indeed, if Property \ref{final-prop1} holds, then by iterating Property \ref{final-prop2}, we have that for some $\dg^*, \ld^* = \Omega_{\eps,H}(1)$, the probability that \HRLBD\,${(\HQ_1(\DCH),\P_1,\M_1,\dg^*,\ld^*)}$ finds a copy of $\M_1$ is $\Omega_{\eps,H}(1)$. Since $\P_1$ is the identity function $\P_1(u) = u$ for every $u \in V$, and since $\HQ_1(\DCH) \equiv \G[\DCH]$, the behavior of \HRLBFS\,${(\HQ_1(\DCH),\P_1,\dg^*,\ld^*)}$ is identical to the behavior of \RLBFS\,${(\G[\DCH],\dg^*,\ld^*)}$, and further, since $\M_1 \equiv H$, the behavior of \HRLBD\,${(\HQ_1(\DCH),\P_1,\M_1,\dg^*,\ld^*)}$ is identical to the behavior of \RLBD\,${(\G[\DCH],H,\dg^*,\ld^*)}$. Therefore, we obtain that the probability that \RLBD\,${(\G[\DCH],H,\dg^*,\ld^*)}$ finds a copy of $H$ is $\Omega_{\eps,H}(1)$, what yields $\otimes$.

%by Lemma \ref{lemma:transformation} (through Lemma \ref{lemma:ExistenceOfH} and Lemma \ref{lemma:basic}) yields the proof.\Artur{Note: the current structure of the text surrounding Lemma \ref{lemma:basic}, Lemma \ref{lemma:ExistenceOfH}, and Lemma \ref{lemma:transformation} is misleading and complicated, and I'm planning to rework it later.}

What remains is to prove that Properties \ref{final-prop1} and \ref{final-prop2} hold, what we do in the following two central claims, whose proofs are deferred to Section \ref{sec:claim:final-prop1-2} below.

\begin{restatable}{claim}{CFI}
\label{claim:final-prop1}
The probability that \HRLBD\,${(\HQ_{\Hsize}(\DCH), \P_{\Hsize}, \M_{\Hsize}, \Hsize^2, 1)}$ finds a copy of $\M_{\Hsize}$ is $\Omega_{\eps,H}(1)$.
\end{restatable}

\begin{restatable}{claim}{CFII}
\label{claim:final-prop2}
Let $1 \le i < \Hsize$, $\dg = \dg(\eps,H) \ge \Hsize$, $\ld = \ld(\eps,H)$, $\dg^* = \Hsize \cdot \dg$ and $\ld^* = 2 \ld$. If the probability that \HRLBD\,${(\HQ_{i+1}(\DCH),\P_{i+1},\M_{i+1},\dg,\ld)}$ finds a copy of $\M_{i+1}$ is $\Omega_{\eps,H}(1)$, then the probability that \HRLBD\,${(\HQ_i(\DCH),\P_i,\M_i,\dg^*,\ld^*)}$ finds a copy of $\M_i$ is $\Omega_{\eps,H}(1)$.
\end{restatable}

With Claims \ref{claim:final-prop1} and \ref{claim:final-prop2} at hand, we obtain that Properties \ref{final-prop1} and \ref{final-prop2} hold, and therefore we can conclude the proof of the proof of Lemma \ref{lemma:ExistenceOfH}, and with this of Theorems \ref{thm:main-H-freeness} and \ref{thm:main-H-freeness-single-call}.
\qed

%---------------------------------------------------------------------------------------------------------------------------------------------------------

\subsection{Proofs of central Claims \ref{claim:final-prop1} and \ref{claim:final-prop2} --- completing the proof of Lemma \ref{lemma:ExistenceOfH}}
\label{sec:claim:final-prop1-2}

In this section we give proofs of two our central results on which relies our proof of Lemma \ref{lemma:ExistenceOfH} (and with this of Theorems \ref{thm:main-H-freeness} and \ref{thm:main-H-freeness-single-call}): Claims \ref{claim:final-prop1} and \ref{claim:final-prop2}.

\medskip

We begin with the proof of Claim \ref{claim:final-prop1}.

\CFI* % STATEMENT OF Claim \ref{claim:final-prop1}
\junk{
\begin{claim}
\label{claim:final-prop1}
The probability that \HRLBD\,${(\HQ_{\Hsize}(\DCH), \P_{\Hsize}, \M_{\Hsize}, \Hsize^2, 1)}$ finds a copy of $\M_{\Hsize}$ is $\Omega_{\eps,H}(1)$.
\end{claim}
}

\begin{proof}
\newcommand{\ind}{\ensuremath{\mathfrak{ind}}} % USED ONLY IN THIS PROOF
\newcommand{\sis}{\ensuremath{\mathfrak{s}}} % USED ONLY IN THIS PROOF
%We prove that the probability that \HRLBD\,${(\HQ_{\Hsize}(\DCH), \P_{\Hsize}, \M_{\Hsize}, \Hsize^2, 1)}$ finds a copy of $\M_{\Hsize}$ is $\Omega_{\eps,H}(1)$.
%
Our construction (see Section \ref{subsec:shrinking-H}) ensures that $\M_{\Hsize}$ has some number $\sis$ of hyperedges $\e_1, \dots, \e_{\sis}$, each $\e_j$ consisting of a single vertex $v_{\Hsize}$, and with the labels of edges $\e_1, \dots, \e_{\sis}$ defining a partition of $\{v_1, \dots, v_{\Hsize - 1}\}$ (that is, $\bigcup_{j=1}^{\sis} \lab(\e_j) = \{v_1, \dots, v_{\Hsize - 1}\}$ and $\lab(\e_{j_1}) \cap \lab(\e_{j_2}) = \emptyset$ for any $j_1 \ne j_2$).

Similarly, our construction (see Section \ref{subsec:shrinking-copies-of-H}) ensures that $\HQ_{\Hsize}(\DCH)$ contains $\sis \cdot |\DCH|$ hyperedges, each hyperedge $\e$ in $\HQ_{\Hsize}(\DCH)$ consisting of a single vertex of color $\chi(v_{\Hsize})$. Furthermore, each such hyperedge $\e$ corresponds (cf. Definition \ref{def:corresponding-edges}) to a copy of one of the hyperedges $\e_1, \dots, \e_{\sis}$ from $\M_{\Hsize}$; let us denote by $\ind(\e)$ the index of the copy $\e_{\ind(\e)}$ corresponding to $\e$. Notice that $\clab(\e) = \clab(\e_{\ind(\e)})$ and $|\{\e \text{ in } \HQ_{\Hsize}(\DCH): \ind(\e) = j\}| = |\DCH|$ for any $j$, $1 \le j \le \sis$.

Let $\dg^* = \Hsize^2$. In view of the comments and the construction above, by Definition \ref{def:finding-colored-copies-of-Mi}, \HRLBD\,${(\HQ_{\Hsize}(\DCH), \P_{\Hsize}, \M_{\Hsize}, \dg^*, 1)}$ finds a copy of $\M_{\Hsize}$ if,
\begin{enumerate}[(1)]
\item in the call to \HRLBFS\,${(\HQ_{\Hsize}(\DCH), \P_{\Hsize}, \dg^*, 1)}$, it selects the starting vertex $u = \P_{\Hsize}(v)$ to be non-isolated in $\HQ_{\Hsize}(\DCH)$, and
\item vertex $u$ chooses among its $\dg^*$ random incident edges all copies of $\e_1, \dots, \e_{\sis}$.
\end{enumerate}

Our definition of $\P$ ensures that $\P_{\Hsize}(x)$ is a non-isolated vertex in $\HQ_{\Hsize}(\DCH)$ if and only if $x$ is a non-isolated vertex in $\G[\DCH]$. %\Artur{Does it require further clarification?}
Therefore we only have to show that $\G[\DCH]$ has $\Omega_{\eps,H}(|V|)$ non-isolated vertices. Let $G^*[\DCH]$ be the subgraph of $\G[\DCH]$ induced by non-isolated vertices. Since $G^*[\DCH]$ consists of $|\DCH|$ edge-disjoint copies of $H$, $G^*[\DCH]$ has $|\DCH| \cdot |E(H)|$ edges. Since $G^*[\DCH]$ is a subgraph of a simple planar graph, $G^*[\DCH]$ is a simple planar graph too, and thus must have at least $\frac13 |\DCH| \cdot |E(H)|$ vertices (cf. Fact \ref{fact:planar_limited_edges}). Therefore, since $|\DCH| = \Omega_{\eps,H}(|V|)$, we conclude that $G^*[\DCH]$ has $\Omega_{\eps,H}(|V|)$ vertices, or equivalently, that $\G[\DCH]$ has $\Omega_{\eps,H}(|V|)$ non-isolated vertices. Therefore, with probability $\Omega_{\eps,H}(1)$ \HRLBFS\,${(\HQ_{\Hsize}(\DCH), \P_{\Hsize}, \dg^*, 1)}$ selects a non-isolated as the starting vertex.

Next, let us condition on the fact that the starting vertex $u = \P_{\Hsize}(v)$ is non-isolated in $\HQ_{\Hsize}(\DCH)$. Analogously to the classic coupon collector's problem, we can argue that if $u$ selects at least $\sis^2$ (in fact, $\sis \ln (1+\sis)$ would suffice too) times incident edges i.u.r., then with probability $\Omega_{\eps,H}(1)$, the set $\mathcal{E}_{\ell,u}$ will contain $\sis$ hyperedges $\e'_1, \dots, \e'_{\sis}$ with $\ind(\e'_j)=j$ for every $j$, $1 \le j \le \sis$. In this case, the set $\mathcal{E}_{\ell,u}$ will contain a copy of $\M_{\Hsize}$.

By our arguments above, this yields the claim.
\end{proof}

We now move to the proof of Claim \ref{claim:final-prop2}.

\CFII* % STATEMENT OF Claim \ref{claim:final-prop2}
\junk{
\begin{claim}
\label{claim:final-prop2}
Let $1 \le i < \Hsize$, $\dg = \dg(\eps,H) \ge \Hsize$, $\ld = \ld(\eps,H)$, $\dg^* = \Hsize \cdot \dg$ and $\ld^* = 2 \ld$. If the probability that \HRLBD\,${(\HQ_{i+1}(\DCH),\P_{i+1},\M_{i+1},\dg,\ld)}$ finds a copy of $\M_{i+1}$ is $\Omega_{\eps,H}(1)$, then the probability that \HRLBD\,${(\HQ_i(\DCH),\P_i,\M_i,\dg^*,\ld^*)}$ finds a copy of $\M_i$ is $\Omega_{\eps,H}(1)$.
\end{claim}
}

\begin{proof}
\newcommand{\sis}{\ensuremath{\mathfrak{s}}} % USED ONLY IN THIS PROOF
\newcommand{\ris}{\ensuremath{\mathfrak{r}}} % USED ONLY IN THIS PROOF
%Let $1 \le i < \Hsize$, $\dg = \dg(\eps,H) \ge \Hsize$, $\ld = \ld(\eps,H)$, $\dg^* = \Hsize \cdot \dg$ and $\ld^* = 2 \ld$. We will prove that if the probability that \HRLBD\,${(\HQ_{i+1}(\DCH),\P_{i+1},\M_{i+1},\dg,\ld)}$ finds a colored copy of $\M_{i+1}$ is $\Omega_{\eps,H}(1)$, then the probability that \HRLBD\,${(\HQ_i(\DCH),\P_i,\M_i,\dg^*,\ld^*)}$ finds a colored copy of $\M_i$ is $\Omega_{\eps,H}(1)$, where $\dg^* = \Hsize \cdot \dg$ and $\ld^* = 2 \ld$.
%
Let us refer to Definition \ref{def:finding-colored-copies-of-Mi} for the meaning of algorithm \HRLBD\,${(\HQ_{s}(\DCH),\P_{s},\M_{s},\dg',\ld')}$ (and thus also of \HRLBFS\,${(\HQ_{s}(\DCH),\P_{s},\dg',\ld')}$) finding a colored copy of $\M_{r}$.

The proof relies on two basic properties that hold with probability $\Omega_{\eps,H}(1)$:
\begin{itemize}
\item that a single step of \HRLBFS\,${(\HQ_{i+1}(\DCH),\P_{i+1},\dg,\ld)}$ can be simulated by 2 steps of \HRLBFS\,${(\HQ_{i}(\DCH),\P_{i},\dg^*,\ld^*)}$ with $\dg^* = \Hsize \cdot \dg$ and $\ld^* = 2 \ld$, and
\item that if \HRLBFS\,${(\HQ_{i+1}(\DCH),\P_{i+1},\dg,\ld)}$ starts at a vertex $u$, then the same vertex $u$ will be processed by \HRLBFS\,${(\HQ_{i}(\DCH),\P_{i},\dg^*,\ld^*)}$ in $L_0 \cup L_1$ {\small (i.e., in one of the first two rounds)}.
\end{itemize}

Once these two claims hold, the proof of Claim \ref{claim:final-prop2} follows immediately.

We begin with showing that a single step of \HRLBFS\,${(\HQ_{i+1}(\DCH),\P_{i+1},\dg,\ld)}$ can be simulated by 2 steps of \HRLBFS\,${(\HQ_{i}(\DCH),\P_{i},\dg^*,\ld^*)}$.

We begin with two auxiliary definitions.
For any pair of edges $\e$ and $\e'$, we say \emph{$\e$ and $\e'$ are semi-equivalent} if their vertex sets are the same and their colored labels are the same.
Let $\e$ be an edge in $\HQ_{i+1}(\DCH)$ that is modeled by edges $\e_1, \dots, e_{\ris}$ in $\HQ_{i}(\DCH)$ (cf. Definition \ref{def:modeling-hyperedge-in-Hi+1}). Then any $\ris$ edges $\e_1', \dots, e_{\ris}'$ in $\HQ_{i}(\DCH)$ are called \emph{sub-equivalent to $\e$} if for every $1 \le j \le \ris$, edges $\e_j$ and $\e_j'$ are semi-equivalent.

The first definition relates to the scenario when \HRLBD\,${(\HQ,\P,\M_j,\dg,\ld)}$ finds a colored copy of $\M_j$ in $\HQ$ that contains edge $\e$ in $\HQ$. In that case, we claim that the algorithm would have found a copy of $\M_j$ also if instead of using edge $\e$, it used any edge semi-equivalent to $\e$. The second definition is used to describe the scenario when \HRLBD\,${(\HQ_{i+1}(\DCH),\P_{i+1},\M_{i+1},\dg,\ld)}$ finds a colored copy of $\M_{i+1}$ by finding edges $\mathcal{E}$ in $\HQ_{i+1}(\DCH)$ matching $\M_{i+1}$. In that case, to find a colored copy of $\M_{i}$, it is enough that \HRLBD\,${(\HQ_{i}(\DCH),\P_{i},\M_{i},\dg^*,\ld^*)}$ finds only edges $\mathcal{E}'$ such that for every $\e \in \mathcal{E}$, $\mathcal{E}'$ contains edges $\e_1', \dots, e_{\sis}'$ in $\HQ_{i}(\DCH)$ that are sub-equivalent to $\e$.

Let us consider a step of creating set $L_{\ell}$ in \HRLBFS\,${(\HQ_{i+1}(\DCH),\P_{i+1},\dg,\ld)}$, and let $u$ be a vertex in $L_{\ell-1}$ with incident edge $\e$. Let $\e$ belong to a copy $\h_{\e}$ of $\M_{i+1}$ in $\HQ_{i+1}(\DCH)$ and let $\widehat{\e}$ be the corresponding edge in $\M_{i+1}$. By our construction, edge $\e$ was either already present in $\HQ_{i}(\DCH)$, or is a result of a contraction in $\HQ_{i}(\DCH)$ of a vertex $x$ with $\chi(x) = \chi(v_i)$. In the latter case, $\e$ is equal to $\N_i^{\h_{\e}}\langle x \rangle$, the set of neighbors of $x$ in $\h_{\e}$ (in $\HQ_i(\DDCH{i})$) other than $x$.

In \HRLBFS\,${(\HQ_{i+1}(\DCH),\P_{i+1},\dg,\ld)}$, when vertex $u$ selects $\dg$ incident edges i.u.r., the probability that $u$ chooses $\e$ among its $\dg$ incident edges in \HRLBFS\,${(\HQ_{i+1}(\DCH),\P_{i+1},\dg,\ld)}$ is $\mathfrak{p}_{u,\e} = 1 - (1 - 1/\deg_{\HQ_{i+1}(\DCH)}(u))^{\dg}$, where $\deg_{\HQ_{i+1}(\DCH)}(u)$ is the number of edges incident to vertex $u$ in $\HQ_{i+1}(\DCH)$.

If edge $\e$ was already present in $\HQ_{i}(\DCH)$, then the probability that $u$ chooses $\e$ among its $\dg$ incident edges in \HRLBFS\,${(\HQ_{i}(\DCH),\P_{i},\dg^*,\ld^*)}$ is equal to $1 - (1 - 1/\deg_{\HQ_{i}(\DCH)}(u))^{\dg^*}$. Next, we notice that for any vertex $x$ in $\HQ_{i+1}(\DCH)$, $\deg_{\HQ_{i+1}(\DCH)}(u) \le \deg_{\HQ_{i}(\DCH)}(u) \le \Hsize \deg_{\HQ_{i+1}(\DCH)}(u)$. (Indeed, for any colored copy $\h$ of $H$ in $\DCH$ that contains vertex $x$, if we contract in $\h$ a neighbor of $x$ in $\HQ_{i}(\DCH)$, then we remove up to $\Hsize$ edges from $\HQ_{i}(\DCH)$ and add exactly one new edge.) This implies that with our setting $\dg^* = \Hsize \cdot \dg$, we have $1 - (1 - 1/\deg_{\HQ_{i}(\DCH)}(u))^{\dg^*} \ge 1 - (1 - 1/(\Hsize \cdot \deg_{\HQ_{i+1}(\DCH)}(u)))^{\Hsize \cdot \dg} = \Omega_{\eps,H}(1 - (1 - 1/\deg_{\HQ_{i+1}(\DCH)}(u))^{\dg})$. \footnote{To see this, think about the following experiment. Choosing $\e$ in $\HQ_{i+1}(\DCH)$ is like choosing one out of $\deg_{\HQ_{i+1}(\DCH)}(u)$ incident edges, and repeating it $\dg$ times; choosing $\e$ in $\HQ_{i}(\DCH)$ is like choosing one out of up to $\Hsize \cdot \deg_{\HQ_{i+1}(\DCH)}(u)$ incident edges, and repeating it $\dg^*$ times. Now, to choose $\e$ in $\HQ_{i}(\DCH)$ we can also split all edges incident to $u$ in $\HQ_{i}(\DCH)$ into $\deg_{\HQ_{i+1}(\DCH)}(u)$ groups, each group of size approximately $\deg_{\HQ_{i}(\DCH)}(u)/\deg_{\HQ_{i+1}(\DCH)}(u)$. Then, the probability that we will choose an edge from the same group as $\e$ is $\mathfrak{p}_{u,\e}$ (approximately, because of rounding) the same as the probability that we will choose edge $\e$ in $\HQ_{i+1}(\DCH)$. Therefore, with probability at most $1/\Hsize$, we would then choose edge $\e$ in $\HQ_{i}(\DCH)$. If we repeat this $\Hsize$ time, we will get probability $\Omega_{\eps,H}(\mathfrak{p}_{u,\e})$. (Notice that we could also be happy with the probability $\mathfrak{p}_{u,\e}/\Hsize$, since this is $\Omega_{\eps,H}(\mathfrak{p}_{u,\e})$.)}
%
%\Artur{I put the previous footnote since this fact is less trivial than I first thought.}
%
Therefore, we can conclude that:
\begin{itemize}[\textbf{Case~1:}]
\item if edge $\e$ is present in $\HQ_{i}(\DCH)$ and in \HRLBFS\,${(\HQ_{i+1}(\DCH),\P_{i+1},\dg,\ld)}$, vertex $u$ selects $\e$ among its $\dg$ incident edges with probability $\mathfrak{p}_{u,\e}$, then in \HRLBFS\,${(\HQ_{i}(\DCH),\P_{i},\dg^*,\ld^*)}$, vertex $u$ selects $\e$ among its $\dg^*$ incident edges with probability $\Omega_{\eps,H}(\mathfrak{p}_{u,\e})$.
\end{itemize}

The case when edge $\e$ is not present in $\HQ_{i}(\DCH)$ and has been obtained as a contraction of vertex $x$ with $\chi(x) = \chi(v_i)$, with $\e = \N_i^{\h_{\e}}\langle x \rangle$, is more complicated.

Since $\HQ_{i+1}(\DCH)$ is consistent for $\DCH$, vertex $x$ is safe with respect to $\DCH$ and $\HQ_{i}(\DCH)$.
Let $x$ be incident to $\deg_{\HQ_{i}(\DCH)}(x)$ edges in $\HQ_i(\DDCH{i})$ and note that $\chi(x) = \chi(v_i)$.
By Remark \ref{remark-properties-of-safe-vertices}, we can group edges incident to $x$ in $\HQ_i(\DDCH{i})$ into $\ris$ groups of \emph{the same size} each (equal to $\deg_{\HQ_{i}(\DCH)}(x)/\ris$), each group corresponding to a copy of one of the $\ris$ edges incident to $v_i$ in $\M_i$, any two edges from the same group being semi-equivalent.
%Further, the edges from the same group are defined by the same vertices %(that is, for any two edges $\e', \e''$ from the same group, for any vertex $y$, $y \in \e'$ iff $y \in \e''$)
%and have the same colored label. %(that is, $\clab(\e') = \clab(\e'')$ for any two edges $\e', \e''$ from the same group).
%In view of that, if we want to find a copy of $H$, then any pair of edges from the same group are replaceable by each other.

After contracting vertex $x$, we will create $\sis = \deg_{\HQ_{i}(\DCH)}(x)/\ris$ new edges $\e_1, \dots, \e_{\sis}$ in $\HQ_{i+1}(\DCH)$, each new edge with the same vertex set $\N_i\langle x\rangle$ that correspond to the set of neighbors of $x$ in $\HQ_i(\DDCH{i})$, and having the same colored label. Thus all new edges $\e_1, \dots, \e_{\sis}$ are semi-equivalent. Furthermore, any $\e_1', \dots, \e_{\ris}'$ incident to $x$ in $\HQ_i(\DDCH{i})$ that are from $\ris$ different groups are sub-equivalent to every edge in $\e_1, \dots, \e_{\sis}$.

We will compare the probability that after arriving at vertex $u$, \HRLBFS\,${(\HQ_{i+1}(\DCH),\P_{i+1},\dg,\ld)}$ visits any of the edges $\e_1, \e_2, \dots, \e_{\sis}$, with the probability that after arriving at $u$, algorithm \HRLBFS\,${(\HQ_{i}(\DCH),\P_{i},\dg^*,\ld^*)}$ visits in $\HQ_{i}(\DCH)$ \ $\ris$ edges that are incident to $x$ in $\HQ_i(\DDCH{i})$ and that are from $\ris$ different groups (and hence are sub-equivalent to every edge in $\e_1, \dots, \e_{\sis}$).

In \HRLBFS, when vertex $u$ selects $\dg$ incident edges i.u.r., the probability that it chooses at least one of the edges $\e_1, \dots, \e_{\sis}$ among its $\dg$ incident edges in \HRLBFS\,${(\HQ_{i+1}(\DCH),\P_{i+1},\dg,\ld)}$ is equal to $\mathfrak{p}_{i+1} = 1 - (1-\sis/\deg_{\HQ_{i+1}(\DCH)}(u))^{\dg}$.% = \Theta_{\eps,H}(1-e^{\dg \sis/\deg_{\HQ_{i+1}(\DCH)}(u)})$.

Let us compare it to the probability that in \HRLBFS\,${(\HQ_{i}(\DCH),\P_{i},\dg^*,\ld^*)}$, when vertex $u$ selects $\dg$ incident edges i.u.r. then one of these edges is incident to vertex $x$, \emph{and} when in \HRLBFS\,${(\HQ_{i}(\DCH),\P_{i},\dg^*,\ld^*)}$ vertex $x$ selects $\dg$ incident edges i.u.r. then at least one edge from each of the $\ris$ groups of edges incident to $x$ in $\HQ_{i}(\DCH)$ is chosen\footnote{Let us notice that we do not assume that $x$ will be processed in the \emph{next round} in \HRLBFS\,${(\HQ_{i}(\DCH),\P_{i},\dg^*,\ld^*)}$, after vertex $u$ is processed. This is because it is possible that vertex $x$ has been processed before vertex $u$, for example, as the very first vertex in the call to \HRLBFS\,${(\HQ_{i}(\DCH),\P_{i},\dg^*,\ld^*)}$. Our arguments imply that both $u$ and $x$ will be processed (in the way we want them to be processed) \emph{not later} than in the next round.}. The first probability, that one of the incident edges selected by $u$ is incident to $x$, is equal to $\mathfrak{p}_i \ge %1 - (1-\sis\cdot\ris/\deg_{\HQ_{i}(\DCH)}(u))^{\dg^*} \ge
1 - (1-\sis/\deg_{\HQ_{i}(\DCH)}(u))^{\dg^*}$, since the number of edges containing both $u$ and $x$ in $\HQ_{i}(\DCH)$ is at least $\sis$. % and at most $\sis \cdot \ris$.
To estimate the second probability, similarly as we were already arguing in the proof of Claim \ref{claim:final-prop1} and analogously to the classic coupon collector's problem, if $x$ selects at least $\ris^2$ (in fact, $\ris \ln (1+\ris)$ would suffice too) times incident edges i.u.r. (and we have $\dg^* \ge \Hsize^2$), then with probability $\Omega_{\eps,H}(1)$, the corresponding set $\mathcal{E}_{\cdot,x}$ will contain at least one edge from each of the $\ris$ groups of edges incident to $x$ in $\HQ_{i}(\DCH)$. Therefore, in summary, with probability $\Omega_{\eps,H}(\mathfrak{p}_i)$, if \HRLBFS\,${(\HQ_{i}(\DCH),\P_{i},\dg^*,\ld^*)}$ visits vertex $u$, then the algorithm will visit (until at most two rounds later) edges $\e_1', \dots, \e_{\ris}'$ that are sub-equivalent to edges $\e_1, \dots, \e_{\sis}$.

Now we only have to match the probabilities of these events in \HRLBFS\,${(\HQ_{i}(\DCH),\P_{i},\dg^*,\ld^*)}$ and in \HRLBFS\,${(\HQ_{i+1}(\DCH),\P_{i+1},\dg,\ld)}$. Since, as we were arguing above, $\deg_{\HQ_{i+1}(\DCH)}(u) \le \deg_{\HQ_{i}(\DCH)}(u) \le \Hsize \deg_{\HQ_{i+1}(\DCH)}(u)$, we note that with our setting $\dg^* = \Hsize \cdot \dg$, we have $\mathfrak{p}_{i+1} = 1 - (1 - \sis/\deg_{\HQ_{i+1}(\DCH)}(u))^{\dg} =
%\Theta_{\eps,H}(1 - (1 - \sis\cdot\ris/\deg_{\HQ_{i+1}(\DCH)}(u))^{\dg^*}) =
\Omega_{\eps,H}(\mathfrak{p}_i)$, using the same arguments as before. This gives the following:
\begin{itemize}[\textbf{Case~2:}]
\item if edge $\e$ is not in $\HQ_{i}(\DCH)$, when \HRLBFS\,${(\HQ_{i+1}(\DCH),\P_{i+1},\dg,\ld)}$ arrives at vertex $u$, if $\mathfrak{p}_{i+1}$ is the probability that $u$ selects an edge semi-equivalent to $\e$ among its $\dg$ incident edges, then when \HRLBFS\,${(\HQ_{i}(\DCH),\P_{i},\dg^*,\ld^*)}$ arrives at $u$ (with $u \in L_{\ell}$), then with probability $\Omega_{\eps,H}(\mathfrak{p}_{i+1})$ the set $\bigcup_{j=1}^{\ell+2}\mathcal{E}_j$ of selected edges until at most two rounds later contains edges $\e_1', \dots, \e_{\ris}'$ that are sub-equivalent to~$\e$.
\end{itemize}

Therefore, in summary, our analysis of Case 1 and Case 2 above implies our claim that a single step of algorithm \HRLBFS\,${(\HQ_{i+1}(\DCH),\P_{i+1},\dg,\ld)}$ can be simulated by 2 steps of algorithm \HRLBFS\,${(\HQ_{i}(\DCH),\P_{i},\dg^*,\ld^*)}$, with the success probability loss of $O_{\eps,H}(1)$. That is, if one arrives at vertex $u$ in step $k$ of \HRLBFS\,${(\HQ_{i+1}(\DCH),\P_{i+1},\dg,\ld)}$ and the probability that one selects an edge semi-equivalent to $\e$ is $\mathfrak{p}_{i+1}$, then if one arrives at vertex $u$ in step $\ell$ of \HRLBFS\,${(\HQ_{i}(\DCH),\P_{i},\dg^*,\ld^*)}$, then with probability $\Omega_{\eps,H}(\mathfrak{p}_{i+1})$, either $\mathcal{E}_{\ell+1}$ contains an edge semi-equivalent to $\e$, or $\bigcup_{j=1}^{\ell+2}\mathcal{E}_j$ contains edges $\e_1', \dots, \e_{\ris}'$ that are sub-equivalent to~$\e$.

\paragraph{Choosing starting vertex.}
%
%It remains to address the probability of choosing $u = \P_{i+1}(v)$ as the starting vertex in \HRLBFS\,${(\HQ_{i+1}(\DCH),\P_{i+1},\dg,\ld)}$.
Let us recall that the probability to choose $u \in V(\HQ_{i+1}(\DCH))$ as a starting vertex $\P_{i+1}(v)$ in \HRLBFS\,${(\HQ_{i+1}(\DCH), \P_{i+1}, \dg, \ld)}$ is $\mathfrak{p}_{i+1} = \frac{|\P_{i+1}^{(-1)}(u)|}{|V|}$. Since we may contract many vertices into $u$ during our construction, the probability of choosing $u$ as a starting vertex in \HRLBFS\,${(\HQ_{i+1}(\DCH),\P_{i+1},\dg,\ld)}$ can be significantly larger than the probability of choosing $u$ in \HRLBFS\,${(\HQ_{i}(\DCH), \P_{i}, \dg^*, \ld^*)}$, which is $\mathfrak{p}_i = \frac{|\P_{i}^{(-1)}(u)|}{|V|}$. However, our definition of $\P_{i+1}$ ensures that
\begin{displaymath}
    |\P_{i+1}^{(-1)}(u)|
        =
    |\P_{i}^{(-1)}(u)| + \sum_{\text{$x$ adjacent to $u$ in $\HQ_{i}(\DCH)$}: \chi(x) = \chi(v_i)} |\P_{i}^{(-1)}(x)|
        \enspace.
\end{displaymath}

Let us notice that if a vertex $x$ of color $\chi(v_i)$ that is adjacent to $u$ in $\HQ_{i}(\DCH)$ is selected as the starting vertex in \HRLBFS\,${(\HQ_i(\DCH), \P_i, \dg^*, \ld^*)}$, which happens with probability $\frac{|\P_{i}^{(-1)}(x)|}{|V|}$, then since (cf. Lemma \ref{lemma:defining-Qi+1}) $x$ is a safe vertex with respect to $\DCH$ and $\HQ_i(\DCH)$, each copy of $\M_i$ in $\HQ_i(\DCH)$ containing vertex $x$ has at least one edge containing also vertex $u$. Therefore, in \HRLBFS\,${(\HQ_i(\DCH), \P_i, \dg^*, \ld^*)}$, we will not only have $x \in L_0$, but also if $\dg = \Omega_{\eps,H}(1)$ is sufficiently large ($\dg > \Hsize$ will suffice), then with probability at least $\frac12$ we will have $u \in L_1$. Summing up over all starting vertices (including $u$), we obtain that $u$ is in $L_0 \cup L_1$ with probability at least $\frac12 \mathfrak{p}_{i+1}$.

Now we are ready to complete the analysis and prove Claim \ref{claim:final-prop2}. Let us consider the random process \HRLBFS\,${(\HQ_{i+1}(\DCH), \P_{i+1}, \dg, \ld)}$ selecting vertices and edges to define $L_j$ and $\mathcal{E}_{j+1}$ for $0 \le j \le \ld$. Similarly, let us consider the random process of \HRLBFS\,${(\HQ_i(\DCH), \P_i, \dg^*, 2\ld)}$ selecting vertices and edges to define $L'_j$ and $\mathcal{E}'_{j+1}$ for $0 \le j \le 2\ld$. Notice that $|\bigcup_{j=0}^{\ld} L_j| = \Omega_{\eps,H}(1)$, $|\bigcup_{j=1}^{\ld} \mathcal{E}_j| = \Omega_{\eps,H}(1)$, $|\bigcup_{j=0}^{2\ld} L'_j| = \Omega_{\eps,H}(1)$, $|\bigcup_{j=1}^{2\ld} \mathcal{E}'_j| = \Omega_{\eps,H}(1)$. Suppose that \HRLBD\,${(\HQ_{i+1}(\DCH), \P_{i+1}, \M_{i+1}, \dg, \ld)}$ starts at a vertex $u = \P_{i+1}(v)$ and finds a copy of $\M_{i+1}$ consisting of edges $\e_1, \dots, \e_k$ in $\HQ_{i+1}(\DCH)$, where $k = \Hsize - i$. Then, our analysis above gives that with at most a constant-factor probability loss, \HRLBD\,${(\HQ_{i}(\DCH), \P_{i}, \M_{i}, \dg^*, 2\ld)}$ will have $u$ in $L'_0 \cup L'_1$, and then, for every edge $\e_j$, $1 \le j \le k$, will either have $\e_j \in \bigcup_{t=1}^{2\ld} \mathcal{E}'_t$ or $\e'_{j_1}, \dots, \e'_{j_\ris} \in \bigcup_{t=1}^{2\ld} \mathcal{E}'_t$, where $\e'_{j_1}, \dots, \e'_{j_\ris}$ are sub-equivalent to edges $\e_j$ (this defines a proper coupling, properly taking care of multiple edges equivalent to $\e_j$). Now, since every edge $\e_j$
\begin{compactenum}[\sl (i)]
%\begin{enumerate}[\sl (i)]
%\begin{inparaenum}[\sl (i)]
\item either corresponds to an edge in both $\M_{i+1}$ and $\M_i$, or
\item corresponds to an edge $\widehat{\e}$ in $\M_{i+1}$ that is modeled by $\widehat{\e'_{j_1}}, \dots, \widehat{\e'_{j_\ris}}$ in $\M_i$, and edges $\e'_{j_1}, \dots, \e'_{j_\ris}$ correspond to the edges $\widehat{\e'_{j_1}}, \dots, \widehat{\e'_{j_\ris}}$,
%\end{inparaenum}
%\end{enumerate}
\end{compactenum}
we can argue that in that case, \HRLBD\,${(\HQ_{i}(\DCH), \P_{i}, \M_{i}, \dg^*, 2\ld)}$ will find a copy of $\M_i$ (cf. Definition~\ref{def:finding-colored-copies-of-Mi}).

Therefore, with only a constant-factor probability loss, if \HRLBD\,${(\HQ_{i+1}(\DCH), \P_{i+1}, \M_{i+1}, \dg, \ld)}$ finds a copy of $\M_{i+1}$ then \HRLBD\,${(\HQ_{i}(\DCH), \P_{i}, \M_{i}, \Hsize \dg, 2\ld)}$ finds a copy of $\M_i$.
\end{proof}

%---------------------------------------------------------------------------------------------------------------------------------------------------------

\section{Extension to families of arbitrary (not necessarily connected) finite graphs}
\label{sec:dics-many}

Our result in Theorem \ref{thm:main-H-freeness} can be easily extended to allow the \emph{forbidden finite graphs $H$ to be arbitrary}, that is, not necessarily connected. Furthermore, the analysis extends in a straightforward way to the case when one wants to \emph{test if for a given arbitrary finite family $\mathcal{H}$ of finite graphs, the input planar graph $G$ is $\mathcal{H}$-free}, that is, contain no copy of any graph from $\mathcal{H}$.
\paragraph{Disconnected $H$.}
Notice that when $H$ is not connected, \RLBD\,${(G,H,\dg,\ld)}$ may not be able to find a copy of $H$ in $G$ since it explores only a small connected neighborhood of the randomly sampled starting vertex $v$. However, one can easily extend the tester to be run separately on each connected component of $H$ to do the job.

Let us assume that $H$ consists of connected components $\mathfrak{h}_1, \mathfrak{h}_2, \dots, \mathfrak{h}_r$. As in Section \ref{subsubsec:coloring-H}, we color the vertices of $H$ arbitrarily, using $\Hsize$ distinct colors $\{1, 2, \dots, \Hsize\}$, one color for each vertex. Our analysis in Section \ref{sec:exploration-testing-H-free} starts with (an existential) Lemma \ref{lemma:edge-disjoint-copies-H-free} that if $G$ is $\eps$-far from $H$-free, then one can color vertices of $G$ with $\Hsize$ colors $\chi$ such that $G$ has a set of $\Omega_{\eps,H}(|V|)$ edge-disjoint colored copies of $H$. It is easy to see that Lemma \ref{lemma:edge-disjoint-copies-H-free} holds also for disconnected $H$. And so, in particular, for every connected component $\mathfrak{h}_i$ of $H$, there are $\Omega_{\eps,H}(|V|)$ edge-disjoint colored copies of $\mathfrak{h}_i$ with colors of the vertices consistent with the coloring $\chi$ of $G$. Furthermore, since all connected components $\mathfrak{h}_1, \mathfrak{h}_2, \dots, \mathfrak{h}_r$ use distinct colors in $H$, these copies will be edge-disjoint between the copies of $\mathfrak{h}_1, \mathfrak{h}_2, \dots, \mathfrak{h}_r$. Then, for every connected component $\mathfrak{h}_i$ of $H$, we run \RLBD\,${(G,\mathfrak{h}_i,\dg_i,\ld_i)}$, and the identical analysis as in Sections \ref{sec:exploration-testing-H-free} -- \ref{sec:final-proof} concludes that Theorem \ref{thm:main-H-freeness-single-call} holds in the following way: there are positive functions $\dg_i = \dg_i(\eps,\mathfrak{h}_i) = O_{\eps,H}(1)$ and $\ld_i = \ld(\eps,\mathfrak{h}_i) = O_{\eps,H}(1)$, such that for any planar graph $G$ that is $\eps$-far from $H$-free, \RLBD$(G,\mathfrak{h}_i,\dg_i,\ld_i)$ finds a colored copy of $\mathfrak{h}_i$ with probability $\Omega_{\eps,H}(1)$. Since the colored copies of connected components $\mathfrak{h}_1, \mathfrak{h}_2, \dots, \mathfrak{h}_r$ are pairwise disjoint in $G$, this implies that if we run \RLBD\,${(G,\mathfrak{h}_i,\dg_i,\ld_i)}$ for $1 \le i \le r$, with appropriate $\dg_i = O_{\eps,H}(1)$ and $\ld_i = O_{\eps,H}(1)$, then for any planar graph $G$ that is $\eps$-far from $H$-free, we find a colored copy of $H$ with probability $\Omega_{\eps,H}(1)$. Therefore, if we repeat this process $O_{\eps,H}(1)$ many times, we can amplify the error probability and obtain that for any planar graph $G$ that is $\eps$-far from $H$-free, we find a colored copy of $H$ with probability at least $\frac23$.
\paragraph{Forbidden family.}
Next, we extend our study to test if a given planar graph contains no copy of any forbidden graph from a given finite family of finite graphs. Let $\mathcal{H}$ be an arbitrary finite family of finite graphs (for a given $\eps>0$, we allow the size to be $O_{\eps}(1)$). We say a simple graph $G$ is \emph{$\mathcal{H}$-free} if it is $H$-free for every $H \in \mathcal{H}$; $G$ is \emph{$\eps$-far from $\mathcal{H}$-free} if one has to delete more than $\eps |V|$ edges from $G$ to obtain an $\mathcal{H}$-free graph. This definition implies that since $\mathcal{H}$ is finite, if $G$ is $\eps$-far from $\mathcal{H}$-free, then there is $H \in \mathcal{H}$ such that $G$ is $\eps/|\mathcal{H}|$-far from $H$-free.

Let us suppose that $\mathcal{H}$ is an arbitrary finite family of finite graphs. (Note that since $\mathcal{H}$ is a finite family of finite graphs, $|\mathcal{H}| = O_{\eps}(1)$.) Then our analysis above can be easily extended to test with a constant number of queries if a planar graph is $\mathcal{H}$-free. Indeed, let us run a constant query-time $\eps/|\mathcal{H}|$-tester for every $H \in \mathcal{H}$, and reject if any of the tests rejects. Notice that if $G$ is $\mathcal{H}$-free then this tester will accept, and if $G$ is $\eps$-far from $\mathcal{H}$-free then since there is $H \in \mathcal{H}$ such that $G$ is $\eps/|\mathcal{H}|$-far from $H$, the tester will reject $G$ with probability at least $\frac23$. %(The query complexity of this tester is $\sum_{H \in \mathcal{H}} \mathfrak{t}(\eps/|\mathcal{H}|,H) = O_{\eps,\mathcal{H}}(1)$, where $\mathfrak{t}(\eps,H)$ is the query complexity of our tester for $H$-freeness.)

The discussion above can be summarized in the following theorem.

\begin{theorem}
\label{thm:main-extension}
Let $\mathcal{H}$ be an arbitrary collection of (not necessarily connected) finite graphs. Then there is a one-sided error property tester that for any simple planar graph $G$ performs a constant number of queries to the random neighbor oracle and accepts if $G$ is $\mathcal{H}$-free, and with probability at least $\frac23$ rejects if $G$ is $\eps$-far from $\mathcal{H}$-free.
\end{theorem}

Theorem \ref{thm:main-extension} holds also if $\mathcal{H}$ varies with different $\eps$. That is, if for a given $\eps > 0$, the goal is to test if $G$ is $\mathcal{H}$-free or is $\eps$-far from $\mathcal{H}$-free, for a finite family of graphs $\mathcal{H}$ that may depend on $\eps$.

%---------------------------------------------------------------------------------------------------------------------------------------------------------

\section{Extending the analysis to minor-free graphs}
\label{sec:minor-free}

While throughout the paper we focused on testing $H$-freeness of \emph{planar graphs}, our techniques can easily be extended to any \emph{class of minor-free graphs}. Recall that a graph $L$ is called a \emph{minor} of a graph $G$ if $L$ can be obtained from $G$ via a sequence of vertex and edge deletions, and edge contractions. For any graph $L$, a graph $G$ is called \emph{$L$-minor-free} if $L$ is not a minor of $G$. (For example, by Kuratowski's Theorem, a graph is planar if and only if it is $K_{3,3}$-minor-free and $K_5$-minor-free.)

Let us fix a graph $L$ and consider the input graph $G$ to be an $L$-minor-free graph. We now argue now that entire analysis presented in the previous sections easily extends to testing $H$-freeness of $G$. The key observation is that our analysis in Sections \ref{sec:outline-proof}--\ref{sec:final-proof} relies only on the following two properties of planar graphs:
\begin{enumerate}[\it (i)]
\item every minor of a planar graph is planar (cf. Fact \ref{fact:planar_minor_planar}),
\item the number of edges in a planar graph is $O(n)$, where $n$ is the number of vertices (cf. Fact \ref{fact:planar_limited_edges}).
\end{enumerate}
It is known that these two properties hold for any class of $L$-minor-free graphs (that is, the first property would be that every minor of an $L$-minor-free graph is $L$-minor-free). Therefore, we can proceed with nearly identical analysis for $L$-minor-free graphs and arrive at the following version of Theorem \ref{thm:main-H-freeness}.

\begin{theorem}
\label{thm:main-H-minor-free}
Let $L$ be a fixed graph. There are positive functions $f$, $g$, and $h$ such that for any $L$-minor-free-graph $G$:
\begin{itemize}
\item if $G$ is $H$-free, then \RBE\,$(G,H,\eps)$ accepts $G$, and
\item if $G$ is $\eps$-far from $H$-free, then \RBE$(G,H,\eps)$ rejects $G$ with probability at least $0.99$.
\end{itemize}
\end{theorem}

Furthermore, in the same way as in Section \ref{sec:dics-many}, we can extend Theorem \ref{thm:main-extension} to obtain the following.

\begin{theorem}
\label{thm:main-H-minor-free-extension}
Let $L$ be a fixed graph. Let $\mathcal{H}$ be an arbitrary collection of (not necessarily connected) finite graphs. Then there is a one-sided error property tester that for any $L$-minor-free-graph $G$ performs a constant number of queries to the random neighbor oracle and accepts if $G$ is $\mathcal{H}$-free, and with probability at least $\frac23$ rejects if $G$ is $\eps$-far from $\mathcal{H}$-free.
\end{theorem}

\begin{remark}
It should be noted that while our main focus is on the random neighbor oracle model, it is straightforward to extend our testers (and their analysis) for $\mathcal{H}$-freeness to the other three oracle access model presented in Section \ref{subsubsec:oracle}. Indeed, since each of these models can trivially simulate the random neighbor oracle model without any loss in the query complexity, Theorem \ref{thm:main-H-minor-free-extension} (and also Theorems \ref{thm:main-H-freeness} and \ref{thm:main-extension}) holds also for all these oracle access models.

However, our main result, the characterization of testable properties in planar graphs, as well as our reduction in Theorem \ref{thm:main-1-sided-reduces-to-H-freeness}, cannot be extended to the other models (see Section \ref{subsubsec:sensitivity-of-the-models}).
\end{remark}

%---------------------------------------------------------------------------------------------------------------------------------------------------------

\section{Conclusions}
\label{sec:conclusions}

The fundamental problem in the area of property testing is to understand the complexity of testing graph properties in all natural models. One of the central questions here is to provide characterizations of testable graph properties in these models, that is, to determine which graph properties can be tested with constant query complexity. While we have characterizations of graph properties testable in the dense graph model, and some understanding of testable graph properties in the bounded-degree graph model, finding such a characterization in a very natural case of general graphs, without any bounds for their maximum degrees, remains a challenging and elusive open problem. The main result of this paper, Theorem \ref{thm:characterization}, resolves an important natural special case of this open problem, which concerns property testers for planar graphs and for minor-closed graphs with one-sided error in the random neighbor oracle model.

\medskip

Our main technical, algorithmic contribution significantly extend the approach from \cite{CMOS11} to prove that $H$-freeness is testable with a constant number of queries for general planar graphs. Our result was proven via a new type of analysis of random exploration of planar graphs and their combination of the study of hypergraph representations of contractions in planar graphs. Our analysis easily carries over to classes of graphs defined by general fixed forbidden minors.

\medskip

Our work is a continuation of our efforts to understand the complexity of testing basic graph properties in graphs with no bounds for the degrees. Indeed, while major efforts in the property testing community have been put to study dense graphs and bounded degree graphs (cf. \cite[Chapter~8-9]{Gol17}), we have seen only limited advances in the study of general graphs, in particular, sparse graphs but without any bounds for the maximum degrees. We believe that this model is one of the most natural models, and it is also most relevant to computer science applications. Similarly as it has been done in \cite[Chapter~10.5.3]{Gol17}, we would advocate further study of this model because of its importance, its applications, and the variety (and beauty) of techniques used to advance this topic.

%\medskip

%---------------------------------------------------------------------------------------------------------------------------------------------------------

%===========================================================================
%= REFERENCES
%===========================================================================

\newcommand{\Proc}{Proceedings of the~}
\newcommand{\ALENEX}{Workshop on Algorithm Engineering and Experiments (ALENEX)}
\newcommand{\BEATCS}{Bulletin of the European Association for Theoretical Computer Science (BEATCS)}
\newcommand{\CCCG}{Canadian Conference on Computational Geometry (CCCG)}
\newcommand{\CIAC}{Italian Conference on Algorithms and Complexity (CIAC)}
\newcommand{\COCOON}{Annual International Computing Combinatorics Conference (COCOON)}
\newcommand{\COLT}{Annual Conference on Learning Theory (COLT)}
\newcommand{\COMPGEOM}{Annual ACM Symposium on Computational Geometry}
\newcommand{\DCGEOM}{Discrete \& Computational Geometry}
\newcommand{\DISC}{International Symposium on Distributed Computing (DISC)}
\newcommand{\ECCC}{Electronic Colloquium on Computational Complexity (ECCC)}
\newcommand{\ESA}{Annual European Symposium on Algorithms (ESA)}
\newcommand{\FOCS}{IEEE Symposium on Foundations of Computer Science (FOCS)}
\newcommand{\FSTTCS}{Foundations of Software Technology and Theoretical Computer Science (FSTTCS)}
\newcommand{\ICALP}{Annual International Colloquium on Automata, Languages and Programming (ICALP)}
\newcommand{\ICCCN}{IEEE International Conference on Computer Communications and Networks (ICCCN)}
\newcommand{\ICDCS}{International Conference on Distributed Computing Systems (ICDCS)}
\newcommand{\IJCGA}{International Journal of Computational Geometry and Applications}
\newcommand{\INFOCOM}{IEEE INFOCOM}
%{IEEE INFOCOM, Annual Joint Conference of the IEEE Computer and Communications Societies}
\newcommand{\IPCO}{International Integer Programming and Combinatorial Optimization Conference (IPCO)}
\newcommand{\ISAAC}{International Symposium on Algorithms and Computation (ISAAC)}
\newcommand{\ISTCS}{Israel Symposium on Theory of Computing and Systems (ISTCS)}
\newcommand{\JACM}{Journal of the ACM}
\newcommand{\LNCS}{Lecture Notes in Computer Science}
\newcommand{\PODS}{ACM SIGMOD Symposium on Principles of Database Systems (PODS)}
\newcommand{\RANDOM}{International Workshop on Randomization and Approximation Techniques in Computer Science (RANDOM)}
\newcommand{\RSA}{Random Structures and Algorithms}
\newcommand{\SICOMP}{SIAM Journal on Computing}
\newcommand{\SODA}{Annual ACM-SIAM Symposium on Discrete Algorithms (SODA)}
\newcommand{\SPAA}{Annual ACM Symposium on Parallel Algorithms and Architectures (SPAA)}
\newcommand{\STACS}{Annual Symposium on Theoretical Aspects of Computer Science (STACS)}
\newcommand{\STOC}{Annual ACM Symposium on Theory of Computing (STOC)}
\newcommand{\SWAT}{Scandinavian Workshop on Algorithm Theory (SWAT)}
\newcommand{\TALG}{ACM Transactions on Algorithms}
\newcommand{\UAI}{Conference on Uncertainty in Artificial Intelligence (UAI)}
\newcommand{\WADS}{Workshop on Algorithms and Data Structures (WADS)}
\newcommand{\TCS}{Theory of Computing Systems}

\renewcommand{\Proc}{{\rm In} Proceedings of the~}

%%%%%%%%%%%%%%%%%%%%%%%%%%%%%%%%%%%%%%%%%%%%%%%%%%

%\bigskip\bigskip%\newpage
%\bibliographystyle{alpha}
\bibliographystyle{IEEEtranS}

%---------------------------------------------------------------------------------------------------------------------------------------------------------

\bigskip
\newpage

\appendix
\begin{center}\huge\bf Appendix\end{center}

%---------------------------------------------------------------------------------------------------------------------------------------------------------

\section{Basic properties of planar graphs}
\label{subsec:basic-planar}

For the sake of completeness, we discuss here some basic (and well known) properties of planar graphs, as frequently used in our paper.

The graph $G' = (V',E')$ obtained by the \emph{contraction of an edge} $(u,v) \in E$ into vertex $u$ is defined as follows: $V' = V \setminus \{v\}$ and $E' = \{(x,y)\in E: x \ne v \land y \ne v\} \cup \{(x,u) : (x,v) \in E \land x \ne u\}$. A graph $G'$ that can be obtained from a graph $G$ via a sequence of edge removals, vertex removals, and edge contractions is called a \emph{minor} of $G$.
Equivalently, a graph $G$ contains an $h$-vertex graph $G'$ as a minor if $G$ contains $\ell$ pairwise disjoint vertex sets $V_1, \dots, V_{\ell}$ such that the graph induced by $G$ on each of these sets is connected, and if $(i,j) \in E(G')$ then $G$ contains at least one edge connecting a vertex of $V_i$ to a vertex of $V_j$. If $G'$ is not a minor of $G$, then $G$ is said to be $G'$-minor free. A \emph{graph property $P$ is minor-closed} if every minor of a graph in $P$ is also in $P$, or equivalently if $P$ is closed under removal of edges, removal of vertices and contraction of edges.

We use the following well-known property of planar graphs.

\begin{fact}
\label{fact:planar_minor_planar}
Any minor of a planar graph is planar.
\end{fact}

Furthermore, we use the following upper bound on the number of edges in a simple planar graph, which follows immediately from Euler's formula.

\begin{fact}
\label{fact:planar_limited_edges}
For any simple planar graph $G = (V,E)$ (with no self-loops or parallel edges), $|E| \le 3|V|-6$.
\end{fact}

We remark that for any class of graphs $\mathcal{G}$ that is defined by a finite collection of forbidden minors similar statements are true, i.e., if $G \in \mathcal{G}$, then any minor of $G$ also belongs to $\mathcal{G}$ and if $G = (V,E) \in \mathcal{G}$, then $G$ has $O(|V|)$ edges (where the constant in the Big-Oh notation depends on the set of forbidden minors).

%---------------------------------------------------------------------------------------------------------------------------------------------------------

\section{Uniform characterization using oblivious testers and forbidden subgraphs}
\label{sec:uniform-characterization}

\SArtur{I don't know if we want this section --- I'm leaving the decision to you.}As mentioned in Section \ref{subsec:uniform-characterization}, while Theorem \ref{thm:canonical-tester} from \cite{CFPS19} allows to simplify the analysis of testable properties, the analysis as in Theorem \ref{thm:main-1-sided-reduces-to-H-freeness} obtains non-uniform testers, in the sense of the dependency on $n$. In this section, we consider a special class of uniform testers, which we call \emph{oblivious testers}, that capture the essence of testers of testable properties in the flavor of Theorem \ref{thm:canonical-tester} (see \cite{AS08} for a similar notion in the context of testing dense graphs).

\begin{definition}\textbf{(Oblivious tester)}\it
\label{def:oblivious-tester}
A tester (one-sided or two-sided) for a graph property $\mathcal{P}$ is said to be \textbf{\emph{oblivious}} if it works as follows: Given an $\eps$, $0 < \eps < 1$, the tester
\begin{itemize}
\item computes an integer $q = q(\eps)$,
\item queries $q$ times the random vertex oracle to obtain a set (possibly, a multiset) $S$ of $q$ random vertices,
\item from each vertex $v \in S$, runs \Traverse\,${(G,v,q,q)}$ to get a $(q,q)$-bounded disc~$U_v$,
\item and then accepts or rejects (possibly randomly) according to $\eps$ and the visited graph $\bigcup_v U_v$.
\end{itemize}
\end{definition}

Notice that thanks to Theorem \ref{thm:canonical-tester}, Definition \ref{def:oblivious-tester} captures the essence of property testing in the random neighbor oracle model, and in that context, it is natural to consider oblivious testers.

\begin{remark}\SArtur{This remark is heavily based on a discussion in \cite{AS08}.}
While oblivious testers seem to be quite natural in our setting, there are two major restrictions that Definition \ref{def:oblivious-tester} imposes on an oblivious tester. The first is that such a tester cannot use the size of the input in order to determine the parameter $q$ which is later used for the size of the sample set $S$ and for the depth and breadth of the bounded discs. While this seems to be a rather simple assumption, it is not difficult to construct non-oblivious testers whose query complexity is $O_{\eps}(1)$, upper bounded by a function of $\eps$, but in fact it depends on the size of the graph (e.g., $q(\eps,n) = 1/\eps +(-1)^n$). Though this seems like a non-important and annoying technicality, it has been noted in other property testing models (see, e.g., \cite{AS08b}) that this subtlety may have nontrivial implications. The second restriction on an oblivious tester is that it cannot use the size of the input in order to make its decisions after the $q$ copies of $(q,q)$-bounded disc has been visited by the tester. (A similar phenomenon has been also noted earlier (cf. \cite{AS08}).) For example, \cite{AS08} gave the following simple example: A graph on an even number of vertices satisfies $\mathcal{P}$ if and only if it is bipartite, while a graph on an odd number of vertices satisfies $\mathcal{P}$ if and only if it is triangle-free. Any tester for $\mathcal{P}$ must use the size of the input graph in order to make its decision.
\end{remark}

Notice that in Definitions \ref{def:semi-H-freeness} and \ref{def:semi-H-rfreeness}, the families of finite graphs $\mathcal{H}$ depend on the graph property $\mathcal{P}$, $\eps$, and $n$. If $\mathcal{H}$ is independent of $n$ (that is, $\mathcal{H}$ depends only on $\mathcal{P}$ and $\eps$), then we will call $\mathcal{P}$ in Definitions \ref{def:semi-H-freeness} and \ref{def:semi-H-rfreeness}, respectively, \emph{uniformly semi-subgraph-free} and \emph{uniformly semi-rooted-subgraph-free}.\SArtur{Since I'm not sure whether we want to keep the discussion about oblivious testers, I've put here the definitions of uniformly semi-subgraph-free and uniformly semi-rooted-subgraph-free properties. If however, we will want to keep Theorem
\ref{thm:main-1-sided-oblivious-reduces-to-H-freeness} then, then maybe these definitions could go directly to Definitions \ref{def:semi-H-freeness} and \ref{def:semi-H-rfreeness}.}

With the definitions of oblivious testers, uniformly semi-subgraph-free and uniformly semi-rooted-subgraph-free properties, and Lemma \ref{lemma:semi-rooted-subgraph-free-yields-semi-subgraph-free} at hand, we can obtain a variant of Theorem \ref{thm:main-1-sided-reduces-to-H-freeness} for oblivious testers.

\begin{theorem}
\label{thm:main-1-sided-oblivious-reduces-to-H-freeness}
If a graph property $\mathcal{P}$ has an oblivious one-sided error tester in the random neighbor oracle model then $\mathcal{P}$ is uniformly semi-subgraph-free.
\end{theorem}

\begin{proof}
We follow the proof of Theorem \ref{thm:main-1-sided-reduces-to-H-freeness}. As before, thanks to Lemma \ref{lemma:semi-rooted-subgraph-free-yields-semi-subgraph-free}, it is enough to show that if a graph property $\mathcal{P}$ has an oblivious one-sided error tester then $\mathcal{P}$ is uniformly semi-\emph{rooted}-subgraph-free.

Let $\mathcal{P}$ be a graph property that has an oblivious one-sided error tester $\mathfrak{T}$. Fix $\eps$, $0 < \eps < 1$. We define $\mathcal{H}$ as a family of rooted graphs, such that a rooted graph $H$ belongs to $\mathcal{H}$, if for some input graph $G$, when the tester $\mathfrak{T}$ is run on $G$ with given $\eps$, then with positive probability
\begin{inparaenum}[(i)]
\item $\mathfrak{T}$ visits (exactly)\SArtur{I'm not sure if I like phrase ``(exactly)'', but here I wanted to stress that the entire visited subgraph $U$ must be root-preserving isomorphic to $H$, not just a subgraph of $U$.} a subgraph of $G$ that is root-preserving isomorphic to $H$ and
\item $\mathfrak{T}$ rejects $G$.
\end{inparaenum}
Observe that $\mathcal{H}$ is independent of $n$. We will show that so defined family $\mathcal{H}$ of rooted graphs satisfies the conditions in Definition~\ref{def:semi-H-rfreeness}, proving that $\mathcal{P}$ is uniformly semi-subgraph-free.

Let us first notice that each rooted graph $\mathcal{H}$ has at most $2 (q(\eps))^{q(\eps)}$ vertices and at most $2 (q(\eps))^{q(\eps)}$ edges, and so $\mathcal{H}$ is a finite family of finite rooted graphs.

Let us next show item {\it (i)} of Definition \ref{def:semi-H-rfreeness}, that any graph $G$ satisfying $\mathcal{P}$ is $\mathcal{H}$-rooted-free. The proof is by contradiction. Suppose that there is a graph $G$ satisfying $\mathcal{P}$ which contains a rooted copy of $H \in \mathcal{H}$.
By definition of $\mathcal{H}$, there must be an input graph $G'$, such that $G'$ has a rooted copy of $H$, and if $\mathfrak{T}$ is run on $G'$ with the fixed $\eps$, then with positive probability, $\mathfrak{T}$ visits that rooted copy of $H$ and then rejects $G'$.
But this implies that if for that $\eps$ we run $\mathfrak{T}$ on $G$, then also with positive probability $\mathfrak{T}$ visits that rooted copy of $H$ in $G$. But since on that basis $\mathfrak{T}$ rejects $G'$ with positive probability, so it must do for $G$. This means that the tester has a nonzero probability of rejecting $G$, contradicting our assumption that the tester $\mathfrak{T}$ is one-sided.

Now, we want to prove item {\it (ii)} of Definition \ref{def:semi-H-freeness}. Let $G$ be a graph that is $\eps$-far from satisfying $\mathcal{P}$. Any tester for $\mathcal{P}$ should reject $G$ with nonzero probability. By definition of an oblivious tester, $G$ must contain a rooted subgraph $H$ such that if the tester $\mathfrak{T}$ gets $H$ from the oracle, then it rejects $G$. By definition of $H$ this means that $H \in \mathcal{H}$, which proves item {\it (ii)} of Definition \ref{def:semi-H-freeness}.

We showed that if $\mathcal{P}$ has an oblivious one-sided error tester then $\mathcal{P}$ is uniformly semi-rooted-subgraph-free. By Lemma \ref{lemma:semi-rooted-subgraph-free-yields-semi-subgraph-free}, this yields that $\mathcal{P}$ is uniformly semi-subgraph-free, completing the proof.
\end{proof}

%---------------------------------------------------------------------------------------------------------------------------------------------------------

\section{Auxiliary tools: Simplifying condition (\ref{part-a-lemma:ExistenceOfH}) of Lemma \ref{lemma:ExistenceOfH}}
\label{sec:condition-a-prime}

In this section we show how one can simplify condition (\ref{part-a-lemma:ExistenceOfH}) of Lemma \ref{lemma:ExistenceOfH} and prove Lemma \ref{lemma:transformation}. Let us recall that Lemma \ref{lemma:transformation} states that if there is a graph $\G[\DCH]$ with a linear number of edge-disjoint colored copies of $H$, then there is always a subset $\DCH' \subseteq  \DCH$ with cardinality $|\DCH'| = \Omega_{\eps,H} (|\DCH|)$ such that the graph $\G[\DCH']$ satisfies property (\ref{part-a-lemma:ExistenceOfH}).

Our arguments follow the approach presented in \cite{CMOS11}. We begin by showing that condition (\ref{part-a-lemma:ExistenceOfH}) of Lemma \ref{lemma:ExistenceOfH} is implied by a simple condition on the degrees of the vertices in $\U$, namely, the degree of each vertex is either $0$ or is a constant factor of its corresponding degree in $G$.

\begin{lemma}[\textbf{Property (\ref{part-a-lemma:ExistenceOfH}')}]
\label{lemma:basic}
Let $G = (V,E)$ be a simple graph and let $\dg, \ld = \Theta_{\eps,H}(1)$. Let $\U$ be a subgraph of $G$ on vertex set $V$ such that the following property holds:
\begin{itemize}
\item[(\ref{part-a-lemma:ExistenceOfH}')] for every vertex $v \in V$, either $deg_{\U}(v) = 0$ or $deg_{\U}(v) = \Omega_{\eps,H}(deg_G(v))$.
\end{itemize}
Then property (\ref{part-a-lemma:ExistenceOfH}) of Lemma \ref{lemma:ExistenceOfH} is satisfied, that is, if \RLBD{($\U,H,\dg,\ld$)} finds a copy of $H$ in $\U$ with probability $\Omega_{\eps,H}(1)$, then \RLBD{($G,H,\dg,\ld$)} finds a copy of $H$ in $G$ with probability $\Omega_{\eps,H}(1)$.
\end{lemma}

\begin{proof}
Take any set of edges $\mathcal{E}$ that can be found by a single call of \RLBFS\,${(\U,\dg,\ld)}$ such that the subgraph of $\U$ induced by the edges $\mathcal{E}$ contains a copy of $H$. Since $\U$ is a subgraph of $G$, \RLBFS\,${(G,\dg,\ld)}$ can find (explore) the same edge set $\mathcal{E}$. Now, we will estimate the relation between the probability that \RLBFS\,${(\U,\dg,\ld)}$ finds $\mathcal{E}$ and the probability that \RLBFS\,${(G,\dg,\ld)}$ finds~$\mathcal{E}$.

By the assumption of the lemma, every vertex visited during the finding of $\mathcal{E}$ must have $\deg_{\U}(v) = \Omega_{\eps,H}(\deg_G(v))$ (since these vertices cannot be isolated in $\U$). Therefore, at every step of the exploration algorithm \RLBFS\,${(\U,\dg,\ld)}$, the probability of following a single edge from $\mathcal{E}$ decreases in $G$ by at most a factor of $O_{\eps,H}(1)$, compared to $\U$. Overall the probability of finding $\mathcal{E}$ in $G$ versus finding it in $\U$ decreases by at most a factor of $\left(O_{\eps,H}(1) \right)^{|\mathcal{E}|} = \left(O_{\eps,H}(1) \right)^{O_{\eps,H}(1)} = O_{\eps,H}(1)$.
\end{proof}

%---------------------------------------------------------------------------------------------------------------------------------------------------------

Lemma \ref{lemma:basic} provides a useful tool that simplifies the framework from Lemma \ref{lemma:ExistenceOfH}, and Lemma \ref{lemma:transformation} shows that in fact the condition on degrees can be always obtained by a simple reduction. That is, if there is a graph $\G[\DCH]$ with a \emph{linear number of edge-disjoint colored copies of $H$}, then Lemma \ref{lemma:transformation} shows that there is always a subset $\DCH' \subseteq  \DCH$ with cardinality $|\DCH'| = \Omega_{\eps,H} (|\DCH|)$ such that the graph $\G[\DCH']$ satisfies property (a) via showing that it satisfies property (a').

\junk{
\begin{lemma}[\textbf{Transformation to obtain property (\ref{part-a-lemma:ExistenceOfH}')}]
\label{lemma:transformation}
Let $G = (V,E)$ be a planar graph. Let $\DCH$ be a set of $\Omega_{\eps,H}(|V|)$ edge-disjoint colored copies of $H$ in $G$. Then there exists a subset $\DCH' \subseteq \DCH$, $|\DCH'| = \Omega_{\eps,H}(|V|)$, such that the graph $\G[\DCH']$ satisfies condition (\ref{part-a-lemma:ExistenceOfH}') of Lemma \ref{lemma:basic}. That is, for every $v \in V$, either $\deg_{\G[\DCH']}(v) = 0$ or $\deg_{\G[\DCH']}(v) = \Omega_{\eps,H}(\deg_G(v))$.
\end{lemma}
}

\LPaI* % STATEMENT OF Lemma \ref{lemma:transformation}

\begin{proof}
We will show that if $\DCH$ is a set of $\Omega_{\eps,H}(|V|)$ edge-disjoint colored copies of $H$ in $G$, then there exists a subset $\DCH' \subseteq \DCH$, $|\DCH'| = \Omega_{\eps,H}(|V|)$, such that the graph $\G[\DCH']$ satisfies condition (\ref{part-a-lemma:ExistenceOfH}') of Lemma \ref{lemma:basic} (that is, for every $v \in V$, either $\deg_{\G[\DCH']}(v) = 0$ or $\deg_{\G[\DCH']}(v) = \Omega_{\eps,H}(\deg_G(v))$). By Lemma \ref{lemma:basic}, this yields the proof of Lemma \ref{lemma:transformation}.

We construct the subset $\DCH'$ by deleting some copies of $H$ from $\DCH$. The process of deleting copies of $H$ is based on the comparison of the original degree of the vertices with the current degree in $\G[\DCH']$. To implement this scheme, we write $\deg_G(v)$ to denote the degree of $v$ in the original graph $G$ and we use the term \emph{current degree} of a vertex $v$ to denote its current degree in the graph $\G[\DCH']$ induced by the \emph{current} set $\DCH'$ of copies of $H$ (where ``current'' means at a given moment in the process). Let $\alpha = \frac{|\DCH|}{|V|} = \Omega_{\eps,H}(1)$. We repeat the following procedure as long as possible: if there is a non-isolated vertex $v \in V$ with current degree in $\G[\DCH']$ at most $\frac{\alpha}{12} \deg_G(v)$, then we delete from $\DCH'$ all copies of $H$ in the current $\DCH'$ incident to $v$. To estimate the number of copies of $H$ deleted, we charge to $v$ the number of deleted copies of $H$ in each such operation. Observe that each $v \in V$ will be processed not more than once. Indeed, once $v$ has been used, it becomes isolated, and hence it is not used again. Therefore, at most $\frac{\alpha}{12} \deg_G(v)$ copies of $H$ from $\DCH'$ can be charged to any single vertex. This, together with the inequality $\sum_{v \in V} \deg_G(v) \le 6 |V|$ by planarity of $\G[\DCH']$, implies that the total number of copies of $H$ removed from $\DCH$ to obtain $\DCH'$ is upper bounded by $\sum_{v \in V} \frac{\alpha}{12} \deg_G(v) \le \frac{\alpha}{2} |V|$. Since $|\DCH| = \alpha |V|$, we conclude that $|\DCH'| \ge |\DCH| - \frac{\alpha}{2} |V| = \frac{\alpha}{2} |V| = \Omega_{\eps,H}(|V|)$.
\end{proof}

%---------------------------------------------------------------------------------------------------------------------------------------------------------

\section{Some basic properties of the process of shrinking $H$ and hypergraph representation of $H$ by $\M_i$ (Section \ref{subsec:shrinking-H})}
\label{proofs-subsec:shrinking-H}

In this section we present some basic properties of the process of shrinking $H$ and hypergraph representation of $H$ by $\M_i$, as defined in Section \ref{subsec:shrinking-H}. While not all of them are necessary for our analysis, we believe they are useful to better understand the ideas behind our approach.

We begin with the following simple claim.

\begin{claim}
\label{claim:properties-of-labels-in-M_i}
For any $i$, $1 \le i \le \Hsize$,
\begin{itemize}
\item $V(\M_i) = \{v_i, v_{i+1}, \dots, v_{\Hsize}\}$,
\item for every hyperedge $\e \in E(\M_i)$, $\lab(\e) \subseteq V(H) \setminus V(\M_i)$, %\Artur{An equivalent way would be to say that $\bigcup_{\e \in E(\M_i)} \e \cap \bigcup_{\e \in E(\M_i)} \lab(\e) = \emptyset$.}
    and
\item for any $\e \in E(\M_i)$, every vertex in $\lab(\e)$ is adjacent in $H$ only to vertices in $\e \cup \lab(\e)$.%\Artur{That is, $\{(x,y) \in E(H): x \in \lab(\e), y \in V(H) \setminus (\e \cup (\lab(\e)))\} = \emptyset$.}.
\end{itemize}
\end{claim}

\begin{proof}
%\small% IT'S IN small SINCE IT'S RATHER STRAIGHTFORWARD/TRIVIAL
Let us first notice that the first fact that $V(\M_i) = \{v_i, v_{i+1}, \dots, v_{\Hsize}\}$ follows trivially from our construction, and so we focus on proving the other two claims.

The proof of the other two parts is by induction on $i$. For $i=1$ the claim is true since $\M_1 = H$ and since in $\M_1$, we have $\lab(\e) = \emptyset$ for every $\e$.
Therefore, let us assume the claim for $i < \Hsize$, and consider it for $i+1$.

The construction of $\M_{i+1}$ ensures that the only changes between $\M_i$ and $\M_{i+1}$ are in vertex $v_i$ and in the edges/hyperedges incident to $v_i$ in $\M_i$.

To see the second part of the claim, note that $\bigcup_{\e \in E(\M_{i+1})} \e = \bigcup_{\e \in E(\M_i)} \e \setminus \{v_i\}$ and $\bigcup_{\e \in E(\M_{i+1})} \lab(\e) = \{v_i\} \cup \bigcup_{\e \in E(\M_i)} \lab(\e)$, and hence the claim that $\bigcup_{\e \in E(\M_{i+1})} \e \cap \bigcup_{\e \in E(\M_{i+1})} \lab(\e) = \emptyset$ follows by induction.

To see the third part of the claim, if $\e \in E(\M_i)$ and $\e \in E(\M_{i+1})$, then the claim follows by induction. Otherwise, if $\e \in E(\M_{i+1})$ and $\e \not\in E(\M_i)$, then $\e = \N_i$. If $\mathcal{E}_i$ denotes the set of edges/hyperedges incident to vertex $v_i$ in $\M_i$, then $\lab(\N_i) = \{v_i\} \cup \bigcup_{\e^* \in \mathcal{E}_i} \lab(\e^*)$. Since by induction, for any $\e^* \in \mathcal{E}_i$ (which is an edge/hyperedge in $\M_i$), every vertex in $\lab(\e^*)$ is adjacent in $H$ only to vertices in $\e^* \cup \lab(\e^*)$, the fact that $\e^* \cup \lab(\e^*) \subseteq \N_i \cup \lab(\N_i)$ implies that every vertex in $\lab(\e^*)$ is adjacent in $H$ only to vertices in $\N_i \cup \lab(\N_i)$. Further, vertex $v_i$ is adjacent in $H$ only to vertices in $\N_i$ and some of vertices in $\bigcup_{\e^* \in \mathcal{E}_i} \lab(\e^*)$. Therefore, every vertex in $\lab(\N_i) = \{v_i\} \cup \bigcup_{\e^* \in \mathcal{E}_i} \lab(\e^*)$ is adjacent in $H$ only to vertices in $\N_i \cup \lab(\N_i)$.
\end{proof}

%---------------------------------------------------------------------------------------------------------------------------------------------------------

Let us state the following property of our construction that follows from our discussion.% (for hyperedges, condition $\N_j \cap \{v_{j+1},\dots,v_{i-1}\} = \emptyset$ follows from the fact that if there is $v_{\ell} \in \N_j$ with $j < \ell < i$, then vertex $v_{\ell}$ would have been contracted while creating $\M_{\ell+1}$ and the edge $\N_j$ would disappear from $\M_{\ell+1}$ and from any further $\M_{\ell+2}, \dots, \M_i$).

\begin{claim}
\label{claim:properties-of-Mi}
For every $i$, $1 \le i \le \Hsize$, the hypergraph $\M_i$ contains vertices $\{v_i,\dots,v_{\Hsize}\}$ and two types of edges:
\begin{itemize}
\item ``regular'' edges: if $(v_j,v_{\ell}) \in E(H)$ with $i \le j, \ell \le \Hsize$, then $(v_j,v_{\ell})$ is an edge in $\M_i$;
\item hyperedges: if there is $j$, $1 \le j < i$, with $\N_j \cap \{v_{j+1},\dots,v_{i-1}\} = \emptyset$ then $\N_j$ forms a hyperedge in $\M_i$.
\end{itemize}
\end{claim}

\begin{proof}
%\small% IT'S IN small SINCE IT'S RATHER STRAIGHTFORWARD/TRIVIAL
The proof is by induction. The claim trivially holds for $\M_1$, since $\M_1 = H$. Therefore, let us assume the claim for $\M_i$ with $i < \Hsize$, and consider it for $i+1$.

The construction of $\M_{i+1}$ ensures that its vertex set is $\{v_{i+1},\dots,v_{\Hsize}\}$ and the only changes between $\M_i$ and $\M_{i+1}$ are in vertex $v_i$ and in the edges/hyperedges incident to $v_i$ in $\M_i$. Any regular edge $(v_j,v_{\ell}) \in E(H)$ with $i \le j, \ell \le \Hsize$ in $\M_i$ stays as a regular edge in $\M_{i+1}$ if $j, \ell > i$. Therefore, if  $(v_j,v_{\ell}) \in E(H)$ with $i+1 \le j, \ell \le \Hsize$, then $(v_j,v_{\ell})$ is an edge in $\M_{i+1}$.

For hyperedges, a hyperedge $\N_j$ ($1 \le j < i$) in $\M_i$ stays as a hyperedge in $\M_{i+1}$ only if $v_i \notin \N_j$. Hence, any such $\N_j$ satisfies the property that $\N_j \cap \{v_{j+1},\dots,v_{i-1}\} = \emptyset$ and that $v_i \notin \N_j$, and therefore $\N_j \cap \{v_{j+1},\dots,v_i\} = \emptyset$.

Furthermore, our construction adds also a new single hyperedge $\N_i$ with all vertices in the hyperedge in $\{v_{i+1},\dots,v_{\Hsize}\}$. Therefore, such a new hyperedge $\N_i$ satisfies the property that $\N_i \cap \{v_{i+1},\dots,v_i\} = \emptyset$. Hence, in either case, if there is $j$, $1 \le j < i+1$, with $\N_j \cap \{v_{j+1},\dots,v_i\} = \emptyset$ then $\N_j$ forms a hyperedge in $\M_{i+1}$, as required.
\end{proof}

%---------------------------------------------------------------------------------------------------------------------------------------------------------

\junk{
\begin{claim}%[Modeling $H$ by $\M_i$]
\label{claim:H-modeled-by-Mi}\Artur{Do we really want or need this claim?}
Let $1 \le i \le \Hsize$. A graph $H = (V(H), E(H))$ is represented by a hypergraph $\M_i$ such that:
\begin{itemize}
\item the set of ``regular'' edges in $\M_i$ is equal to the set of edges in the subgraph of $H$ induced by the vertex set of $\M_i$, and
\item every hyperedge $\e$ in $\M_i$ represents the subgraph of $H$ with the vertex set $\e \cup \lab(\e)$ and the edge set $\{(x,y) \in E(H): x \in \lab(\e) \text{ and } y \in \e \cup \lab(\e)\}$.
\end{itemize}
Furthermore, all so obtained subgraphs of $H$ are edge-disjoint and their union defines $H$.
\end{claim}

\begin{proof}
%\small% IT'S IN small SINCE IT'S RATHER STRAIGHTFORWARD/TRIVIAL
The proof is by induction on $i$. The claim obviously holds for $i=1$, since $\M_1 = H$. Next, assuming the claim holds for $i$, let us consider the claim for $i+1$. Since the first part about regular edges follows immediately from Claim \ref{claim:properties-of-Mi}, we will only focus on the other parts.

There are two changes between $\M_i$ and $\M_{i+1}$. Firstly, we remove from $\M_i$ set $\mathcal{E}_i$ of all edges and hyperedges that contain vertex $v_i$. And next, we add a new hyperedge $\N_i = \bigcup_{\e \in \mathcal{E}_i} \e \setminus \{v_i\}$ with label $\lab(\N_i) = \{v_i\} \cup \bigcup_{\e \in \mathcal{E}_i} \lab(\e)$. We will consider separately the edges in $\mathcal{E}_i$ and the hyperedges in $\mathcal{E}_i$. Let $\mathcal{E}_i^E$ denote the set of all ``regular'' edges in $\mathcal{E}_i$ incident to $v_i$ in $\M_i$ (all these edges belong to $E(H)$) and let $\mathcal{E}_i^h$ denote the set of all hyperedges in $\mathcal{E}_i$ containing $v_i$ in $\M_i$.

For regular edges $e \in \mathcal{E}_i^E$, let $\N_i^E$ be the set of neighbors of $v_i$ in $\mathcal{E}_i^E$, that is, by induction, $\N_i^E = \{ v_j \in V(H): (v_i,v_j) \in E(H) \text{ and } j > i\}$.
Similarly, by induction, every hyperedge $\e \in \mathcal{E}_i^h$ represents the subgraph of $H$ with the vertex set $\e \cup \lab(\e)$ and the edge set $\{(x,y) \in E(H): x \in \lab(\e), y \in \e \cup \lab(\e)\}$.
Therefore, by induction, all edges and hyperedges $\e \in \mathcal{E}_i$ together represent the subgraph of $H$ with
\begin{itemize}
\item the vertex set $\N_i^E \cup \{v_i\} \cup \bigcup_{\e \in \mathcal{E}_i^h} (\e \cup \lab(\e)) = \bigcup_{\e \in \mathcal{E}_i} (\e \cup \lab(\e))$ and
\item the edge set $\mathcal{E}_i^E \cup \bigcup_{\e \in \mathcal{E}_i^h} \{(x,y) \in E(H): x \in \lab(\e), y \in \e \cup \lab(\e) \}$.
\end{itemize}
%Now, note that $\bigcup_{\e \in \mathcal{E}_i} (\e \cup \lab(\e)) = \N_i \cup \lab(\N_i)$.
Hence, it suffices to show that the graph with the vertex set $\bigcup_{\e \in \mathcal{E}_i} (\e \cup \lab(\e))$ and the edge set $\mathcal{E}_i^E \cup \bigcup_{\e \in \mathcal{E}_i^h} \{(x,y) \in E(H): x \in \lab(\e), y \in \e \cup \lab(\e) \}$ is identical to the graph induced (in the sense of Claim \ref{claim:H-modeled-by-Mi}) by the hyperedge $\N_i$. (We note that by induction, this also implies the edge-disjointness of our construction.)

The graph induced by the hyperedge $\N_i$ has
\begin{itemize}
\item the vertex set $\N_i \cup \lab(\N_i)$ and
\item the edge set $\{(x,y) \in E(H): x \in \lab(\N_i), y \in \N_i \cup \lab(\N_i)\}$,
\end{itemize}
where $\N_i = \bigcup_{\e \in \mathcal{E}_i} \e \setminus \{v_i\}$ and $\lab(\N_i) = \{v_i\} \cup \bigcup_{\e \in \mathcal{E}_i} \lab(\e)$.

Since $\N_i \cup \lab(\N_i) = \left(\bigcup_{\e \in \mathcal{E}_i} \e \setminus \{v_i\}\right) \cup \left(\{v_i\} \cup \bigcup_{\e \in \mathcal{E}_i} \lab(\e)\right) = \bigcup_{\e \in \mathcal{E}_i} \e \cup \{v_i\} \cup \bigcup_{\e \in \mathcal{E}_i} \lab(\e)$, and since $v_i \in \e$ for every $\e \in \mathcal{E}_i$, we obtain that the vertex set of the graph induced by $\N_i$ is $\N_i \cup \lab(\N_i) = \bigcup_{\e \in \mathcal{E}_i} (\e \cup \lab(\e))$, as required.

Next, let us consider the edges of the graph induced by the hyperedge $\N_i$, that is, the set $\{(x,y) \in E(H): x \in \lab(\N_i), y \in \N_i \cup \lab(\N_i)\}$, which is the same as $\{(x,y) \in E(H): x \in \{v_i\} \cup \bigcup_{\e \in \mathcal{E}_i^h} \lab(\e), y \in \bigcup_{\e \in \mathcal{E}_i} (\e \cup \lab(\e))\}$. We split this set into four separate subsets:
\begin{itemize}
\item $\{(v_i,y) \in E(H): y \in \bigcup_{\e \in \mathcal{E}_i} \e\}$, which is equal to $\mathcal{E}_i^E$.
\item $\{(v_i,y) \in E(H): y \in \bigcup_{\e \in \mathcal{E}_i^h} \lab(\e)\}$ which is the same as $\bigcup_{\e \in \mathcal{E}_i^h} \{(v_i,y) \in E(H): y \in \lab(\e)\}$.
\item $\{(x,y) \in E(H): x \in \bigcup_{\e \in \mathcal{E}_i^h} \lab(\e), y \in \bigcup_{\e \in \mathcal{E}_i^h} \lab(\e)\}$, which, since by Claim \ref{claim:properties-of-labels-in-M_i} $\lab(\e) \cap \lab(\e') = \emptyset$ for distinct $\e, \e' \in \mathcal{E}_i^h$, is the same as $\bigcup_{\e \in \mathcal{E}_i^h} \{(x,y) \in E(H): x, y \in \lab(\e)\}$.
\item $\{(x,y) \in E(H): x \in \bigcup_{\e \in \mathcal{E}_i^h} \lab(\e), y \in \bigcup_{\e \in \mathcal{E}_i} \e \}$, which since $\bigcup_{\e \in \mathcal{E}_i} \e = \N_i \cup \{v_i\}$, is the same as $\bigcup_{\e \in \mathcal{E}_i^h} \{(x,y) \in E(H): x \in \lab(\e), y \in \N_i \cup \{v_i\}\}$, which in turn, by Claim \ref{claim:properties-of-labels-in-M_i}, is the same as $\bigcup_{\e \in \mathcal{E}_i^h} \{(x,y) \in E(H): x \in \lab(\e), y \in \e \}$.
\end{itemize}

If we take the union of all these four sets, then we obtain that the set of the edges of the graph induced by the hyperedge $\N_i$ is equal to the following:
\begin{align*}
    \{&(x,y) \in E(H): x \in \lab(\N_i), y \in \N_i \cup \lab(\N_i)\}
        =
    \{(x,y) \in E(H): x \in \{v_i\} \cup \bigcup_{\e \in \mathcal{E}_i^h} \lab(\e), y \in \bigcup_{\e \in \mathcal{E}_i} (\e \cup \lab(\e))\}
        \\
        &=
    \{(v_i,y) \in E(H): y \in \bigcup_{\e \in \mathcal{E}_i} \e\}
        \cup
    \{(v_i,y) \in E(H): y \in \bigcup_{\e \in \mathcal{E}_i^h} \lab(\e)\}
        \ \cup
        \\
        &\qquad\qquad
    \{(x,y) \in E(H): x \in \bigcup_{\e \in \mathcal{E}_i^h} \lab(\e), y \in \bigcup_{\e \in \mathcal{E}_i^h} \lab(\e)\}
        \cup
    \{(x,y) \in E(H): x \in \bigcup_{\e \in \mathcal{E}_i^h} \lab(\e), y \in \bigcup_{\e \in \mathcal{E}_i} \e \}
        \\
%        &=
%    \mathcal{E}_i^E
%        \cup
%    \bigcup_{\e \in \mathcal{E}_i^h} \{(v_i,y) \in E(H): y \in \lab(\e)\}
%        \cup
%    \bigcup_{\e \in \mathcal{E}_i^h} \{(x,y) \in E(H): x, y \in \lab(\e)\}
%        \ \cup
%        \\
%        &\qquad\qquad
%    \bigcup_{\e \in \mathcal{E}_i^h} \{(x,y) \in E(H): x \in \lab(\e), y \in \e\}
%        \\
        &=
    \mathcal{E}_i^E
        \cup
    \bigcup_{\e \in \mathcal{E}_i^h}
        \{(v_i,y) \in E(H): y \in \lab(\e)\}
            \cup
        \{(x,y) \in E(H): x, y \in \lab(\e)\}
            \cup
        \{(x,y) \in E(H): x \in \lab(\e), y \in \e\}
        \\
%        &=
%    \mathcal{E}_i^E
%        \cup
%    \bigcup_{\e \in \mathcal{E}_i^h}
%        \{(x,y) \in E(H): x \in \lab(\e),
%            y \in (\e \cup \{v_i\} \cup \lab(\e))\}
%        \\
        &=
    \mathcal{E}_i^E
        \cup
    \bigcup_{\e \in \mathcal{E}_i^h}
        \{(x,y) \in E(H): x \in \lab(\e),
            y \in (\e \cup \lab(\e))\}
        \enspace,
\end{align*}
where in the last equation we used the fact that $v_i \in \e$ for every $\e \in \mathcal{E}_i^h$. This completes the proof.
\end{proof}
}

%---------------------------------------------------------------------------------------------------------------------------------------------------------

\section{Basic properties of consistent hypergraphs}
\label{subsec:properties-of-consistent-hypergraphs}

In this section we will present some basic properties of consistent hypergraphs (as defined in Section \ref{subsubsec:safe-vertices-and-consistent-hypergraphs}) used in our analysis.

Let us begin with the following simple claim.

\begin{claim}
\label{claim:consistent-for-subset}
Let $\DDCH{i}$ be a set of edge-disjoint colored copies of $H$ in $G$ such that the hypergraph $\HQ_i(\DDCH{i})$ is consistent for $\DDCH{i}$. Then for any $\DCH \subseteq \DDCH{i}$, the hypergraph $\HQ_i(\DCH)$ is consistent for $\DCH$.
\end{claim}

\begin{proof}
%\Artur{Isn't it straightforward? Or does it require a proof?}
%
By Remark \ref{remark:independent-shrinking}, we can define the hypergraph $\HQ_i(\DDCH{i})$ independently for each copy of $H$ in $\DDCH{i}$. Thus, if $\DCH \subseteq \DDCH{i}$ then $\HQ_i(\DCH)$ is a sub-hypergraph of $\HQ_i(\DDCH{i})$, and hence all safe vertices in $\HQ_i(\DDCH{i})$ are also safe in $\HQ_i(\DCH)$, meaning that $\HQ_i(\DCH)$ is consistent for $\DCH$.
\end{proof}

Let us also state the following simple claim.

\begin{claim}
\label{claim:final-hypergraph-invariant}
Let $\DDCH{1}, \DDCH{2}, \dots, \DDCH{i}$ be a set of edge-disjoint colored copies of $H$ in $G$ with $\DDCH{i} \subseteq \dots \subseteq \DDCH{2} \subseteq \DDCH{1}$. Let $\HQ_1(\DDCH{1}), \HQ_2(\DDCH{2}), \dots, \HQ_i(\DDCH{i})$ be the sequence of hypergraphs constructed by the algorithm above, with each $\HQ_j(\DDCH{j})$ consistent for $\DDCH{j}$. Then the same hypergraph $\HQ_i(\DDCH{i})$ would be obtained if we applied the algorithm above for the sequence $\DDCH{j} = \DDCH{i}$, for every $1 \le j \le i$.
\qed
\end{claim}

\junk{
\begin{proof}\small% IT'S IN small SINCE IT'S RATHER STRAIGHTFORWARD/TRIVIAL
\textcolor[rgb]{0.51,0.00,0.00}{Isn't it straightforward? Or does it require a proof?}
\end{proof}
}

%---------------------------------------------------------------------------------------------------------------------------------------------------------

\junk{
Similarly as we did it in Section \ref{subsec:shrinking-H}, the description above corresponds to the following modeling of the set $\DDCH{i}$ of edge-disjoint colored copies of $H$ in $G$ by the hypergraph $\HQ_i$.

\begin{claim}
\label{claim:copies-of-H-modeled-by-Qi}
The edge set of graph $\G[\DDCH{i}]$ is represented by a hypergraph $\HQ_i(\DDCH{i})$ consistent for $\DDCH{i}$ such that:
\begin{itemize}
\item the set of ``regular'' edges in $\HQ_i(\DDCH{i})$ is equal to the set of edges in the subgraph of $\G[\DDCH{i}]$ induced by the vertex set of $V(\HQ_i(\DDCH{i}))$, and
\item every hyperedge $\e$ with label $\lab(\e)$ in $\HQ_i(\DDCH{i})$ represents the subgraph of $\G[\DDCH{i}]$ with the vertex set $\e \cup \lab(\e)$ and the edge set $\{(x,y) \in E(\G[\DDCH{i}]): x \in \e \text{ and } y \in \e \cup \lab(\e)\}$.
\end{itemize}
Furthermore, every edge from $\G[\DDCH{i}]$ appears in this representation exactly once.
\end{claim}

\begin{proof}
%\small% IT'S IN small SINCE IT'S RATHER STRAIGHTFORWARD/TRIVIAL
\Artur{Does it make sense? Double-check the proof.}%
The proof is by induction. The claim is obviously true for $i=1$, since $\HQ_1(\DDCH{1}) = \G[\DDCH{1}]$.

Next, assume the claim holds for $i$ and we show by induction that it holds for $i+1 \le \Hsize$.

Since $\DDCH{i+1} \subseteq \DDCH{i}$, by the inductive hypothesis and Claims \ref{claim:consistent-for-subset} and \ref{claim:final-hypergraph-invariant}, we have that
\begin{itemize}
\item the set of ``regular'' edges in $\HQ_i(\DDCH{i+1})$ is equal to the set of edges in the subgraph of $\G[\DDCH{i+1}]$ that is induced by the vertex set of $V(\HQ_i(\DDCH{i+1}))$, and
\item every hyperedge $\e$ in $\HQ_i(\DDCH{i+1})$ represents the subgraph of $\G[\DDCH{i+1}]$ with the vertex set $\e \cup \lab(\e)$ and the edge set $\{(x,y) \in E(\G[\DDCH{i+1}]): x \in \e \text{ and } y \in \e \cup \lab(\e)\}$.
\end{itemize}
\end{proof}
}

We will also use the following property of consistent hypergraphs.

\begin{claim}
\label{claim:neighbors-in-consistent-hypergraphs}
Let $\DDCH{i}$ be a set of edge-disjoint colored copies of $H$ in $G$. Let $\HQ_1(\DDCH{i}), \HQ_2(\DDCH{i}), \dots, \HQ_i(\DDCH{i})$ be the sequence of hypergraphs constructed by the algorithm above (cf. Section \ref{subsec:shrinking-copies-of-H}), with each $\HQ_j(\DDCH{i})$ being consistent for $\DDCH{i}$, $1 \le j \le i$. Then, for every $j$, $1 \le j \le i-1$, for any vertex $u \in V(\HQ_j(\DDCH{i}))$ with $\chi(u) = \chi(v_j)$, the neighbors in $\HQ_j(\DDCH{i})$ have distinct colors.
\end{claim}

\begin{proof}
%\small% IT'S IN small SINCE IT'S RATHER STRAIGHTFORWARD/TRIVIAL
The proof follows directly from the definition of safe vertices and consistent hypergraphs. Indeed, since $\HQ_{j+1}(\DDCH{i})$ is consistent for $\DDCH{i}$, by definition, every vertex $u \in V(\HQ_j(\DDCH{i}))$ with $\chi(u) = \chi(v_j)$ is safe with respect to $\DDCH{i}$ and $\HQ_j(\DDCH{i})$. That is, from definition of being safe, for all colored copies $\h \in \DDCH{i}$ of $H$ that contain $u$, the sets $\N_j^{\h}\langle u\rangle$ are the same, where $\N_j^{\h}\langle u \rangle$ is the set of neighbors of $u$ in $\h$ in the hypergraph $\HQ_j(\DDCH{i})$. Since every copy $\h \in \DDCH{i}$ of $H$ consists of vertices of distinct colors, this yields the claim.
\end{proof}

%---------------------------------------------------------------------------------------------------------------------------------------------------------

\section{Lemma \ref{lemma:central-small-degrees}: Planarization of hypergraphs via shadow graphs}
\label{subsec:planarization-of-hypergraphs}

\newcommand{\fr}{\ensuremath{\mathfrak{col}}}
\newcommand{\FREEC}{\ensuremath{\mathfrak{\alpha}}}
\renewcommand{\FREEC}{\ensuremath{\Xi}}

\junk{\Artur{Notation:
    \begin{compactitem}
    \item \texttt{$\backslash$FREEC} denotes $\FREEC$, a set of colors of vertices in $\HQ_i(\DDCH{i})$ (e.g., $\FREEC = \{1, \dots, \Hsize\} \setminus \{\chi(v_j): 1 \le j < i\}$)
    \item \texttt{$\backslash$fr} denotes $\fr$, a color, typically from $\FREEC$
    \item \texttt{$\backslash$GG} denotes $\GG$, also $\GG_1, \dots, \GG_i$ and as a shadow graph $\GG(\HQ_i(\DDCH{i}))$
    \item \texttt{$\backslash$GG\_i\^{\fr}} denotes $\GG_i^{\fr}$, the color-$\fr$ shadow graph of the hypergraph $\HQ_i(\DDCH{i})$
    \item \texttt{$\backslash$Gii = ($\backslash$Vii, $\backslash$Eii)} denote $\Gii = (\Vii,\Eii)$, a simple graph that is a union of graphs $\GG^{\fr}(\HQ_i(\DDCH{i}))$
    \item \texttt{$\backslash$NEIG\{$\alpha$\}\{$\beta$\}} (with two parameters $\alpha, \beta$) denotes $\NEIG{\alpha}{\beta}$, set of neighbors of vertex $\beta$ in graph/hypergraph~$\alpha$
    \item \texttt{$\backslash$HQ} denotes $\HQ$ (also in $\HQ_i(\DDCH{i})$, which is a hypergraph constructed from $\DDCH{i}$ by contracting vertices of $H$)
    \end{compactitem}
}
}
In this section we show how to model hypergraphs $\HQ_i(\DDCH{i})$ using planar graphs (via the notion of \emph{shadow graphs}) to establish the proof of Lemma \ref{lemma:central-small-degrees}. In what follows, for fixed $i$, we will mimic the construction of the hypergraph $\HQ_i(\DDCH{i})$ to construct color-$\fr$ shadow graphs $\GG^{\fr}(\HQ_i(\DDCH{i}))$, one for each relevant color $\fr$, such that each $\GG^{\fr}(\HQ_i(\DDCH{i}))$ is \emph{planar} %(see Claim \ref{claim:shadow-graphs-are-planar})
and it \emph{maintains the neighborhood of all vertices of color $\fr$} in $\HQ_i(\DDCH{i})$. %(see Lemma \ref{lemma:key-property-color-shadow-graphs})
With this construction at hand, Lemma \ref{lemma:central-small-degrees} will easily follow.

%---------------------------------------------------------------------------------------------------------------------------------------------------------

As in the conditions of Lemma \ref{lemma:central-small-degrees}, let $\DDCH{i}$ be a set of edge-disjoint colored copies of $H$ in $G$ and let $\HQ_i(\DDCH{i})$ be a hypergraph consistent for $\DDCH{i}$. Let us recall how the hypergraph $\HQ_i(\DDCH{i})$ is built by our algorithm from Section \ref{subsec:shrinking-copies-of-H}. In the construction of $\HQ_{i}(\DDCH{i})$, we assume that we have already fixed $v_1, \dots, v_{i-1}$ (and we have \textbf{not} fixed the order of other vertices from $H$, since in fact, these choices will depend on our constructions of $\HQ_{i}(\DDCH{i}), \HQ_{i+1}(\DDCH{i+1}), \dots, \HQ_{\Hsize}(\DDCH{\Hsize})$). The algorithm takes $\HQ_{i-1}(\DDCH{i-1})$ with $\DDCH{i} \subseteq \DDCH{i-1}$, and first removes all hyperedges corresponding to the edge-disjoint copies of $H$ in $\DDCH{i-1} \setminus \DDCH{i}$ and then takes the set $\DDCH{i}$ of copies of $H$ and shrink them, in the same way as $\M_{i-1}$ is transformed into $\M_{i}$. Let us note that by Claim \ref{claim:final-hypergraph-invariant}, the hypergraph $\HQ_i(\DDCH{i})$ can be built by applying our algorithm above with all sets $\DDCH{j} = \DDCH{i}$ for all $j \le i$.

Before we proceed, let us introduce some useful notation. Fix $i$. Let $\FREEC = \{1, \dots, \Hsize \} \setminus \{\chi(v_j): j < i\}$, that is, $\FREEC$ is the set of the colors of vertices from $\HQ_i(\DDCH{i})$. Let us recall that since for any $j \le i$, the hypergraph $\HQ_j(\DDCH{i})$ is consistent for $\DDCH{i}$, by Claim \ref{claim:neighbors-in-consistent-hypergraphs}, every  vertex $u \in V(\HQ_{j-1}(\DDCH{i}))$ with $\chi(u) = \chi(v_{j-1})$ has all neighbors in $\HQ_{j-1}(\DDCH{i})$ with distinct colors. To facilitate this property, for any set $X \subseteq V$ consisting of vertices of distinct colors (e.g., $X = \e$ for an edge/hyperedge in $\HQ_j(\DDCH{i})$), if $X$ has a vertex of color from outside $\FREEC$ (that is, $\{\chi(x): x \in X\} \setminus \FREEC \ne \emptyset$), then we call a vertex $y$ in $X$ with $\chi(y) \notin \FREEC$ the \emph{lowest color vertex of $X$} if it minimizes $\ell$ with $\chi(y) = \chi(v_{\ell})$ (that is, for any vertex $z \in X$ with $\chi(z) \in \{\chi(v_j): j < i\}$, if $\chi(y) = \chi(v_{\ell})$ and $\chi(z) = \chi(v_{s})$ then $\ell \le s$).

Let $\fr$ be an arbitrary color from $\FREEC$. We mimic the algorithm that builds $\HQ_i(\DDCH{i})$ to create a sequence of graphs $\GG_1^{\fr}, \dots, \GG_i^{\fr}$ as follows:

%---------------------------------------------------------------------------------------------------------------------------------------------------------
\begin{walgo}\vspace*{-0.3in}
\begin{itemize}
\item Set $\GG_1^{\fr}$ to be equal to the graph $\G[\DDCH{i}]$ after removing all isolated vertices in $\G[\DDCH{i}]$.
\item For $j:=2$ to $i$, build $\GG_j^{\fr}$ as follows:
\begin{itemize}[$\diamond$]
\item Take vertex $v_{j-1} \in V(H)$.
\item For every vertex $u \in V(\HQ_{j-1}(\DDCH{i}))$ with $\chi(u) = \chi(v_{j-1})$:
    \begin{itemize}[$\circ$]
    \item let $\Gamma_{j-1}(u)$ be the set of all neighbors of $u$ in $\HQ_{j-1}(\DDCH{i})$ ($u \notin \Gamma_{j-1}(u)$);
    \item if $\Gamma_{j-1}(u)$ has a vertex of color not from $\FREEC$ (i.e., $\{\chi(x): x \in \Gamma_{j-1}(u)\} \setminus \FREEC \ne \emptyset$) then
    \begin{itemize}[$\triangleright$]
    \item let $w$ be a lowest color vertex in $\Gamma_{j-1}(u)$;
    %\item let $w \in \Gamma_{j-1}(u)$ with $\chi(w) \notin \FREEC$ that minimizes $\ell$ with $\chi(w) = \chi(v_{\ell})$;
    \item contract edge $(u,w)$ into vertex $w$;
    \end{itemize}
    \item else {\small (i.e., $\{\chi(x): x \in \Gamma_{j-1}(u)\} \subseteq \FREEC$)}, if there is $w \in \Gamma_{j-1}(u)$ with $\chi(w) = \fr$, then
    \begin{itemize}[$\triangleright$]
    \item contract edge $(u,w)$ into vertex $w$;
    \end{itemize}
    \item else, remove vertex $u$.
    \end{itemize}
\item Remove all parallel edges and all self-loops.
\end{itemize}
%\item Remove all edges containing no vertex of color $\fr$.
\end{itemize}
\end{walgo}
%---------------------------------------------------------------------------------------------------------------------------------------------------------

The graph $\GG_i^{\fr}$ will be called the \emph{color-$\fr$ shadow graph} of the hypergraph $\HQ_i(\DDCH{i})$ and will be denoted by $\GG^{\fr}(\HQ_i(\DDCH{i}))$.%\Artur{Figure needed.} %Notice that every edge in $\GG^{\fr}(\HQ_i(\DDCH{i}))$ contains a vertex of color $\fr$.

Observe that any $\GG_j^{\fr}$ is a simple graph (contains no self-loops nor parallel edges). Furthermore, to argue that the algorithm above makes sense, we will have to ensure that every time we perform contraction of an edge $(u,w)$ into vertex $w$, we must have that $(u,w)$ is an edge in $\GG_{j-1}^{\fr}$. Let us also notice that every time we refer to the lowest color vertex $w$ in the algorithm, by Claim \ref{claim:neighbors-in-consistent-hypergraphs}, this vertex is well defined (since all vertices from $\Gamma_{j-1}(u)$ have distinct colors).

In what follows, we will prove three keys properties of our construction,
%
%\begin{enumerate}[(1)]
\begin{compactenum}[(1)]
%\begin{inparaenum}[(i)]
\item that each $\GG_j^{\fr}$ is planar,
\item that for every contraction of an edge $(u,w)$ into vertex $w$, $(u,w)$ is an edge in $\GG_{j-1}^{\fr}$, and\label{case2-in-subsec:planarization-of-hypergraphs}
\item that we maintain some \emph{partial} neighborhoods of all vertices of color $\fr$ \emph{and} of vertices that later can be contracted to create new edges (note that all these vertices must have colors from outside $\FREEC$, since vertices from $\FREEC$ will not be contracted in future iterations).\label{case3-in-subsec:planarization-of-hypergraphs}
%\end{inparaenum}
\end{compactenum}
%\end{enumerate}
%
Case (\ref{case3-in-subsec:planarization-of-hypergraphs}) requires some additional care, so that if we contract an edge/hyperedge $\e$ in $\HQ_j(\DDCH{i})$ and if $\e$ has a vertex of color from outside $\FREEC$, then we will maintain only the neighborhood of one vertex from this edge, the first one that will be later contracted in the algorithm --- which is the lowest color vertex of $\e$.

%---------------------------------------------------------------------------------------------------------------------------------------------------------

\subsection{Properties of color-$\fr$ shadow graphs}
\label{subsubsec:shadow-graphs}

Let us begin with a characterization of the vertex sets in $\GG_j^{\fr}$ and $\GG^{\fr}(\HQ_i(\DDCH{i}))$.

\begin{claim}
\label{claim:shadow-graphs-properties-2}
For every $j$, $1 \le j \le i$, for every $u \in V$, vertex $u$ is not in $\GG_j^{\fr}$ if and only if either $u$ is an isolated vertex in $\G[\DDCH{i}]$ or $u \not\in V(\HQ_j(\DDCH{i}))$.%\Artur{Equivalently, either $u$ is an isolated vertex in $\G[\DDCH{i}]$ or $\chi(u) \in \{\chi(v_j): j < i\}$.}
%If $\GG^{\fr}(\HQ_i(\DDCH{i}))$ is the color-$\fr$ shadow graph of $\HQ_i(\DDCH{i})$, then for every $u \in V$, vertex $u$ is not in $\GG^{\fr}(\HQ_i(\DDCH{i}))$ if and only if either $u$ is an isolated vertex in $\G[\DDCH{i}]$ or $u \not\in V(\HQ_i(\DDCH{i}))$.
\end{claim}

\begin{proof}
%\small% IT'S IN small SINCE IT'S RATHER EASY
The claim easily follows from our construction. Firstly, the first step of our construction removes all isolated vertices in $\G[\DDCH{i}]$. Secondly, the only other moment when one removes vertices is when one removes every vertex $u \in V(\HQ_{j-1}(\DDCH{i}))$ with $\chi(u) = \chi(v_{j-1})$ and $2 \le j \le i$. That is, one removes all non-isolated vertices $u$ in $\G[\DDCH{i}]$ with $\chi(u) \in \{\chi(v_j): j < i\}$, which are exactly all vertices $u \not\in V(\HQ_i(\DDCH{i}))$.
\end{proof}

Our next claim describes a key property of color-$\fr$ shadow graphs. The construction of the color-$\fr$ shadow graph mimics the construction of the hypergraph $\HQ_i(\DDCH{i})$ with one key difference: while to construct $\HQ_j(\DDCH{i})$ from $\HQ_{j-1}(\DDCH{i})$ we remove every vertex $u$ of color $\chi(v_{j-1})$ from $\HQ_{j-1}(\DDCH{i})$ and add a new hyperedge ``connecting'' the neighbors of $u$ in $\HQ_{j-1}(\DDCH{i})$, in the color-$\fr$ shadow graph we perform a similar operation to define $\GG_j^{\fr}$, but instead of connecting the neighbors using a single hyperedge, we either connect them by adding edges from all neighbors to a single vertex (edge-contraction), or we do nothing.
The following lemma shows that this construction properly maintains the neighborhoods of vertices of color $\fr$ and our property (\ref{case2-in-subsec:planarization-of-hypergraphs}) above.

\begin{lemma}
\label{lemma:key-property-color-shadow-graphs}
Let $\DDCH{i}$ be a set of edge-disjoint colored copies of $H$ in $G$ and let $\HQ_i(\DDCH{i})$ be a hypergraph consistent for $\DDCH{i}$. Let $\fr$ be any color in $\FREEC$. Then,
\begin{enumerate}[(a)]
\item for any vertex $u \in V(\HQ_i(\DDCH{i}))$ of color $\fr$, if $x \in V \setminus \{u\}$ is a neighbor of $u$ in $\HQ_i(\DDCH{i})$, then $x$ is also a neighbor of $u$ in the color-$\fr$ shadow graph $\GG^{\fr}(\HQ_i(\DDCH{i}))$, and \label{prop-a}%\footnote{We could actually prove that this property is an ``if and only if'' ($x$ is a neighbor of $u$ in $\HQ_i(\DDCH{i})$ iff $x$ is a neighbor of $u$ in $\GG^{\fr}(\HQ_i(\DDCH{i}))$), but since this would make this proof a little longer, the current statement suffices.}
\item every time to define $\GG_j^{\fr}$, $j > 1$, we perform contraction of an edge $(u,w)$ into vertex $w$, we have that $(u,w)$ is an edge in $\GG_{j-1}^{\fr}$.\label{prop-b}
\end{enumerate}
\end{lemma}

\begin{proof}
We first prove by induction on $j$ the following invariant for every $j$, $1 \le j \le i$:
\begin{enumerate}[(1)]
\item if $\e$ is an edge in $\HQ_j(\DDCH{i})$ that contains at least one vertex with colors from outside $\FREEC$, then %\Artur{We could also say $\e \setminus \{w\} \subseteq \Gamma_j(w)$.}
    for every $x \in \e \setminus \{w\}$, $\GG_j^{\fr}$ contains edge $(x,w)$, where $w$ is the lowest color vertex in $\Gamma_{j-1}(u)$;\label{fir-cond}
\item if $\e$ is an edge in $\HQ_j(\DDCH{i})$ that contains only vertices with colors from $\FREEC$ and contains a vertex $u$ of color $\fr$, then %\Artur{We could also say $\e \setminus \{u\} \subseteq \Gamma_j(u)$.}
    for every $x \in \e \setminus \{u\}$, $\GG_j^{\fr}$ contains edge $(x,u)$;\label{sec-cond}
%\item if $\e$ is an edge in $\HQ_j(\DDCH{i})$ that contains only vertices with colors from $\FREEC \setminus {\fr}$, then
\item every time to define $\GG_j^{\fr}$, $j > 1$, we perform contraction of an edge $(u,w)$ into vertex $w$, we have that $(u,w)$ is an edge in $\GG_{j-1}^{\fr}$.\label{third-cond}
\end{enumerate}

(Let us remark that the reason of special treatment of the edges/hyperedges $\e$ in $\HQ_j(\DDCH{i})$ containing only vertices with colors from $\FREEC \setminus {\fr}$, is that our construction ensures that all such edges will stay unchanged in $\GG_{j+1}^{\fr}, \dots, \GG_i^{\fr}$, and hence, since they contain no vertices of color $\fr$, they are irrelevant for the set of neighbors of any vertex $u$ of color $\fr$.)

To prove the invariants, let us first note that
since $\GG_1^{\fr} = \G[\DDCH{i}] = \HQ_1(\DDCH{i})$, all invariants trivially hold for $j=1$.

Next, let us assume that $j > 1$.

First, we observe that invariant (\ref{third-cond}) for $j$ follows immediately from invariant (\ref{fir-cond}) for $j-1$. Indeed, let us consider a vertex $u \in V(\HQ_{j-1}(\DDCH{i}))$ with $\chi(u) = \chi(v_{j-1})$. Then, every neighbor $x$ of $u$ will be adjacent to $u$ via an edge/hyperedge in $\HQ_{j-1}(\DDCH{i})$ containing $u$. Since each edge has vertices of distinct colors, vertex $u$ is the lowest color vertex in $\e$. Therefore, by invariant (\ref{fir-cond}), graph $\GG_{j-1}^{\fr}$ contains edge $(x,u)$.

Next, we prove that invariants (\ref{fir-cond}) and (\ref{sec-cond}) hold for $j>1$, assuming the invariants hold for $j-1$.

First, let us notice that any edge/hyperedge $\e$ in $\HQ_{j-1}(\DDCH{i})$ that contains no vertex of color $\chi(v_{j-1})$ will also be an edge in $\HQ_{j}(\DDCH{i})$, and hence invariants (\ref{fir-cond}) and (\ref{sec-cond}) will be satisfied.

Therefore we only have to consider the case when a new hyperedge $\e$ is created in $\HQ_j(\DDCH{i})$. This edge $\e$ has been obtained by taking a vertex $z \in V(\HQ_{j-1}(\DDCH{i}))$ with $\chi(z) = \chi(v_{j-1})$, and creating $\e$ to be equal to the set of all neighbors of $z$ in $\HQ_{j-1}(\DDCH{i})$. Notice that by invariant (\ref{fir-cond}) and because $z$ is the lowest color vertex in every edge/hyperedge incident to $z$ in $\HQ_{j-1}(\DDCH{i})$, vertex $z$ is adjacent in $\GG_{j-1}^{\fr}$ to all of its neighbors in $\HQ_{j-1}(\DDCH{i})$, that is, $\e \subseteq \Gamma_{j-1}(z)$. Therefore, the operation of contracting any edge $(z,y)$ in $\GG_{j-1}^{\fr}$ to vertex $y$ in $\GG_{j}^{\fr}$ will result in vertex $y$ being adjacent in $\GG_{j}^{\fr}$ to every vertex from $\e \setminus \{y\}$. Therefore, by the way how our algorithm constructing $\GG_{j}^{\fr}$ chooses vertex $w$, we can conclude that invariants (\ref{fir-cond}) and (\ref{sec-cond}) hold for~$j$.

Now, once we have proven the invariants, we can conclude the proof of Lemma \ref{lemma:key-property-color-shadow-graphs}. Indeed, invariant (\ref{sec-cond}) for $j=i$ implies the proof of property (\ref{prop-a}) and invariant (\ref{third-cond}) yields property (\ref{prop-b}).
\end{proof}

Notice that property (\ref{prop-b}) of Lemma \ref{lemma:key-property-color-shadow-graphs} together with the well-known fact that any number of edge-contractions of a planar graph leaves the graph planar (cf. Section \ref{subsec:basic-planar}) yields the following.

\begin{claim}
\label{claim:shadow-graphs-are-planar}
If $G$ is a planar graph then so is %the color-$\fr$ trace graph of the shadow graph
$\GG_i^{\fr}$, for every $i$ and $\fr \in \FREEC$.
\end{claim}

Now we are ready to complete the proof of Lemma \ref{lemma:central-small-degrees}.

\LSG* % STATEMENT OF Lemma \ref{lemma:central-small-degrees}

\begin{proof}
Let us define a simple graph $\Gii = (\Vii,\Eii)$ that is a union of graphs $\GG^{\fr}(\HQ_i(\DDCH{i}))$ for $\fr \in \FREEC$. That is, its vertex set $\Vii$ is equal to the set of non-isolated vertices in $\HQ_i(\DDCH{i})$ (or equivalently, vertices of $\G[\DDCH{i}]$ of colors from $\FREEC$) and its edge set $\Eii$ contains all edges from all graphs $\GG^{\fr}(\HQ_i(\DDCH{i}))$ for $\fr \in \FREEC$, that is,
\begin{displaymath}
    \Eii =
        \{(x,y) \in V^2: x \ne y \text{ and there exists } \fr \in \FREEC
            \text{ such that }
            (x,y) \in E(\GG^{\fr}(\HQ_i(\DDCH{i})))
        \}
    \enspace.
\end{displaymath}

For any vertex $u \in V$, let $\NEIG{\Gii}{u}$ be the set of neighbors of vertex $u$ in $\Gii$ and let $\NEIG{\HQ_i(\DDCH{i})}{u}$ be the set of neighbors of vertex $u$ in $\HQ_i(\DDCH{i})$. Notice that by Lemma \ref{lemma:key-property-color-shadow-graphs}, for any vertex $u \in V$, every neighbor $x \in V \setminus \{u\}$ of $u$ in $\HQ_i(\DDCH{i})$ is also a neighbor of $u$ in $\Gii$. That is,
\begin{equation}
\label{eq:subset-neighbors}
    \text{ for every } u \in V \text{ it holds that }
    \NEIG{\HQ_i(\DDCH{i})}{u} \setminus \{u\}
        \subseteq \NEIG{\Gii}{u}
    \enspace.
\end{equation}

Now, we define the \emph{shadow graph} $\Gi$ to be the maximal subgraph of $\Gii$ for which (\ref{eq:subset-neighbors}) holds, that is, a subgraph of $\Gii$ for which $\NEIG{\HQ_i(\DDCH{i})}{u} \setminus \{u\} = \NEIG{\Gi}{u}$ for every $u \in V$. Clearly, since $\Gi$ is a subgraph of $\Gii$, this completes the proof of Lemma \ref{lemma:central-small-degrees}.
\end{proof}

%---------------------------------------------------------------------------------------------------------------------------------------------------------

\section{Proof of Lemma \ref{lemma:small-vertices-new}: Finding many copies of $H$ with low-degree vertices}
\label{sec:proof-lemma:small-vertices-new}

In this section we prove Lemma \ref{lemma:small-vertices-new}, which states that for any set $\DDCH{i}$ of edge-disjoint colored copies of $H$ in $G$ with consistent $\HQ_i(\DDCH{i})$, there is a set $\DCH \subseteq \DDCH{i}$ of size at least $\frac{|\DDCH{i}|}{4\Hsize+2}$ such that in the hypergraph $\HQ_i(\DCH)$, every copy of $H$ in $\DCH$ has a vertex with at most $6 \Hsize$ distinct neighbors.
The proof of Lemma \ref{lemma:small-vertices-new} follows closely the arguments from \cite{CMOS11}, though the analysis needs to be expanded to deal with the underlying hypergraphs rather than graphs, and to rely on a graph representation that is a union of several planar graphs, rather than a single simple planar graph.

\LSV* % STATEMENT OF Lemma \ref{lemma:small-vertices-new}
\junk{
\begin{restatable}{lemma}{LSV}
\label{lemma:small-vertices-new}
Let $\DDCH{i}$ be a set of edge-disjoint colored copies of $H$ in $G$ and let $\HQ_i(\DDCH{i})$ be a hypergraph consistent for $\DDCH{i}$. Then, there is a set $\DCH \subseteq \DDCH{i}$ of size at least $\frac{|\DDCH{i}|}{4\Hsize+2}$ such that in the hypergraph $\HQ_i(\DCH)$, every copy of $H$ in $\DCH$ has a vertex %in $V(\HQ_i(\DCH))$
with at most $6 \Hsize$ distinct neighbors.
%for which there is a color $\cc$ (with $\cc \in \{1, \dots, \Hsize\} \setminus \{\chi(v_j): j < i\}$) such that every vertex of color $\cc$ in $V(\HQ_i(\DDCH{i}))$ has at most $6 \Hsize$ distinct neighbors in the hypergraph $\HQ_i(\DCH)$.
\end{restatable}
}

\begin{proof}
Our proof relies on Lemma \ref{lemma:central-small-degrees}, which ensures that in order to analyze the neighbors of any vertex in $\HQ_i(\DDCH{i})$ (or its sub-hypergraph $\HQ_i(\DCH)$ with $\DCH \subseteq \DDCH{i}$) it is sufficient to consider the neighbors of that vertex in $\GG(\HQ_i(\DDCH{i}))$ (or its relevant subgraph).

Notice that the vertex set of $\GG(\HQ_i(\DDCH{i}))$ is the set of all non-isolated vertices in $\HQ_i(\DDCH{i})$. Since $\GG(\HQ_i(\DDCH{i}))$ is a union of at most $\Hsize$ simple planar graphs, by Euler's formula, in any subgraph of $\GG(\HQ_i(\DDCH{i}))$ there exists a non-isolated vertex with at most $6 \cdot \Hsize - 1$ neighbors; taking into account edges incident to $u$ in $\HQ_i(\DDCH{i})$ that contain $u$ itself, there is always a non-isolated vertex in $\HQ_i(\DDCH{i})$ with at most $6 \cdot \Hsize$ neighbors. We will rely on this property throughout the proof.\footnote{To prove it, let $\widehat{G}$ be a subgraph $\GG(\HQ_i(\DDCH{i}))$, and let $V_{\widehat{G}}$ be the set of non-isolated vertices in $\widehat{G}$. By Euler's formula (cf. Fact \ref{fact:planar_limited_edges}), $\widehat{G}$ has at most $\Hsize \cdot (3 |V_{\widehat{G}}| - 6) < 3 \cdot \Hsize \cdot |V_{\widehat{G}}|$ edges. Therefore, since $\sum_{u \in V_{\widehat{G}}}\deg(u) = 2 \cdot E(\widehat{G}) < 6 \cdot \Hsize \cdot |V_G|$, there must be a vertex in $V_{\widehat{G}}$ with at most $6\Hsize-1$ distinct neighbors.}

%\paragraph{Representation of copies of $H$ in $\GG(\HQ_i(\DDCH{i}))$.}
%
In what follows, we will consider subsets of the input set $\DDCH{i}$ of copies of $H$ and their representation in the subgraph of the shadow graph $\GG(\HQ_i(\DDCH{i}))$ (by using the properties from Lemma \ref{lemma:central-small-degrees}).
%\Artur{This definition/notion looks a little messy and so maybe the text should be polished.}
Any copy $\h$ of $H$ in $\DDCH{i}$ is represented by a subgraph of $\GG(\HQ_i(\DDCH{i}))$, such that if $\h$ corresponds to a copy of $\M_i$ in $\HQ_i(\DDCH{i})$, then for any vertex $u \in V(\HQ_i(\DDCH{i}))$ in that copy, the neighbors of $u$ in that copy are also the neighbors of $u$ in that subgraph of $\GG(\HQ_i(\DDCH{i}))$. %\Artur{Add a figure.}
This definition can be extended to the representation of any subset $\DCH^* \subseteq \DDCH{i}$ of copies of $H$: $\DCH^*$ is represented by a subgraph of $\GG(\HQ_i(\DDCH{i}))$ with the edge set that is a union of all subgraphs corresponding to all copies of $H$ in $\DCH^*$. This representation will allow to naturally define the operation of removal of some copies of $H$ from $\DDCH{i}$ in the context of the subgraphs of $\GG(\HQ_i(\DDCH{i}))$.

\medskip

Let $\GG$ be the shadow graph $\GG(\HQ_i(\DDCH{i}))$, as defined in Lemma \ref{lemma:central-small-degrees}. We find $\DCH$ in two phases.

\paragraph{Phase~1:}
Let $\DCH$ be initially set up to be the input set $\DDCH{i}$ of copies of $H$. We partition $\DCH$ into \emph{levels}, iteratively removing the copies of $H$ until $\DCH$ is empty. In the $j^{\text{th}}$ iteration, we choose an arbitrary vertex $u_j$ that belongs to at least one copy of $H$ in $\DCH$ and which has at most $6 \Hsize$ distinct neighbors in the current graph $\GG$ representing $\DCH$.
%\Artur{Is this clear? I mean, the ``current'' $\GG$ is just a subgraph of $\GG$ representing the current $\DCH$.}
(Here $\DCH$ refers to the current set $\DCH$, i.e., after the removal of the sets from the previous iterations of the repeat-loop.) Every copy $\h$ of $H$ in $\DCH$ that contains $u_j$ is removed from $\DCH$. If a copy $\h$ is removed in the $j^{\text{th}}$ iteration, then its \emph{level} $\ell(\h)$ is equal to $j$.

\paragraph{Phase~2:}
We start again with $\DCH$ being the input set $\DDCH{i}$ of copies of $H$. We iterate through the levels in decreasing order. For each level $j$, we let $A(j)$ denote the current subset of copies of $H$ in $\DCH$ at level $j$. By definition of the level, all copies of $H$ in $A(j)$ must contain vertex $u_j$. Furthermore, we define $B(j)$ to be the subset of copies of $H$ in $\DCH$ that contain $u_j$ and have a level smaller than $j$. We observe that if we remove all copies of $H$ in $B(j)$ from $\DCH$, then every copy of $H$ in $A(j)$ contains a vertex (e.g., vertex $u_j$) with at most $6 \Hsize$ distinct neighbors in $\GG$. The second phase relies on this observation, and for every $j$, we will decide whether we want to return in the final $\DCH$ all copies of $H$ in $A(j)$, in which case we will remove all copies of $H$ in $B(j)$, or not. For that, we compare the size of $A(j)$ to the size of $B(j)$. If $|A(j)| \ge \frac{1}{2\Hsize} \cdot |B(j)|$, then we keep $A(j)$ and remove $B(j)$ from $\DCH$; otherwise, we remove $A(j)$.
%Below we argue that at most half of the copies of $H$ from $\DCH$ are removed because they are in some removed set $A(j)$. Furthermore, for every $2\Hsize$ copies of $H$ that are removed because they are contained in a set $B(j)$, at least one copy remains in $\DCH$. This will yield the lemma.
By our arguments above, the set $\DCH \subseteq \DDCH{i}$ obtained at the end will consist solely of copies of $H$ that contain at least one vertex  with at most $6 \Hsize$ distinct neighbors in $\GG$. Then, we only will have to prove below that $|\DCH| \ge \frac{|\DDCH{i}|}{4\Hsize+2}$.

We will now present more detailed arguments after describing a pseudocode of the process.

%---------------------------------------------------------------------------------------------------------------------------------------------------------
\medskip
\begin{walgo}
\AL\,(set $\DDCH{i}$ of copies of $H$ and a shadow graph $\GG(\HQ_i(\DDCH{i}))$)
\small
\begin{itemize}
\item $j = 1$
\item[] \textbf{\!\!Phase 1:}
\item $\DCH = \DDCH{i}$; $\GG = \GG(\HQ_i(\DDCH{i}))$
\item Repeat until $\DCH$ is empty:
    \begin{itemize}[$\circ$]
    \item Let $u_j$ be a non-isolated vertex that has at most $6 \Hsize$ distinct neighbors in~$\GG$
    \item For all copies of $H$ \ $\h \in \DCH$ that contain $u_j$, let $\ell(\h)=j$
		\item Remove from $\DCH$ all copies of $H$ that contain $u_j$ and update $\GG$ accordingly
    \item $j=j+1$
		\end{itemize}
\item[]\textbf{\!\!Phase 2:}
\item $\DCH = \DDCH{i}$
\item Repeat until $j=1$:
     \begin{itemize}[$\circ$]
		 \item $j=j-1$
		 \item $A(j) = \{\h \in \DCH: \ell(\h) = j\}$
		 \item $B(j) = \{\h \in \DCH: \ell(\h) < j \text{ and } \h \text{ contains } u_j\}$
	   \item if $|A(j)| \ge \frac{1}{2\Hsize} \cdot |B(j)|$ then
		       $\DCH = \DCH \setminus B(j)$ else $\DCH = \DCH \setminus A(j)$
        \item Update $\GG$ accordingly
		 \end{itemize}
\item Return $\DCH$
\end{itemize}
\end{walgo}
%---------------------------------------------------------------------------------------------------------------------------------------------------------

%\Artur{In the algorithm above we should make sure that it's clear that by removing a copy $\h$ of $H$ from $\DCH$, we mean the removal of the corresponding edges from $\GG$.}
In what follows we will prove the correctness of the algorithm. We first observe that Phase~1 terminates since $\GG$ is a union of at most $\Hsize$ copies of planar graph and of self-loops, and therefore by Euler's formula, it has a non-isolated vertex with at most $6 \Hsize$ neighbors (this also holds during the execution of the algorithm since planarity is closed under edge removals).

It remains to analyze Phase 2 of the algorithm. Every copy of $H$ in $\DCH$:
%
%\begin{inparaenum}[\it (a)]
\begin{enumerate}[\it (a)]
\item is removed because it is contained in some set $A(j)$ that is removed from $\DCH$ in Phase 2, or
\item is removed because it is contained in some set $B(j)$ that is removed from $\DCH$ in Phase 2, or
\item is not removed and stays in the final set $\DCH$ (and hence, by our arguments above, it contains at least one vertex  with at most $6 \Hsize$ distinct neighbors in $\GG$, and thus in $\HQ_i(\DCH)$).
\end{enumerate}
%\end{inparaenum}

Let $\alpha$, $\beta$, $\gamma$ be the respective numbers of copies of $H$ in the original $\DDCH{i}$ (notice that $\gamma = |\DCH|$). Clearly, $|\DDCH{i}| = \alpha + \beta + \gamma$ and to prove Lemma \ref{lemma:small-vertices-new} we have to show that $\gamma \ge \frac{1}{4\Hsize+2} \cdot |\DDCH{i}|$. We proceed in two steps. We first prove in Claim \ref{claim:1a} that $\alpha \le \frac12 |\DDCH{i}|$, which implies that $\beta + \gamma \ge \frac12 |\DDCH{i}|$. Then we argue in Claim \ref{claim:1b} that $2 \Hsize \gamma \ge \beta$. This yields $(2\Hsize+1) \gamma \ge \frac12 |\DDCH{i}|$ and hence $|\DCH| = \gamma \ge \frac{1}{4\Hsize+2} \cdot |\DDCH{i}|$, completing the proof of Lemma \ref{lemma:small-vertices-new}.

\begin{claim}
\label{claim:1a}
$\alpha \le \frac12 |\DDCH{i}|$.
\end{claim}

\begin{proof}
We charge the vertices from the removed sets $A(j)$ to the sets $B(j)$ and derive a bound on the sum of sizes of the sets $B(j)$. Recall that every copy of $H$ contains $\Hsize$ vertices. In every copy, one vertex is the vertex that has at most $6 \Hsize$ distinct neighbors in $\GG$, when the copy is removed in Phase 1 of the algorithm. Thus, every copy is contained in at most $\Hsize-1$ different sets $B(j)$. It follows that
\begin{displaymath}
    \sum_{j} |B(j)| \le (\Hsize-1) \cdot |\DDCH{i}|
    \enspace.
\end{displaymath}
Let $R$ denote the set of indices $j$ such that $A(j)$ is removed from $\DCH$ during Phase 2. Observe that whenever we remove a set $A(j)$, we have $|A(j)| < \frac{1}{2\Hsize} |B(j)|$ by the condition in the process. It follows that
\begin{align*}
    \alpha
        &=
    \sum_{j \in R} |A(j)|
        <
    \sum_{j \in R} \frac{1}{2\Hsize} |B(j)|
        \le
    \frac{\Hsize-1}{2\Hsize} \cdot |\DDCH{i}|
        <
    \frac{1}{2} \cdot |\DDCH{i}|
        \enspace.
        \qedhere
\end{align*}
\end{proof}

\begin{claim}
\label{claim:1b}
$2 \Hsize \gamma \ge \beta$.
\end{claim}

\begin{proof}
For every set $B(j)$ removed from $\DCH$, we know that $|A(j)| \ge \frac{1}{2\Hsize} |B(j)|$. At the point of time when $B(j)$ is removed from $\DCH$, the set $A(j)$ remains in $\DCH$ because $A(j)$ and $B(j)$ are disjoint. Since we are iterating downwards through the levels of the copies of $H$, the set $A(j)$ is also disjoint from all sets $B(j')$, $j'<j$, and so it is not removed also in any future iteration of the repeat loop. Thus, in this case each copy of $H$ from $A(j)$ remains in $\DCH$ until the end of the process and contributes to the value of $\gamma$. Let $R'$ be the set of indices $j$ such that $A(j)$ remains in $\DCH$ during Phase 2 (and hence $B(j)$ is removed from $\DCH$). Since each copy of $H$ in $A(j)$, $j\in R'$, contributes to $\gamma$ and since sets $A(j)$ are disjoint, we obtain $\sum_{j\in R'} |A(j)| \le \gamma$. Hence,
\begin{align*}
    \frac{1}{2\Hsize} \beta
        &=
    \frac{1}{2\Hsize} \sum_{j\in R'} |B(j)|
        \le
    \sum_{j\in R'} |A(j)|
        \le
    \gamma
        \enspace,
\end{align*}
which implies the claim.
\end{proof}

With Claims \ref{claim:1a}--\ref{claim:1b} at hand, we obtain that the set $\DCH$ ($\DCH \subseteq \DDCH{i}$) contains copies of $H$ such that
\begin{itemize}
\item each copy of $H$ in $\DCH$ has a vertex with at most $6 \Hsize$ distinct neighbors in $\HQ_i(\DCH)$, and
\item $|\DCH| = \gamma \ge \frac{1}{4\Hsize+2} \cdot |\DDCH{i}|$.
\end{itemize}

This completes the proof of Lemma \ref{lemma:small-vertices-new}.
\end{proof}

%---------------------------------------------------------------------------------------------------------------------------------------------------------

\end{document}